\documentclass[twoside,11pt]{article}

\usepackage[preprint]{jmlr2e}

\usepackage{enumerate}
\usepackage{algorithm, tabularx}
\usepackage[noend]{algpseudocode}
\usepackage{amsmath,latexsym,wrapfig,enumitem,mathtools}
\allowdisplaybreaks[1]
\usepackage{thm-restate}

\usepackage{tikz}
\usetikzlibrary{decorations.pathreplacing,calc}
\newcommand{\tikzmark}[1]{\tikz[overlay,remember picture] \node (#1) {};}
\newcommand*{\AddNote}[4]{%
    \begin{tikzpicture}[overlay, remember picture]
        \draw [decoration={brace,amplitude=0.5em},decorate,thick]
            ($(#3)!(#1.north)!($(#3)-(0,1)$)$) --  
            ($(#3)!(#2.south)!($(#3)-(0,1)$)$)
                node [align=center, text width=2.5cm, pos=0.5, anchor=west] {#4};
    \end{tikzpicture}
}%

\makeatletter
\newcommand{\multiline}[1]{%
  \begin{tabularx}{\dimexpr\linewidth-\ALG@thistlm}[t]{@{}X@{}}
    #1
  \end{tabularx}
}
\makeatother

\DeclareMathOperator{\x}{\mathbf{x}}
\DeclareMathOperator{\uu}{\mathbf{u}}
\DeclareMathOperator{\vv}{\mathbf{v}}
\DeclareMathOperator{\ee}{\mathbf{e}}
\DeclareMathOperator{\ff}{\mathbf{f}}
\DeclareMathOperator{\U}{\mathbf{U}}
\DeclareMathOperator{\V}{\mathbf{V}}
\DeclareMathOperator{\Z}{\mathbf{Z}}
\DeclareMathOperator{\A}{\mathbf{A}}

\DeclareMathOperator{\M}{\mathbf{M}}
\DeclareMathOperator{\N}{\mathbf{N}}
\DeclareMathOperator{\D}{\mathbf{D}}
\DeclareMathOperator{\Q}{\mathbf{Q}}
\DeclareMathOperator{\I}{\mathbf{I}}
\DeclareMathOperator{\X}{\mathbf{X}}
\DeclareMathOperator{\Y}{\mathbf{Y}}
\DeclareMathOperator{\W}{\mathbf{W}}
\DeclareMathOperator{\R}{\mathbf{R}}
\DeclareMathOperator{\real}{\mathbb{R}}
\DeclareMathOperator{\Xhat}{\vphantom{\mathbf{X}} \smash[t]{\hat{\mathbf{X}}}}
\DeclareMathOperator{\Xbar}{\vphantom{\mathbf{X}} \smash[t]{\bar{\mathbf{X}}}}
\DeclareMathOperator{\mmu}{\boldsymbol{\mu}}

\DeclareMathOperator{\Delt}{\boldsymbol{\Delta}}
\DeclareMathOperator{\Sig}{\boldsymbol{\Sigma}}

\DeclareMathOperator*{\argmin}{\mathrm{argmin} ~ }
\DeclareMathOperator*{\argmax}{\mathrm{argmax} ~ }
\newcommand{\norm}[1]{\lVert \, #1  \, \rVert}

\jmlrheading{1}{2020}{1-67}{1/20}{}{tang20ipca}{Tiffany M. Tang and Genevera I. Allen}

\ShortHeadings{Integrated Principal Components Analysis}{Tang and Allen}
\firstpageno{1}

\begin{document}

\graphicspath{{Figures/}}

\title{Integrated Principal Components Analysis}

\author{\name Tiffany M.\ Tang \email tiffany.tang@berkeley.edu\\
       \addr Department of Statistics\\
       University of California, Berkeley\\
       Berkeley, CA 94720, USA
       \AND
       \name Genevera I.\ Allen \email gallen@rice.edu\\
       \addr Department of Electrical and Computer Engineering\\
       Rice University\\
       Houston, TX 77005, USA}


\maketitle

\begin{abstract}
Data integration, or the strategic analysis of multiple sources of data simultaneously, can often lead to discoveries that may be hidden in individualistic analyses of a single data source. We develop a new unsupervised data integration method named Integrated Principal Components Analysis (iPCA), which is a model-based generalization of PCA and serves as a practical tool to find and visualize common patterns that occur in multiple data sets. The key idea driving iPCA is the matrix-variate normal model, whose Kronecker product covariance structure captures both individual patterns within each data set and joint patterns shared by multiple data sets. Building upon this model, we develop several penalized (sparse and non-sparse) covariance estimators for iPCA, and using geodesic convexity, we prove that our non-sparse iPCA estimator converges to the global solution of a non-convex problem. We also demonstrate the practical advantages of iPCA through extensive simulations and a case study application to integrative genomics for Alzheimer's disease. In particular, we show that the joint patterns extracted via iPCA are highly predictive of a patient's cognition and Alzheimer's diagnosis.
\end{abstract}

\begin{keywords}
data integration, multi-view data, matrix-variate normal, dimension reduction, integrative genomics
\end{keywords}

\section{Introduction}

The recent growth in both data volume and variety drives the need for principled data integration methods that can analyze multiple sources of data simultaneously. For instance, meteorologists must integrate data from satellites, ground-based sensors, and numerical models for forecasting \citep{ghil1991meteorology}. Audio and video are often combined for surveillance as well as speech recognition \citep{shivappa2010audiovisual}; and as new high-throughput technologies arise in biology, scientists are leveraging information from multiple genomic sources to better understand complex biological processes \citep{huang2017more}. By exploiting the commonalities and diversity of information from different data sets, data integration methods have the potential to provide a holistic and perhaps more realistic model of the phenomena at hand. 

In this work, we focus on facilitating unsupervised learning tasks such as pattern recognition, dimension reduction, and visualization for integrated data in various applications. More specifically, we consider the multi-view data setting, where we observe multiple data sources (or matrices) with features of different types that are measured on the same set of samples. Our goal in this setting is to develop a practical statistical data integration method that 1) leverages multiple data sources to discover and visualize dominant joint patterns among the samples that are common across all data sets; 2) generalizes the classical principal components analysis (PCA), thereby inheriting its nice properties including easily-interpretable visualizations, a unique solution, and nested, orthogonal components that can be quickly obtained all at once; and 3) has provable statistical and optimization guarantees to begin bridging the gap between the theory and practice of data integration.

To this end, we propose Integrated Principal Components Analysis (iPCA), which extends a model-based framework of the classical Principal Components Analysis (PCA) to integrated data. As such, iPCA inherits the many advantages and nice familiar interpretations of PCA. The key idea here behind iPCA is to leverage a new but natural connection between data integration and the matrix-variate normal model. This enables us to flexibly model a rich set of dependencies among features and samples simultaneously, as is often crucial in integrated data. Building upon this model, we develop novel penalized covariance estimators for iPCA including a new geodesically convex penalty with both theoretical and practical advantages. While the main contributions of this work are methodological and applied in nature, we also begin to study the theoretical properties of our approach. In particular, we show that our non-sparse iPCA estimator converges to the global solution of a non-convex problem using geodesic convexity, and in the Appendix, we show that our sparse iPCA estimator consistently estimates the underlying joint subspace. Finally, through simulations and a careful case study application to integrative genomics for Alzheimer's disease, we demonstrate the superior empirical performance of iPCA for integrative exploratory data analysis, joint pattern recognition, and visualization. Specifically, in our case study, iPCA is able to identify joint patterns that are biologically-meaningful and can separate patients by their cognitive capabilities and Alzheimer's diagnosis. This interesting finding and application to Alzheimer's disease not only demonstrates the practical utility of iPCA but also points to promising hypotheses for follow-up experiments in Alzheimer's research.


\subsection{Related Work}
Currently, existing data integration methods for unsupervised learning tend to fall into one of two categories - the matrix factorizations or the multiblock PCA family - but each have major practical limitations that we address with iPCA. 

First, matrix factorizations, including coupled matrix and tensor (CMTF) factorizations \citep{singh2008collective, acar2014coupled} and Joint and Individual Variation Explained (JIVE) \citep{lock2013joint}, each solve an optimization problem to factorize the integrated data sets into a low-rank joint variation matrix, encoding the shared patterns, and a low-rank individual variation matrix, encoding the patterns specific each data set. Matrix factorizations are well-liked due to their simplicity and computational feasibility, but a practical challenge with these methods is that they do not have a unique solution and heavily depend on the ranks of the factorized matrices. That is, the top factors from CMTF and JIVE are non-nested and can change drastically depending on the user-specified rank. Since the ``optimal'' rank is almost always unknown and must be specified a priori, this poses significant interpretability challenges as the practitioner could end up with many different, but equally valid, solutions corresponding to different choices of ranks and random seeds. On a related front, the generalized SVD (GSVD) \citep{alter2003generalized, ponnapalli2011higher} provides an exact matrix decomposition for integrated data that does not depend on the matrix rank. Nevertheless, it is limited in scope since the GSVD assumes each matrix has full row rank, excluding problems with both high- and low-dimensional data sets.

On the other hand, the multiblock PCA family, which includes methods such as Multiple Factor Analysis (MFA) \citep{escofier1994mfa, abdi2013mfa} and Consensus PCA \citep{westerhuis1998multiblock}, works to generalize PCA to the integrated data regime and does not suffer the rank-dependence or interpretability issues as with matrix factorizations. In these multiblock PCA methods, each data matrix is normalized according to a specific procedure, and then PCA is performed on the normalized concatenated data. However, these normalization schemes are often ad-hoc, and it is unclear which scheme works best and for which situation. Closely related to this is Distributed PCA \citep{fan2017pca}, which integrates data that are stored across multiple servers and implicitly assumes that the data are $i.i.d.$ across the different servers. This assumption differs from our target setting, where we allow for heterogeneity among the different sources (or servers in the distributed context).


To date, an unsupervised data integration method, which both generalizes PCA and automatically determines the best way to normalize for the different scales and signal strengths between sources, does not exist. Motivated by these limitations, we develop iPCA as a proper generalization of PCA, thus inheriting its many properties and advantages (e.g., a unique solution, easily-interpretable visualizations, and nested, orthogonal principal components) unlike the matrix factorizations; and unlike the multiblock PCA methods, iPCA automatically adjusts for the different scales and signals between data sources without the need for a specified normalization scheme. The two main building blocks of iPCA are PCA and the matrix-variate normal distribution, which we review next.

\subsection{Principal Components Analysis} \label{sec:pca}
Given a column-centered data set $\X \in \mathbb{R}^{n \times p}$ with $n$ samples and $p$ features, recall that PCA finds orthogonal directions $\vv_1, \dots, \vv_m \in \mathbb{R}^{p}$, which maximize the covariance $\Delt \in \mathcal{S}^p_{++}$, where $\mathcal{S}^p_{++}$ is the set of $p \times p$ positive definite matrices. That is, for each $j = 1, \dots, m$,
\begin{align}
\vv_j = \argmax_{\vv \in \mathbb{R}^{p}} \vv^T \Delt \vv \quad \text{subject to } \vv^T \vv = 1, ~ \vv^T \vv_i = 0 ~ \forall \: i < j.  \label{pca_opt}
\end{align}
It is well-known that the PC loading $\vv_j$ is the eigenvector of $\Delt$ with the $j^{th}$ largest eigenvalue, and its corresponding PC score is $\uu_j := \X \vv_j$. In practice, since the population covariance $\Delt$ is typically unknown, the sample version of PCA plugs in an estimate $\hat{\Delt} := \frac{1}{n} \X^T \X$ for $\Delt$ in \eqref{pca_opt}. It follows that the PC loadings are the eigenvectors of $\hat{\Delt}$, and the PC scores are the scaled eigenvectors of $\hat{\Sig} := \frac{1}{p} \X \X^T$. To later establish the link between iPCA and PCA, we point out that $\hat{\Delt}$ is the maximum likelihood estimator (MLE) of $\Delt$ under the multivariate normal model $\x_1, \dots, \x_n \smash[t]{\stackrel{iid}{\sim}} N(\mathbf{0}, \Delt)$ and $\hat{\Sig}$ is the MLE of $\Sig$ under $\x'_1, \dots, \x'_p \smash[t]{\stackrel{iid}{\sim}} N(\mathbf{0}, \Sig)$, where $\x_i$ is the $i^{th}$ row of $\X$, and $\x'_j$ is the $j^{th}$ column of $\X$. Thus, there is a dual row/column interpretation of the PCA model. While this is not the only way of viewing PCA, this formulation illustrates two points which we explore further in iPCA: (1) PCA finds linear projections of the data that maximize the variance under a multivariate normal model; (2) eigenvectors correspond to the dominant (or variance-maximizing) patterns in the data.


\subsection{Matrix-variate Normal Model} \label{sec:mvn_model}
Laid out in \citet{gupta1999matrix} and \citet{dawid1981some}, the matrix-variate normal distribution is an extension of the multivariate normal distribution such that the matrix is the unit of study. Formally, we say $\X \in \mathbb{R}^{n \times p}$ follows a matrix-variate normal distribution and write $\X \sim N_{n,p}(\mathbf{M}, \Sig \otimes \Delt)$ if $\mathrm{vec}(\X^T)$ follows a multivariate normal distribution with a Kronecker product covariance structure, $\mathrm{vec}(\X^T) \sim N(\mathrm{vec}(\mathbf{M}^T), \Sig \otimes \Delt)$. Here, $\mathrm{vec}(\X) \in \mathbb{R}^{np}$ is the column vector formed by stacking the columns of $\X$ below one another. We call $\mathbf{M} \in \mathbb{R}^{n \times p}$ the mean matrix, $\Sig \in \mathcal{S}^n_{++}$ the row covariance matrix, and $\Delt \in \mathcal{S}^p_{++}$ the column covariance matrix. 

Put differently, the row covariance $\Sig$ encodes the dependencies between rows of $\X$ while the column covariance $\Delt$ encodes the dependencies among columns, i.e., $\smash[tb]{\X_{i,\cdot} \sim N(\mathbf{M}_{i,\cdot}, \Sig_{ii}\Delt)}$ and $\smash[tb]{\X_{\cdot, j} \sim N(\mathbf{M}_{\cdot, j}, \Delt_{jj} \Sig)}$. It can also be shown that if $\Sig = \I$ and $\mathbf{M} = \mathbf{0}$, we are in the familiar multivariate normal setting, $\smash[tb]{\x_1, \dots, \x_n \stackrel{iid}{\sim} N(\mathbf{0}, \Delt)}$, and if $\Delt = \I$ and $\mathbf{M} = \mathbf{0}$, then $\smash[tb]{\x'_1, \dots, \x_p' \stackrel{iid}{\sim} N(\mathbf{0}, \Sig)}$. The matrix-variate normal model, however, is far more general than the multivariate normal. While the multivariate normal can only model relationships between elements of a single row or a single column in $\X$, the matrix-variate normal can model relationships between elements from different rows and columns. With this level of flexibility, the matrix-variate normal has proven to be a versatile tool in various contexts such as graphical models \citep{yin2012model, tsiligkaridis2013convergence, zhou2014gemini}, spatio-temporal models \citep{greenewald2015robust}, and transposable models \citep{allen2010transposable}. Our work on iPCA is the first to consider the matrix-variate normal model in light of data integration.

\subsection{Outline}
Building upon the matrix-variate normal model, we introduce iPCA as a proper generalization of PCA to the integrated data regime in Section~\ref{sec:ipca}. In Section~\ref{sec:cov_est}, we discuss covariance estimation methods for iPCA and begin to study some of their theoretical properties. We then demonstrate the strong empirical performance of iPCA in Section~\ref{sec:empirical_results} through simulations and a real data application to integrative genomics for Alzheimer's disease. Finally, we conclude with a discussion of iPCA in Section~\ref{sec:discussion}.



\section{Integrated PCA} \label{sec:ipca}

At its core, iPCA, like PCA, is an unsupervised tool for exploratory data analysis, pattern recognition, and visualization. Unlike PCA however, iPCA aims to extract dominant joint patterns which are \textit{common} to multiple data sets, not necessarily the variance-maximizing patterns since they might be specific to one data set. These joint patterns are typically of considerable interest to practitioners as its common occurrence in multiple data sets may point to some foundational mechanism or structure. For instance, scientists may be more interested in uncovering the patterns (or clusters) of patients who have similar gene expression levels and miRNA expression levels than those patients with similar gene expression levels alone. 

\begin{figure}[!t]
\centering
\includegraphics[width =  1\linewidth]{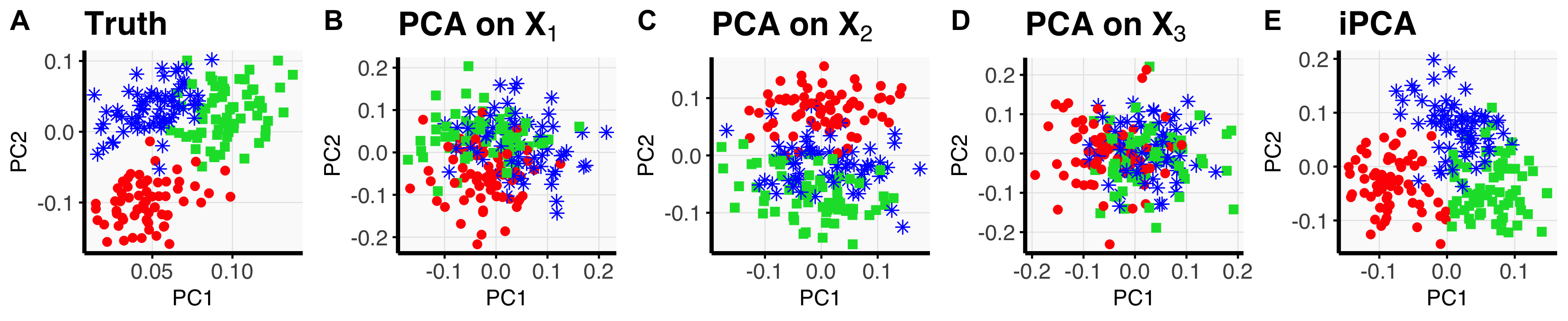}
\caption{Coupled matrices $\X_1, \X_2, \X_3$ with $n = 200$, $p_1 = 300$, $p_2 = 500$, $p_3 = 400$ were simulated from the iPCA model \eqref{pop_model}. Here, $\Sig$ and $\Delt_1, \Delt_2, \Delt_3$ were taken to be as in the base simulation described in Section~\ref{sec:sims}. (A) plots the top two eigenvectors of $\Sig$. In separate PCA analyses (B-D), the individual signal in each data set masks the true joint signal, but (E) iPCA (using the multiplicative Frobenius estimator) exploits the integrated data structure and recovers the true joint signal.}
\label{fig:illustrative}
\end{figure}

Figure~\ref{fig:illustrative} illustrates a motivating example for when iPCA is advantageous. In the example, strong dependencies among features obscure the true joint patterns among the samples so that the true joint signal is not the variance-maximizing direction. As a result, applying PCA separately to each of the data sets (panels B-D) fails to reveal the joint signal. iPCA can better recover the joint signal because it exploits the known integrated data structure and extracts the shared information among all three data sets simultaneously. 

Generally speaking, iPCA finds these joint patterns by modeling the dependencies between and within data sets via the matrix-variate normal model. The inherent Kronecker product covariance structure enables us to decompose the total covariance of each data matrix into two components---an individual column covariance structure which is unique to each data set and a joint row covariance structure which is shared among all data sets. The joint row covariance structure is our primary interest, and maximizing this joint variation will yield the dominant patterns which are common across all data sets. In the following sections, we will introduce iPCA and provide interpretations and intuition into the model.

\subsection{Population Model of iPCA} \label{sec:ipca_pop_model}

Suppose we observe $K$ coupled data matrices, $\X_{1}, \ldots, \X_{K}$, of dimensions $n \times p_{1}, \ldots, n \times p_{K}$, where $n$ is the number of samples and $p_k$ is the number of features in $\X_k$. Throughout this paper, we let $p := \sum_{k=1}^{K} p_k$ and $\tilde{\X} := [\X_1, \dots, \X_K]$. Since iPCA is primarily interested in the data variation, let us assume that each data matrix $\X_k$ has been mean centered so that each column of $\X_k$ has a mean of 0. Suppose also that each of the data matrices are measured on the same $n$ samples and that all rows of $\X_{k}$ are perfectly aligned (see Figure~\ref{fig:setting}). Under the iPCA model, we assume that each data set $\X_{k}$ arises from a matrix-variate normal distribution,
\begin{align} \label{pop_model}
\X_k \stackrel{ind.}{\sim} N_{n, p_k} (\mathbf{0}, \: \Sig \otimes \Delt_k), \qquad (k = 1, \dots, K)
\end{align}
where $\Sig$ is an $n \times n$ row covariance matrix that is jointly shared by all data matrices, and $\Delt_k$ is a $p_{k} \times p_{k}$ column covariance matrix that is specific to $\X_{k}$. We next provide additional intuition into the model parameters $\Sig, \Delt_1, \dots, \Delt_K$ and unpack the iPCA modeling assumptions.

\paragraph{Feature Dependencies $\Delt_k$}
By properties of the matrix-variate normal, we can interpret $\Delt_k$ as describing the dependence structure among features in $\X_k$, giving rise to feature patterns unique to $\X_k$. As a consequence of the iPCA model in \eqref{pop_model}, this feature (or column) dependence is revealed as
\begin{align}
[\X_k]_{i, \cdot} \sim N(\mathbf{0}, \: \Sig_{ii}\Delt_k), \qquad(k = 1, \dots, K, \: i = 1, \dots, n).
\end{align}
In other words, each sample or row in $\X_k$ has its own variance scaling factor given by $\Sig_{ii}$ while all dependencies among the features are captured by $\Delt_k$.

\paragraph{Sample Dependencies $\Sig$}
Analogously, we can interpret $\Sig$ as describing the common row dependence structure, corresponding to patterns among the samples that are shared by all $K$ data sets. According to the iPCA model in \eqref{pop_model}, the sample (or row) dependence  manifests itself as
\begin{align}
[\X_k]_{\cdot, j} \sim N(\mathbf{0}, \: [\Delt_k]_{jj}\Sig), \qquad(k = 1, \dots, K, \: j = 1, \dots, p_k).
\end{align}
Here, each feature or column in $\X_k$ gets its own variance scaling factor $[\Delt_k]_{jj}$ while all dependencies among the samples are given by $\Sig$.

\begin{figure}[!tb]
\centering
\includegraphics[width =  .6\linewidth]{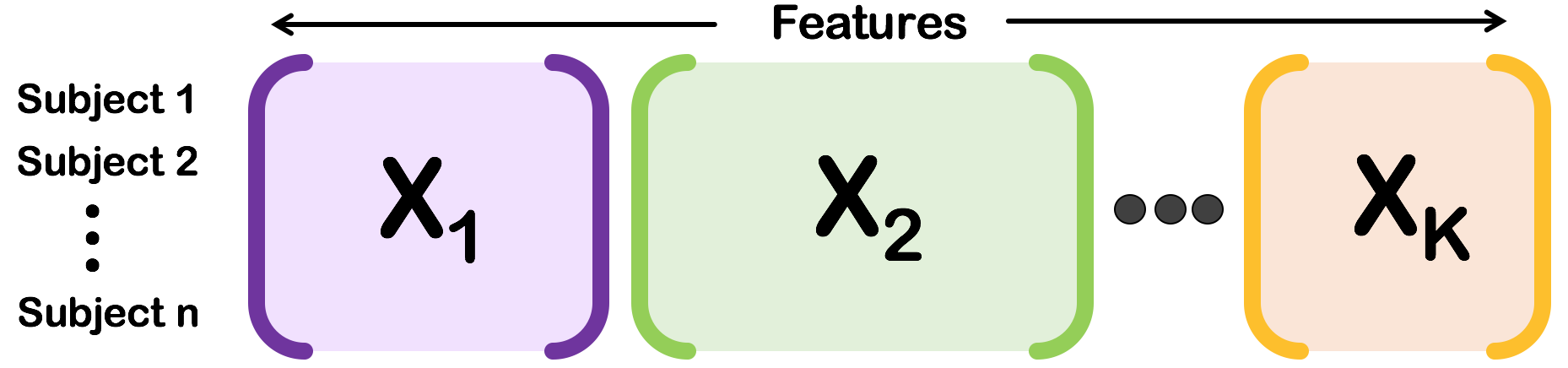}
\caption{Integrated Data Setting for iPCA: We observe $K$ different but coupled data matrices, each with a distinct set of features that are measured on the same set of $n$ samples. Assume that the rows align.}
\label{fig:setting}
\end{figure}

\paragraph{iPCA Through a Whitening Lens}
In addition to these marginal interpretations of $\Sig$ and $\Delt_k$, it is also important to gain intuition into how $\Sig$ and $\Delt_k$ interact together within the iPCA model. The simplest way to do so is through a whitening perspective, where we note that the iPCA model in \eqref{pop_model} can be equivalently rewritten as follows: for each $k = 1, \dots, K$,
\begin{align}
\Sig^{-1/2} \X_k \Delt_k^{-1/2} = \Z_k, \qquad \text{where } [\Z_k]_{ij} \stackrel{iid}{\sim} N(0, 1). \label{eq:ipca_model_equiv_z}
\end{align}
This is to say that after whitening each data matrix $\X_k$ of both its feature dependencies in $\Delt_k$ and the shared row dependencies in $\Sig$, the remainder is distributed $i.i.d.$ from a standard normal distribution. Thus, all dependencies in $\X_k$ must be captured by $\Sig$ and $\Delt_k$. This assumption may seem strong at first glance, but it is a simple generalization of the usual assumptions in PCA.

To see this, notice another equivalence relation to the iPCA model \eqref{pop_model} based on whitening: for each $k = 1, \dots, K$,
\begin{align}
\Y_k := \X_k \Delt_k^{-1/2} \stackrel{ind.}{\sim} N_{n, p_k} (\mathbf{0}, \: \Sig \otimes \mathbf{I}_{p_k}), \label{eq:ipca_model_equiv}
\end{align}
or equivalently, in vector-form,
\begin{align}
[\X_k \Delt_k^{-1/2}]_{\cdot, j} \stackrel{iid}{\sim} N(\mathbf{0}, \: \Sig) \quad \text{for each } j = 1, \dots, p_k.
\end{align}
That is, in the idealistic scenario when population covariances are known, the iPCA model simply states that each data set $\X_k$, after being whitened of its feature dependencies $\Delt_k$, follows a normal distribution with the common row covariance $\Sig$. This is not a new assumption. In fact, it is the primary modeling assumption if we were to apply classical PCA to the concatenated centered and whitened data $[\Y_{1}, \dots, \Y_{K}]$. Further, like PCA, while iPCA assumes and works best with normally-distributed data, it is still practically effective in the non-normal regime. We provide empirical evidence of this robustness in Section~\ref{sec:empirical_results}.

It is important to note, however, that when $\Delt_k \neq \mathbf{I}$ such as in most cases with real data, the underlying iPCA model greatly differs from that of applying PCA to the concatenated un-whitened data (henceforth referred to as concatenated PCA). While the normal assumption in iPCA is often insignificant in practice, it is crucial to account for $\Delt_k$ in the integrated data model. Intuitively, the feature dependencies $\Delt_k$ in iPCA can be viewed as nuisance parameters that obscure the sample dependencies $\Sig$ and vice versa, so ignoring the feature dependencies completely as in concatenated PCA can be extremely problematic. The challenge here though is that $\Delt_k$ and $\Sig$ are both unknown in practice and must be estimated. In the setting where we only observe one data set $\X_1$, it is not possible to distinguish or estimate both $\Delt_1$ and $\Sig$, but with multiple observed data sets $\X_1, \dots, \X_K$, we can and should leverage the additional data to distinguish between the feature and sample dependencies. As we will see later, this is an upshot of iPCA---namely, that iPCA models and estimates both $\Delt_k$ and $\Sig$ concurrently, exploiting the integrated structure while also accounting for the opposition's nuisance effects.

Now under these modeling assumptions from the iPCA model in \eqref{pop_model}, which we have seen to be generalizations of the classical PCA assumptions, iPCA extends the variance-maximizing ideas of PCA and achieves its objective of finding the dominant joint and individual patterns in the data by maximizing the joint row covariance $\Sig$ and individual column covariances $\Delt_1, \dots, \Delt_K$ simultaneously. Namely, for $k =1, \dots, K$, iPCA solves
\begin{align}
\uu_i &= \argmax_{\uu \in \mathbb{R}^n} \uu^T \Sig \uu \qquad \text{ subject to } \uu^T \uu = 1, ~ \uu^T \uu_l = 0 ~ \forall \: l < i, \quad (i = 1, \dots, n) \label{max_joint_ipca} \\
\vv^k_j &= \argmax_{\vv \in \mathbb{R}^{p_k}} \vv^T \Delt_k \vv \qquad \text{subject to } \vv^T \vv = 1, ~ \vv^T \vv^k_l = 0 ~ \forall \: l < j, \quad (j = 1, \dots, p_k) \label{max_individual_ipca}
\end{align}
for which we know the solution to be given by the eigendecompositions of $\Sig$ and $\Delt_k$, respectively. That is, $\uu_i$ is the eigenvector of $\Sig$ with the $i^{th}$ largest eigenvalue, and $\vv_j^k$ is the eigenvector of $\Delt_k$ with the $j^{th}$ largest eigenvalue. Most notably, $\uu_1$ maximizes the joint variation and is interpreted as the most dominant pattern among the samples, which occurs in all $K$ data sets. We call the columns of $\U := [\uu_1, \dots, \uu_n]$ the integrated principal component (iPC) scores and the columns of $\V_k := [\vv^k_1, \dots, \vv^k_{p_k}]$ the iPC loadings for the $k^{th}$ data set. Note that though the population covariances in \eqref{pop_model} are not identifiable (e.g., $\Sig \otimes \Delt_k = c\Sig \otimes \frac{1}{c} \Delt_k$ for $c \in \mathbb{R}$), the iPC scores and loadings are identifiable since eigenvectors are scale-invariant. 

Since we are often most interested in the joint patterns, we primarily plot the iPC scores $\U$ to visualize the joint patterns in sample space. To visualize the individual feature patterns from the $k^{th}$ data set, we can also plot the iPC loadings $\V_k$.

\subsection{Sample Version of iPCA} \label{sec:ipca_sample}

To perform iPCA in practice, we must typically plug in estimators $\hat{\Sig}$ and $\hat{\Delt}_k$ for $\Sig$ and $\Delt_k$ since the population covariances in \eqref{pop_model} are almost always unknown. We summarize the sample version of iPCA as follows:
\begin{enumerate}
\item \textbf{Model} each (column-centered) data set via a matrix-variate normal model: $\X_{k} \sim N_{n,p_{k}}( \mathbf{0} , \: \Sig \otimes \Delt_k)$, $\: k = 1, \dots, K$.
\item \textbf{Estimate} the covariance matrices $\Sig$ and $\Delt_1, \dots, \Delt_K$ simultaneously to obtain $\hat{\Sig}$ and $\hat{\Delt}_1, \dots, \hat{\Delt}_K$. Methods for covariance estimation will be discussed in Section~\ref{sec:cov_est}.
\item Compute the \textbf{eigenvectors}, say $\hat{\U} = \text{eigenvectors of } \hat{\Sig}$ and $\hat{\V}_k = \text{eigenvectors of } \hat{\Delt}_k$. We interpret $\hat{\U}$ as the dominant joint patterns in sample space and $\hat{\V}_k$ as the dominant patterns in feature space which are specific to $\X_k$.
\item \textbf{Visualize} and \textbf{explore} the dominant joint patterns by plotting the iPC scores $\hat{\U}$ and the dominant individual patterns by plotting the iPC loadings $\hat{\V}_k$.
\end{enumerate}

\subsubsection{Variance Explained by iPCA} \label{sec:var_explained}

After performing iPCA, we can also interpret the signal in the obtained iPCs through a notion of variance explained, analogous to that in PCA.

\begin{definition} \label{def:ipca_var_explained}
Assume that $\X_k$ has been column centered. We define the cumulative proportion of variance explained in data set $\X_k$ by the top $m$ iPCs to be
\begin{align} \label{ipca_var_explained}
\text{PVE}_{k,m} := \frac{\norm{ ( \smash[t]{\U^{(m)}} )^T \X_k \smash{\V_k^{(m)}}}_F^2}{\norm{\X_k}_F^2},
\end{align}
where $\U^{(m)} = [\uu_1, \dots, \uu_m]$ are the top $m$ iPC scores, and $\V_k^{(m)} = [\vv^k_1, \dots, \vv^k_m]$ are the top $m$ iPC loadings associated with $\X_k$.
\end{definition}

\begin{definition} \label{def:mar_ipca_var_explained}
The marginal proportion of variance explained in data set $\X_k$ by the $m^{th}$ iPC is defined as $\text{MPVE}_{k,m} := \text{PVE}_{k,m} - \text{PVE}_{k,m-1}$.
\end{definition}

We verify in Appendix~\ref{sec:s_var_explained} that $\text{PVE}_{k,m}$ is a proportion and monotonically increasing as $m$ increases. Aside from being well-defined, we also show in Appendix~\ref{sec:s_var_explained} that \eqref{ipca_var_explained} generalizes the cumulative proportion of variance explained in PCA and hence, is a natural definition. Note however that unlike PCA, it may be that $\text{MPVE}_{k,m+1} > \text{MPVE}_{k,m}$ in iPCA. This is because iPCA does not maximize the \textit{total} variance. For example, if $\text{MPVE}_{1,2} > \text{MPVE}_{1,1}$, this simply means that data set $\X_1$ contributed more variation to the joint pattern in iPC2 than in iPC1.


\subsection{Connections to Existing Methods} \label{sec:ipca_relation}

Throughout our development of iPCA, we have established several connections between iPCA and PCA, which demonstrate that iPCA is indeed a natural extension of PCA. We also find it instructive to draw on connections between iPCA and other existing data integration methods to develop an even deeper understanding of iPCA.

\subsubsection{Relationship to Multiblock PCA Family}
As discussed in \citet{abdi2013mfa}, multiblock PCA methods reduce to performing PCA on the normalized concatenated data $\tilde{\X}' = [\X_1', \dots, \X_K']$, where each $\X_k$ has been normalized to $\X_k'$ according to some procedure. We will show later in Proposition~\ref{mle_naive} that performing PCA on the unnormalized concatenated data (i.e., concatenated PCA) is a special case of iPCA, where we assume $\Delt_k = \I$ for each $k$. Proposition~\ref{mle_naive} can also be easily extended to show that multiblock PCA methods are a special case of iPCA for some fixed $\Delt_1, \dots, \Delt_K$, and the exact form of $\Delt_k$ depends on the normalization procedure. For example, since MFA normalizes $\X_k$ by dividing all of its entries by its largest singular value $\sigma_{\max, k}$, MFA is a special case of iPCA, where each $\Delt_k = \sigma_{\max, k} \I$. Put differently, MFA assumes $[\X_k \sigma_{\max, k}^{-1/2}]_{\cdot, j} \stackrel{iid}{\sim} N(\mathbf{0}, \Sig)$ whereas iPCA assumes $[\X_k \Delt_k^{-1/2}]_{\cdot, j} \stackrel{iid}{\sim} N(\mathbf{0}, \Sig)$. 

This gives rise to another interpretation of iPCA: iPCA is a generalization and unifying framework for the entire multiblock PCA family. However, while the multiblock PCA methods assume that $\Delt_k$ takes a specific form, iPCA does not impose any restrictions on the form of $\Delt_k$ and instead freely estimates the full $\Delt_k$ matrix simultaneously with $\Sig$. In doing so, iPCA acts as an automatic way of normalizing for the different scales and signals between data sources. Without appropriate normalization, the estimated principal components from multi-block PCA methods such as concatenated PCA will be biased as quantified by the following proposition.

\begin{restatable}{proposition}{bias} \label{bias_prop}
Suppose that $\X_k \sim N_{n, p_k}(\mathbf{0}, \Sig \otimes \Delt_k)$ for each $k = 1, \dots, K$, and define $\X = [\X_1, \dots, \X_K]$, $\Delt = \text{diag}(\Delt_1, \dots, \Delt_K)$, and $\Z = [\Z_1, \dots, \Z_K]$, where $\Z_k$ is as defined in \eqref{eq:ipca_model_equiv_z} previously. If $\X \Delt^{-1/2} = \mathbf{U} \mathbf{D} \mathbf{V}^T$ and $\X = \tilde{\mathbf{U}} \tilde{\mathbf{D}} \tilde{\mathbf{V}}^T$ are their respective compact SVDs, then $\mathbf{U}$ and $\tilde{\mathbf{U}}$ denote the iPCA and concatenated PCA scores, respectively, and
\begin{align}
\mathbf{U} - \tilde{\mathbf{U}} = \Sig^{1/2} \Z (\mathbf{V} \mathbf{D}^{-1} - \Delt^{1/2} \tilde{\mathbf{V}} \tilde{\mathbf{D}}^{-1}).
\end{align} 
Moreover, the iPCA and concatenated PCA scores are equal if and only if $\mathbf{V} \mathbf{D}^{-1} - \Delt^{1/2} \tilde{\mathbf{V}} \tilde{\mathbf{D}}^{-1}$ is in the nullspace of $\boldsymbol{\Z}$ or $\Delt^{1/2} \tilde{\mathbf{V}} = \mathbf{V} \mathbf{D}^{-1} \tilde{\mathbf{D}}$.
\end{restatable}

Proposition~\ref{bias_prop} implies that if the feature dependencies, captured by $\Delt_k$, rotate the data $\X_k$ in a way such that it does not obscure the joint row patterns, then concatenated PCA will work adequately. However, for most situations and examples of $\Delt_k$ that we expect to find in real data, this is unlikely to occur, and the resulting eigenvectors of concatenated PCA will be severely biased. To avoid this bias, both the sample and feature dependencies must be accounted for in the estimation procedure as done in iPCA.

\subsubsection{Relationship to Matrix Factorizations}

Coupled matrix factorizations (CMF) decompose each data set $\X_k \in \real^{n \times p_k}$ into the product of low-rank joint factor $\U \in \real^{n \times m}$ and a low-rank individual factor $\V_k \in \real^{m \times p_k}$ so that $\X_k \approx \U \V_k$ \citep{singh2008collective, acar2014coupled}. This approximate factorization is related to iPCA in that our matrix-variate normal model assumes a similar multiplicative structure. Specifically, from \eqref{eq:ipca_model_equiv_z}, the iPCA model can be viewed as $\X_k = \Sig^{1/2} \Z_k \Delt_k^{1/2}$, where $\Z_k$ is standard normal random matrix. Moreover, an argument similar to Theorem 2 from \citet{hastie2015matrix} shows that one solution of the CMF optimization problem (with $\ell_2$ penalties) is the solution of concatenated PCA and thereby a special case of iPCA. 

Despite this connection however, there is a fundamental difference between CMF and iPCA. On the one hand, CMF assumes $\X_k$ can be approximated by a low-rank matrix, and the estimation of the CMF factors actively depends on the pre-specified rank $m$. On the other hand, the rank of $\X_k$ plays absolutely no role in the iPCA assumptions or the estimation step of iPCA. Consequently, the joint and individual CMF factors can change drastically depending on the pre-specified rank unlike in iPCA. CMF also does not have a unique solution nor enforces orthogonal components whereas iPCA gives unique, nested, orthogonal components that can be interpreted in the same way as in PCA.

In contrast to the multiplicative models of CMF and iPCA, JIVE, which has been commonly used in integrative genomics, assumes an additive model and decomposes coupled data into the sum of a low-rank joint variation matrix, a low-rank individual variation matrix, and an error matrix \citep{lock2013joint}. Additive and multiplicative models, being quite different models, are each advantageous in different situations, but as with CMF, the estimation of JIVE depends on the pre-specified ranks of its factors and results in non-nested, rank-dependent joint and individual components.


\section{Covariance Estimators for iPCA} \label{sec:cov_est}

We next return to address the covariance estimation step when fitting the iPCA model to data. In Section~\ref{sec:unpenalized}, we consider the traditional maximum likelihood approach but find that it suffers from substantial limitations in the integrated data setting. These limitations ultimately drive the need for new estimators, which we develop in Section~\ref{sec:penalized}.




\subsection{Unpenalized Maximum Likelihood Estimators} \label{sec:unpenalized}

Guided by the formulation of PCA in Section~\ref{sec:pca}, we instinctively try to estimate the iPCA population covariances $\Sig$ and $\Delt_1, \dots, \Delt_K$ via maximum likelihood estimation. Under the iPCA population model \eqref{pop_model}, the log-likelihood function reduces to
\begin{align} \label{ff_loglike}
\ell(\Sig^{-1}, \Delt^{-1}_1, \dots, \Delt^{-1}_K) &\propto p \log |\Sig^{-1}| + n \sum_{k=1}^{K} \log |\Delt^{-1}_k| - \sum_{k = 1}^{K} \mathrm{tr} \left( \Sig^{-1} \X_{k}\Delt^{-1}_k \X_{k}^{T} \right),
\end{align}
so by taking partial derivatives of \eqref{ff_loglike} with respect to each parameter, we obtain
\begin{restatable}{lemma}{unpenalizedmle}\label{unpenalized_mle}
The unpenalized MLEs of $\Sig$ and $\Delt_1, \dots, \Delt_K$ satisfy
\begin{align}
\hat{\Sig} &= \frac{1}{p} \sum_{k=1}^{K} \X_{k} \smash[t]{\hat{\Delt}}^{-1}_k \X_{k}^{T} \label{unpenalized_sig}\\
\hat{\Delt}_k &= \frac{1}{n} \X_{k}^{T} \smash[t]{\hat{\Sig}}^{-1} \X_{k} \quad \forall \: k = 1, \dots, K \label{unpenalized_delt}.
\end{align}
\end{restatable}


However, with only one matrix observation per matrix-variate normal model in the iPCA context, existence of the MLE is not guaranteed. In fact, the following theorem essentially implies that the MLE does not exist for all practical purposes.

\begin{restatable}{theorem}{nonexistence} \label{nonexistence_thm}
\begin{enumerate}[label=(\roman*)]
\item Suppose that $\X_k$ has been column-centered so that each column has a mean of 0 for each $k = 1, \dots, K$. Then the unpenalized MLEs for $\Sig$ and $\Delt_1, \ldots, \Delt_K$ are not positive definite and hence do not exist. 
\item Suppose that $\X_k$ has not been column-centered but that $\mathrm{rank}(\tilde{\X}) = n$ and $\mathrm{rank}(\X_k) = p_k$ for $k = 1, \ldots , K$. 
\begin{enumerate}
\item If $n \neq p_k$ for some $k = 1, \dots, K$, then the unpenalized log-likelihood function $\ell(\Sig^{-1}, \Delt_1^{-1}, \dots, \Delt_K^{-1})$ is unbounded.
\item If the unpenalized log-likelihood function $\ell(\Sig^{-1}, \Delt_1^{-1}, \dots, \Delt_K^{-1})$ is bounded, then the unpenalized MLE for $\Sig$, $\Delt_1, \ldots, \Delt_K$ exist.
\end{enumerate}
\end{enumerate} 
\end{restatable}

The proof of Theorem~\ref{nonexistence_thm} also shows that if the unpenalized MLEs for $\Sig$ and $\Delt_{1}, \dots, \Delt_K$ exist, then $p_k = n \leq p$ for each $k = 1, \dots, K$. Thus in summary, if $p_k \neq n$ for some $k$, $n > p$, or the data matrices have been column-centered to have mean 0, then the unpenalized MLEs do not exist. These severe restrictions motivate new covariance estimators.

For example, one alternative but naive approach is to estimate $\Sig$ and $\Delt_k$ by setting their counterparts to $\I$.

\begin{restatable}{proposition}{mlenaive} \label{mle_naive}
\begin{enumerate}[label=(\roman*)]
\item The MLE for $\Delt_k$, assuming that $\Sig = \I$,  is
\begin{align}
\hat{\Delt}_k = \frac{1}{n} \X_{k}^{T} \X_{k}. 
\end{align}
\item The MLE for $\Sig$, assuming $\Delt_k = \I$ for all $k = 1, \dots, K$, is 
\begin{align}
\hat{\Sig} = \frac{1}{p} \sum_{k=1}^{K} \X_{k} \X_{k}^{T} = \frac{1}{p} \tilde{\X} \tilde{\X}^{T}. \label{naive_sig}
\end{align}
\end{enumerate}
\end{restatable}
 
This approach for estimating $\Sig$ is the familiar MLE for the concatenated data $\tilde{\X}$, and hence, performing concatenated PCA is equivalent to a special case of iPCA, where we set $\Delt_k = \mathbf{I}$ for each $k$. While this formalizes the connection between PCA and iPCA, we will see in Section~\ref{sec:empirical_results} that concatenated PCA performs poorly when the data sets are of different scales or when the feature dependencies are stronger than the sample dependencies. In the next section, we discuss more effective methods for estimating the iPCA covariances.

\subsection{Penalized Maximum Likelihood Estimators} \label{sec:penalized}

Given that the unpenalized MLEs do not exist for a large class of problems, one possible solution is to develop penalized MLEs, which solve
\begin{align}
\smash[t]{\hat{\Sig}}^{-1}, \smash[t]{\hat{\Delt}}^{-1}_1, \dots, \smash[t]{\hat{\Delt}}^{-1}_K &=  \smash{\argmax_{\substack{\Sig^{-1} \succ 0 \\ \Delt^{-1}_1, \dots, \Delt^{-1}_K \succ 0}}} \Big \{  p \log | \Sig^{-1} | + n \sum_{k=1}^{K} \log | \Delt^{-1}_k | - \sum_{k = 1}^{K} \mathrm{tr}\left( \Sig^{-1} \X_{k} \Delt^{-1}_k \X_{k}^{T} \right) \nonumber \\
& \qquad \qquad \qquad \qquad \qquad - P(\Sig^{-1}, \Delt^{-1}_1, \dots, \Delt^{-1}_K) \Big \}. \label{penalized_mle}
\end{align}

Similar to previous work on the penalized matrix-variate normal log-likelihood \citep{yin2012model, allen2010transposable}, we can apply an additive-type penalty and define the additive $L_q$ iPCA penalty to be 
\begin{align*}
\smash[tb]{P_q(\Sig^{-1}, \Delt^{-1}_1, \dots, \Delt^{-1}_K) = \lambda_{\Sig} \norm{\Sig^{-1}}_q + \sum_{k=1}^{K} \lambda_k \norm{\Delt^{-1}_k}_q},
\end{align*}
where $\lambda_{\Sig}, \lambda_1, \dots, \lambda_K > 0$ are tuning parameters. Though there are many potential choices of norm-penalties here, one natural choice is the additive Frobenius penalty, $\norm{\cdot}_q = \norm{\cdot}_{F}^2$, as it is a proper generalization of PCA. That is, performing iPCA with the additive Frobenius penalty is equivalent to PCA in the $K = 1$ case (see Theorem 1 in \citealt{allen2010transposable}). When $K \geq 1$, the additive Frobenius penalty induces a smoothness over the eigenvalues of the covariance matrices and returns a dense full-rank covariance estimator. In the sparse covariance setting, we can instead induce sparsity through the additive $L_1$ penalty $\norm{\cdot}_q = \norm{\cdot}_{1,\text{off}}$, where $\norm{\mathbf{A}}_{1,\text{off}} = \sum_{i \neq j} A_{ij}$. Applying the additive $L_1$ penalty to the inverse covariance matrix is common practice, but one can alternatively apply the additive $L_1$ penalty to the inverse correlation matrix. This latter approach was adopted in \cite{zhou2014gemini} and \cite{rothman2008sparse}. Given the plethora and well-studied nature of sparse graphical model estimation, we leverage existing ideas and tools to show that this approach, namely, the additive $L_1$ penalty applied to the inverse correlation matrix, consistently estimates the true underlying joint subspace under certain conditions at a rate of $O \left( \sum_{k=1}^{K}\frac{p_k}{p} \sqrt{\frac{\max\{s_{\Sig}, 1\} \log(\max\{n, p_k\})}{p_k}} \right)$, where $s_{\Sig}$ is the number of non-zero off-diagonal entries in $\Sig^{-1}$. However, due to the superior empirical performance of the iPCA Frobenius estimators over the sparse iPCA estimators (see Section~\ref{sec:empirical_results}), we leave the precise theorem statement and discussion of the $L_1$ subspace consistency result to Appendix~\ref{sec:s_consistency}.

Despite the popularity of additive-type penalties in the literature, an overarching downside with these existing penalties in the integrated data regime is that solving \eqref{penalized_mle} with additive penalties is a non-convex problem, for which we can only guarantee convergence to a local solution. Nonetheless, though \eqref{penalized_mle} is non-convex in Euclidean space, \citet{wiesel2012geodesic} showed that the matrix-variate normal log-likelihood is geodesically convex (g-convex) with respect to the manifold of positive definite matrices. G-convexity is a generalized notion of convexity on a Riemannian manifold, and like convexity, all local minima of g-convex functions are globally optimal. Exploiting this idea of g-convexity, we propose a novel type of penalty, named the multiplicative Frobenius iPCA penalty
\begin{align*}
\smash[tb]{P^*(\Sig^{-1}, \Delt^{-1})  = \sum_{k=1}^{K} \lambda_k \norm{\Sig^{-1} \otimes \Delt^{-1}_k}_F^2,}
\end{align*}
which we will show to be g-convex in Theorem~\ref{multfrob_gconvex}. Note that since $\norm{\mathbf{A} \otimes \mathbf{B}}_F^2 = \norm{\mathbf{A}}_F^2 \norm{\mathbf{B}}_F^2$, the multiplicative penalty can be rewritten as a product $\norm{\Sig^{-1}}_F^2 \sum_{k=1}^{K} \lambda_k \norm{\Delt^{-1}_k}_F^2$, giving rise to its name.

Like the additive Frobenius iPCA estimator, the multiplicative Frobenius iPCA estimator is a shrinkage technique that returns a dense covariance estimate with smoothed eigenvalues when $K \geq 1$, and when $K = 1$, it is equivalent to PCA (see Appendix~\ref{sec:s_equiv}).

Having introduced several different types of penalized iPCA covariance estimators, namely, the additive Frobenius estimator, multiplicative Frobenius estimator, additive $L_1$ covariance estimator, and additive $L_1$ correlation estimator, the question for practitioners becomes how to compute these estimators, how to select the penalty parameters, and which estimator to use in which situation. We discuss each in turn next.

\subsubsection{Flip-flop Algorithms for iPCA Estimators}

For each of the aforementioned penalties, we can compute the corresponding penalized MLEs via Flip-Flop algorithms (also known as block coordinate descent algorithms), which iteratively optimize over each of the parameters, one at a time, while keeping all other parameters fixed. These algorithms are derived fully in Appendix~\ref{sec:s_penalized_mle}, but in general, for the Frobenius penalties, each Flip-Flop update has a closed form solution determined by a full eigendecomposition. For the $L_1$ penalties (also known as the Kronecker Graphical Lasso, \citealp{tsiligkaridis2013convergence}), each update can be solved by the graphical lasso \citep{hsieh2011quic}. We provide the multiplicative Frobenius Flip-Flop algorithm here in Algorithm~\ref{alg:multfrob}, and as the other algorithms take similar forms, we leave them to Appendix~\ref{sec:s_penalized_mle}. The following theorem guarantees numerical convergence of the Flip-Flop algorithms to a local solution for the multiplicative Frobenius, additive Frobenius, and additive $L_1$ penalties.

\begin{algorithm}[!t]
\caption{Flip-Flop Algorithm for Multiplicative Frobenius iPCA Estimators}\label{alg:multfrob}
\begin{algorithmic} [1]
\State Center the columns of $\X_1, \dots, \X_K$, and initialize $\hat{\Sig}$, $\hat{\Delt}_1, \dots, \hat{\Delt}_K$ to be positive definite.
\While{not converged}
\State Take eigendecomposition: $\sum_{k=1}^{K} \X_{k} \smash[t]{\hat{\Delt}}^{-1}_k \X_{k}^{T} = \U \boldsymbol{\Gamma} \U^T$~ \tikzmark{top1}
\State Regularize eigenvalues: $\Phi_{ii} = \frac{1}{2p} \left( \Gamma_{ii} + \sqrt{\Gamma_{ii}^{2} + 8 p \sum_{k=1}^{K} \lambda_{k} \norm{\smash[t]{\hat{\Delt}}^{-1}_k}_{F}^{2} } \right)$~\tikzmark{right1}
\State Update $\smash[t]{\hat{\Sig}}^{-1} = \U \boldsymbol{\Phi}^{-1} \U^{T}$ \tikzmark{bottom1}

\For{$k = 1, \dots, K$}
\State Take eigendecomposition: $\X_{k}^{T}  \smash[t]{\hat{\Sig}}^{-1} \X_{k} = \V \boldsymbol{\Phi} \V^{T}$\tikzmark{top2}
\State Regularize eigenvalues: $\Gamma_{ii} = \frac{1}{2n} \left(\Phi_{ii} + \sqrt{\Phi_{ii}^{2} + 8 n  \lambda_{k} \norm{\smash[t]{\hat{\Sig}}^{-1}}_{F}^{2}} \right)$.
\State Update $\smash[t]{\hat{\Delt}}^{-1}_k = \V \boldsymbol{\Gamma}^{-1} \V^{T}$ \tikzmark{bottom2} 
\EndFor
\EndWhile
\vspace{-16pt}
\end{algorithmic}
\AddNote{top1}{bottom1}{right1}{Update $\Sig$}
\AddNote{top2}{bottom2}{right1}{Update $\Delt_k$}
\end{algorithm}

\begin{restatable}{theorem}{localconvergence}\label{converge_to_stationary}
Suppose that the objective function in \eqref{penalized_mle} is bounded below. Suppose also that either (i) $P(\Sig^{-1}, \Delt^{-1}_1, \dots, \Delt^{-1}_K)$ is a differentiable convex function with respect to each coordinate or (ii) $P(\Sig^{-1}, \Delt^{-1}_1, \dots, \Delt^{-1}_K) = P_{0}(\Sig^{-1}) + \sum_{k=1}^{K} P_k(\Delt^{-1}_k)$, where $P_i$ is a (non-differentiable) convex function for each $k = 1, \dots, K$. If either (i) or (ii) holds, then the Flip-Flop algorithm corresponding to \eqref{penalized_mle} converges to a stationary point of the objective.
\end{restatable}

However, building upon \citet{wiesel2012geodesic} and the notion of g-convexity, we can prove a far stronger result for the multiplicative Frobenius iPCA estimator.

\begin{restatable}{theorem}{globalconvergence}\label{multfrob_gconvex}
The multiplicative Frobenius iPCA estimator is jointly geodesically convex in $\Sig^{-1}$ and $\Delt^{-1}_1, \ldots, \Delt^{-1}_K$. Because of this, the Flip-Flop algorithm for the multiplicative Frobenius iPCA estimator given in Algorithm~\ref{alg:multfrob} converges to the global solution.
\end{restatable}

There are currently only a handful of non-convex problems where there exists an achievable global solution, so this guarantee that the multiplicative Frobenius iPCA estimator always reaches a global solution is both extremely rare and highly desirable. In Section~\ref{sec:empirical_results}, we will also see that the multiplicative Frobenius iPCA estimator undoubtedly gives the best empirical performance, indicating that in addition to its optimization-theoretic advantages from global convergence, there are significant practical advantages associated with the g-convex penalty. A self-contained review of g-convexity and the proof of Theorem~\ref{multfrob_gconvex} are given in Appendix~\ref{sec:s_geodesic}.

\subsubsection{Tuning Penalty Parameters}\label{sec:penalty_params}

To select penalty parameters for \eqref{penalized_mle} in a data-driven manner. We propose to do this via a cross-validation-like framework. We note also that the following framework can also be used to perform iPCA in missing data scenarios.

Let $\Lambda$ denote the space of penalty parameters, and let $\lambda := (\lambda_{\Sig}, \lambda_{1}, \dots, \lambda_K)$ be a specific choice of penalty parameters in $\Lambda$. The idea is to first randomly leave out scattered elements from each $\X_k$. Then, for each $\lambda \in \Lambda$, impute the missing elements via an EM-like algorithm, similar to \citet{allen2010transposable}. Finally, select the $\lambda$ which minimizes the error between the imputed values and the observed values.  

Searching over all combinations of penalty parameters in $\Lambda$ however can be computationally intractable if $K$ or the data sets themselves are large. In these cases, we can select the penalty parameters in a greedy manner: first fix $\lambda_1, \dots, \lambda_K$ and optimize over $\lambda_{\Sig}$, then fix $\lambda_{\Sig}, \lambda_2, \dots, \lambda_K$ and optimize $\lambda_{1}$, and so forth, and we stop after optimizing $\lambda_K$. 

Note also that it can be substantially easier and faster to tune the multiplicative Frobenius penalty, which has $K$ penalty parameters, compared to the additive iPCA penalties with $K + 1$ parameters. Because the choice of penalty parameter can significantly impact the empirical performance of iPCA, having one less parameter to tune is an extremely important practical advantage, and we attribute part of the strong empirical performance of the multiplicative Frobenius iPCA estimator, displayed in Section~\ref{sec:empirical_results}, to this advantage. 

Details, technical derivations, and numerical results regarding our imputation method and penalty parameter selection are provided in Appendix~\ref{sec:s_selecting_penalty_parameters}.

\subsection{Choosing the Type of Penalized iPCA Estimator}

As a result of its global convergence guarantee and the reduced complexity in tuning fewer penalty parameters, we strongly recommend using the multiplicative Frobenius estimator in practice. Though we have yet to prove statistical guarantees for the multiplicative Frobenius estimator, the strong empirical performance of the multiplicative Frobenius estimator, seen next in Section~\ref{sec:empirical_results}, firmly supports this recommendation. Even in the sparse setting (see Figure~\ref{fig:sparse}), the multiplicative Frobenius estimator performs only slightly worse than the additive $L_1$ iPCA estimators, for which we have proved subspace consistency guarantees (see Appendix~\ref{sec:s_consistency}). This empirically demonstrates the robustness and applicability of the multiplicative Frobenius estimator to a diverse array of problems.

\section{Empirical Results} \label{sec:empirical_results}

In the following simulations and case study, we evaluate iPCA against individual PCAs on each of the data sets $\X_k$, concatenated PCA, distributed PCA, JIVE, and MFA. Note that many data integration methods from the multiblock PCA family are known to perform similarly to MFA \citep{abdi2013mfa}, so we only include MFA to minimize redundancy. We also omit CMF as it performs similarly to concatenated PCA and the GSVD since it is not applicable for integrated data with both low-dimensional and high-dimensional data sets.

Our focus here will be on the non-sparse setting while we leave the sparse simulations to Appendix~\ref{sec:s_sims}. Within the class of iPCA estimators, we thus concentrate our attention on the additive and multiplicative Frobenius iPCA estimators in these dense simulations, but to also represent the sparse estimators, we include the most commonly used sparse estimator, the additive $L_1$ penalty ($\norm{\cdot}_{1,\text{off}}$) applied to the inverse covariance matrices. The stopping rule in the Flip-Flop algorithms for the additive and multiplicative Frobenius iPCA estimators is given by $\bar{\boldsymbol{\lambda}}^{1/2} ||\hat{\Sig}_t^{-1} - \hat{\Sig}_{t-1}^{-1}||_F / ||\hat{\Sig}_{t-1}^{-1}||_F < 10^{-6}$, where $\bar{\boldsymbol{\lambda}}$ denotes the mean of the penalty parameters and $\hat{\Sig}_t$ denotes the estimate of $\Sig$ in the $t^{th}$ iteration. Due to computational constraints, we stop the $L_1$ Flip-Flop algorithm after one iteration, and we select the iPCA penalty parameters in a greedy manner, as discussed in Section~\ref{sec:penalty_params}.

\subsection{Simulations} \label{sec:sims}

The base simulation is set up as follows: Three coupled data matrices, $\X_{1}, \X_{2}, \X_{3}$, with $n = 150$, $p_{1} = 300$, $p_{2} = 500, p_{3} = 400$, were simulated according to the iPCA Kronecker covariance model \eqref{pop_model}. Here, $\Sig$ is taken to be a full-rank spiked covariance matrix, where the top two eigenvalues are simulated to be much larger than the rest. These top two factors in $\Sig$ form the three clusters as shown in Figure~\ref{fig:illustrative}A. $\Delt_{1}$ is an autoregressive Toeplitz matrix with entry $(i,j)$ given by $.9^{|i-j|}$; $\Delt_{2}$ is the observed covariance matrix of miRNA data from TCGA ovarian cancer \citep{cancer2011integrated}; and $\Delt_{3}$ is a block-diagonal matrix with five equally-sized blocks. We also ensured that the largest eigenvalue of each $\Delt_k$ was larger than that of $\Sig$ so that the joint patterns are intentionally obscured by individualistic patterns. From this base simulation, we systematically varied the parameters---number of samples, number of features, and strength of the joint signal in $\Sig$ (i.e. $\norm{\Sig}_2$)---one at a time while keeping everything else constant.

We evaluate the performance of various methods using the subspace recovery error: If the true underlying subspace of $\Sig$ was simulated to be of dimension $d$ with the orthogonal eigenbasis $\uu_1, \dots, \uu_d$ and the top $d$ eigenvectors of the estimate $\hat{\Sig}$ are given by $\hat{\uu}_1, \dots, \hat{\uu}_d$, then the subspace recovery error is defined to be $\frac{1}{d}\norm{\smash{\hat{\U} \hat{\U}^T - \U \U^T}}_F^2$, where $\U = [\uu_1, \dots, \uu_d]$ and $\hat{\U} = [\hat{\uu}_1, \dots, \hat{\uu}_d]$. This metric simply quantifies the distance between the true subspace of $\Sig$ and the estimated subspace from $\hat{\Sig}$. We note that a lower subspace recovery error implies higher estimation accuracy, and in the base simulation, the true subspace of $\Sig$ is given by the number of spikes so that $d = 2$. Although there are other metrics like canonical angles, which also quantify the distance between subspaces, these metrics behave similarly to the subspace recovery error and are omitted for brevity.

\begin{figure}
\centering
\includegraphics[width =  1\linewidth]{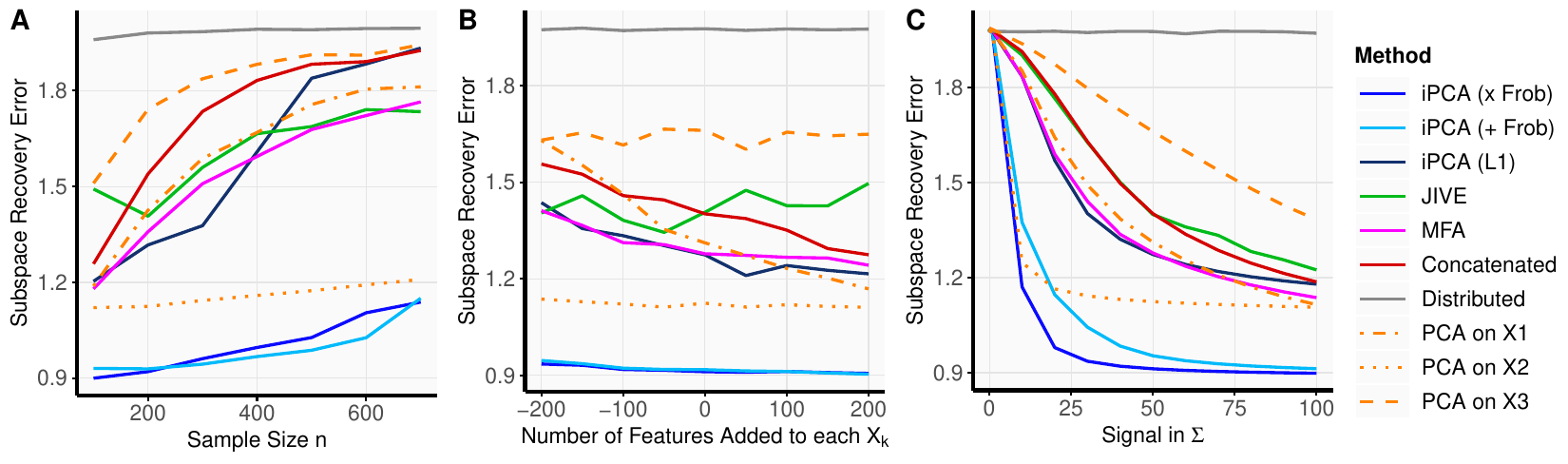}
\caption{\em \footnotesize Subspace recovery as simulation pararmeters vary from the base simulation: (A) As the number of samples increases, it becomes more difficult to estimate the joint row subspace; (B) As the number of features increases, it becomes slightly easier to estimate the joint row subspace; (C) Performance drastically improves as  the strength of the joint signal in $\Sig$ (i.e. the top singular value of $\Sig$) increases. Moreover, in almost every scenario, the multiplicative and additive Frobenius iPCA estimators outperform their competitors.}
\label{fig:sims}
\end{figure}

The average subspace recovery error, measured over $50$ trials, from various simulations are shown in Figure~\ref{fig:sims}. We clearly see that the additive and multiplicative Frobenius iPCA estimators consistently outperformed all other methods. Since $\Sig$ was not simulated to be sparse, it is no surprise that the Frobenius iPCA estimators outperformed the $L_1$ iPCA estimator. It is also expected that distributed PCA performs poorly since the $\Delt_k$'s are not all identical, violating its basic assumption. What may be surprising is that doing PCA on $\X_2$ performed better than its competitors, excluding the Frobenius iPCA estimators. We speculate that this is because the observed covariance $\Delt_2$ happened to be a very low-rank matrix, and because $\Delt_2$ was low-rank, the signal from $\Sig$ most likely dominated much of the variation in the second PC. Looking ahead at Figure~\ref{fig:sims_robust}A, as Laplacian error was added to the simulated data, PCA on $\X_2$ failed to recover the true signal since the Laplacian error increasingly contributed to the variation in the data. We also point out that MFA always yielded a lower error than concatenated PCA in Figure~\ref{fig:sims}, indicating that there is value in normalizing data sets to be on comparable scales. On the other hand, we must be weary of this normalization process. In the case of these simulations, PCA on $\X_2$ outperformed MFA, illustrating that normalization can sometimes remove important information.

To verify that these simulation results are not heavily dependent by the base simulation setup of $\Sig$ and $\Delt_k$, we also ran simulations, varying the dimension $d$ of the true joint subspace $\U$ and the number of data sets $K$. We provide these results in Appendix~\ref{sec:s_sims}.

Beyond simulating from the iPCA model \eqref{pop_model}, we check for robustness from the two main iPCA assumptions---normality and separability (i.e., the Kronecker covariance structure). To deviate from normality, we add Laplacian noise to the base simulation setup, and to depart from the Kronecker covariance structure, we simulate data from the JIVE model. The results are summarized in Figure~\ref{fig:sims_robust}, and we leave the simulation details as well as other simulations that demonstrate robustness to Appendix~\ref{sec:s_sims}.

As seen in Figure~\ref{fig:sims_robust}A, the Frobenius iPCA estimators appear to be relatively robust to non-Gaussian noise and outperformed their competitors even as the standard deviation of added Laplacian errors increased.  From the simulations under the JIVE model in Figure~\ref{fig:sims_robust}B, we see that as the amount of noise increases, JIVE given the known ranks yields the lowest error, as expected. But similar to how JIVE was comparable to competing methods under the iPCA model (Figure~\ref{fig:sims}), the iPCA estimators are comparable to competing methods for high noise levels under the JIVE model. At low noise levels however, the Frobenius iPCA estimators are surprisingly able to recover the true joint subspace better than JIVE with the known ranks. Further investigation into this peculiar behavior reveals that the Frobenius iPCA estimators give lower subspace recovery errors but much larger approximation errors $\norm{\Sig - \hat{\Sig}}_F^2$, compared to JIVE with the known ranks. This brings up a subtle, but important distinction---iPCA revolves around estimating the underlying subspace, determined by eigenvectors, while JIVE focuses on minimizing the matrix approximation error. These are inherently different objectives, and it is common for iPCA to estimate the eigenvectors well at the cost of a poor matrix approximation due to the regularized eigenvalues. 

\begin{figure}
\centering
\includegraphics[width =  .76\linewidth]{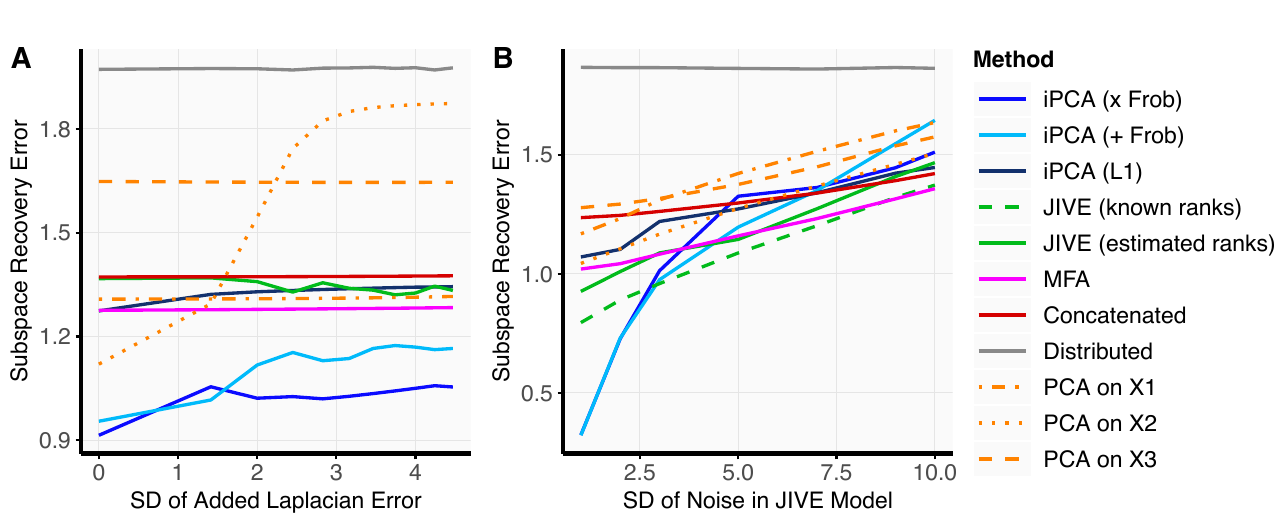}
\caption{\em \footnotesize Robustness Simulations: (A) As Laplacian error is increasingly added to the simulated data sets, the Frobenius iPCA estimators appear to be robust to the departures from Gaussianity; (B) As the amount of noise in the JIVE model increases, iPCA seems to be comparable to existing methods, illustrating its relative robustness to departures from the Kronecker product model.}
\label{fig:sims_robust}
\end{figure}


\subsection{Case Study: Integrative Genomics of Alzheimer's Disease} \label{sec:rosmap}

A key motivating example for our research is in integrative genomics, where the goal is to combine multiple genomic sources to gain insights into the genetic basis of diseases. In particular, apart from the APOE gene, little is known about the genomic basis of Alzheimer's disease (AD) and the genes which contribute to dominant expression patterns in AD. In this case study, we delve into the integrative genomics of AD and jointly analyze miRNA expression, gene expression via RNASeq, and DNA methylation data obtained from the Religious Orders Study Memory and Aging Project (ROSMAP) Study \citep{mostafavi2018rosmap}. The ROSMAP study is a longitudinal clinical-pathological cohort study of aging and AD, consisting of $507$ subjects, $309$ miRNAs, $900$ genes, and $1250$ CpG (methylation) sites after preprocessing (which we detail in Appendix~\ref{sec:s_rosmap}). This data is uniquely positioned for the study of AD since its genomics data is collected from post-mortem brain tissue from the dorsolateral prefrontal cortex, an area known to play a critical role in cognitive functions.

For our analysis, we consider two clinical outcomes: clinician's diagnosis and global cognition score. The clinician's diagnosis is the last clinical evaluation prior to the patient's death and is a categorical variable with three levels---Alzheimer's disease (AD), mild cognitive impairment (MCI), and no cognitive impairment (NCI). Global cognition score, a continuous variable, is the average of 19 cognitive tests and is the last cognitive testing score prior to death. While the clinician's diagnosis is sometimes subjective, global cognition score is viewed as a more objective measure of cognition. Our goal is to find common patterns among patients, which occur in all three data sets, and to understand whether these joint patterns are predictive of AD, as measured by the clinician's diagnosis and global cognition score.

To this end, we run iPCA and other existing methods to extract dominant patterns from the ROSMAP data. Figure~\ref{fig:rosmap_pc} shows the PC plots obtained from the various methods---each point represents a subject and is colored by either clinician's diagnosis or cognition score.

\begin{figure}[!b]
\centering
\includegraphics[width =  1\linewidth]{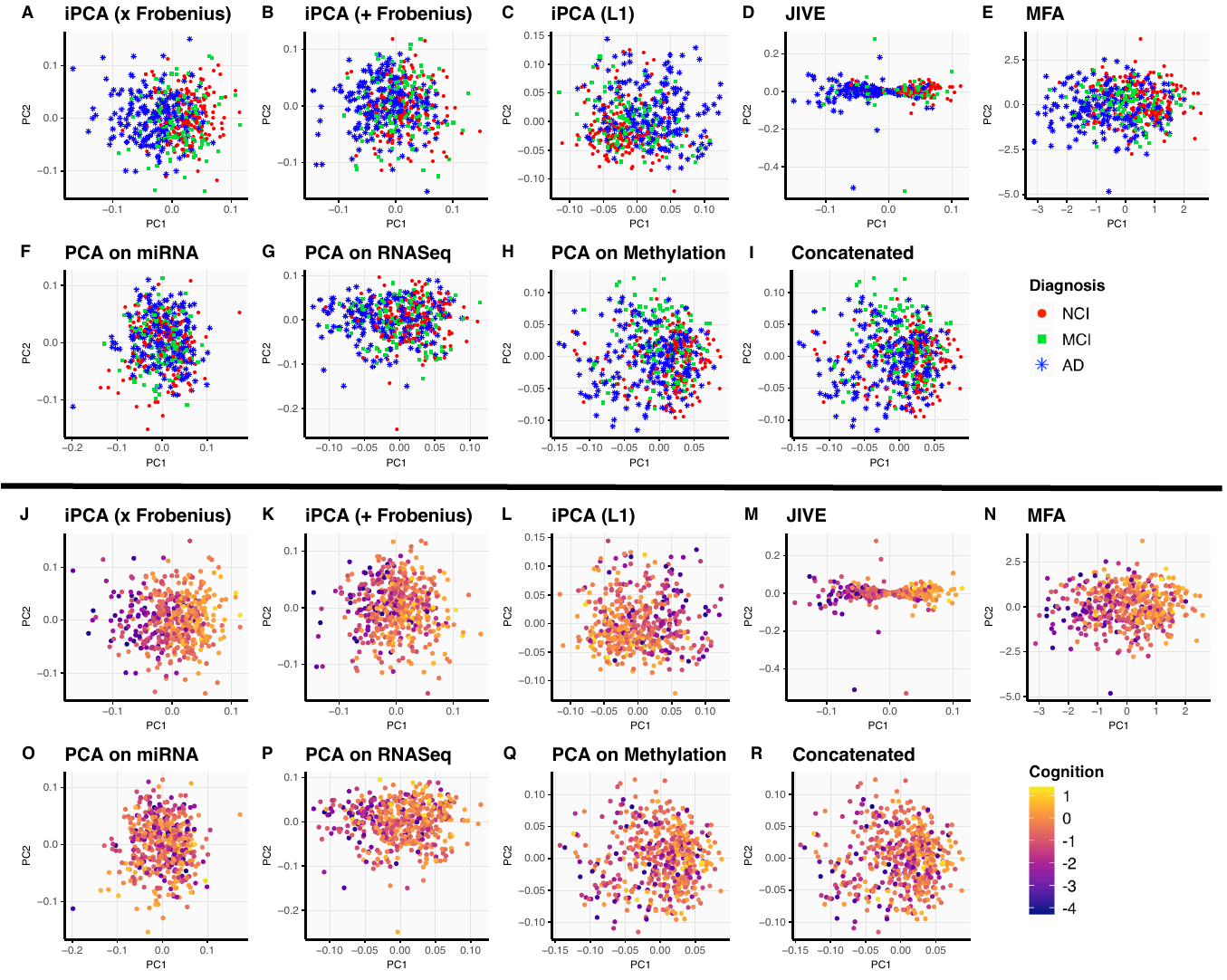}
\caption{\em \footnotesize We plot the first two (integrated) principal components from various methods applied to the ROSMAP data. Each point represents a subject, colored by the clinician's diagnosis in panels A-I and by global cognition score in panels J-R.}
\label{fig:rosmap_pc}
\end{figure}

Since visuals are a subjective measure of performance, we quantify it by taking the top PCs and using them as predictors in a random forest to predict the outcome of interest. The random forest test errors, averaged over 100 random training/test splits, are shown in Figure~\ref{fig:rosmap_rf}. Here, we see that the joint patterns extracted from iPCA using the multiplicative Frobenius penalty were the most predictive of the clinician's diagnosis of AD and the patient's global cognition score. Moreover, most of the predictive power can be attributed to the first three iPCs, which we visualize in Figure~\ref{fig:rosmap_3d_var}A-B. We also note that the top iPCs from iPCA with the multiplicative Frobenius penalties were more predictive than combining the PCs from the three individual PCAs performed on each data set. This showcases empirically that a joint analysis of the integrated data sets can be advantageous over three disparate analyses.

\begin{figure}[!tb]
\centering
\includegraphics[width =  .725\linewidth]{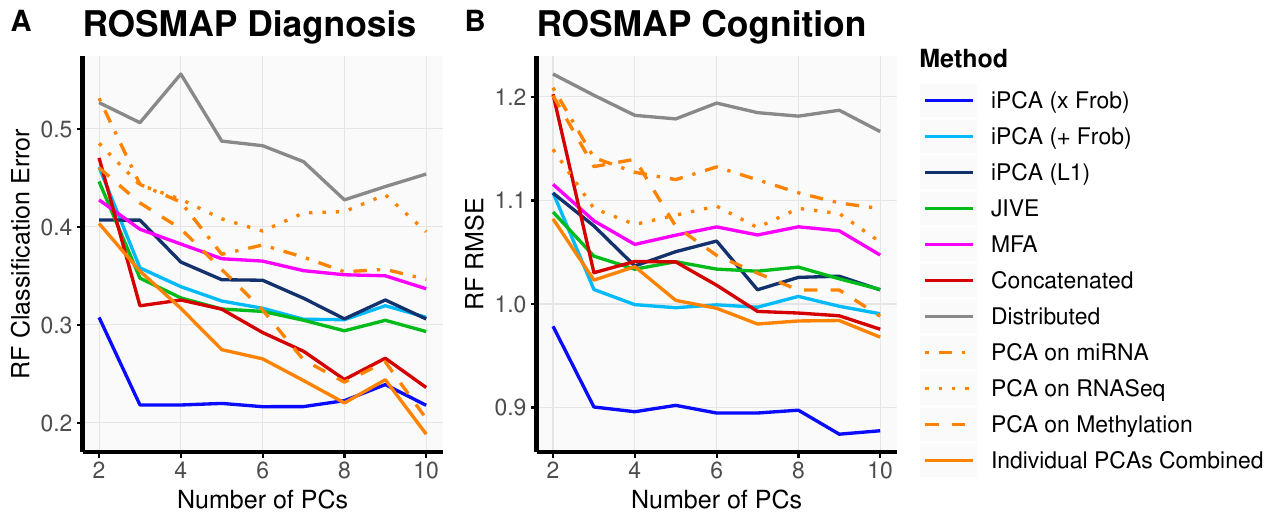}
\caption{\em \footnotesize We took the top PCs and used them as predictor variables in a random forest to predict (A) the clinician's diagnosis and (B) the global cognition score. For the random forest, we split the ROSMAP data into a training ($n = 375$) and test set ($n = 132$) and used the default random forest settings in R. The average test error from the random forests over 100 random splits are shown as the number of PCs used in the random forests increases.}
\label{fig:rosmap_rf}
\end{figure}

Beyond the high predictive power of iPCA with the multiplicative Frobenius penalty, it is perhaps more important for scientists to be able to interpret the iPCA results. One way is through the proportion of variance explained by the joint iPCs, as defined in Section~\ref{sec:var_explained}. Figure~\ref{fig:rosmap_3d_var}C shows the marginal proportions for the top 5 iPCs. It reveals that the RNASeq data set contributed the most variation in the joint patterns found by iPC1 and iPC2, and the miRNA data set contributed the most variation in iPC3. More interestingly, even though iPC2 and iPC3 have relatively small variances, iPCA is able to pick out these weak joint signals, which we found to be predictive of AD. This reiterates that the most variable patterns in the data are not necessarily the most relevant patterns for the question at hand. In this case, our goal was to find joint patterns which occur in all three data sets, and since the joint signal is not the most dominant source of variation in each data set, no individualistic PCA analysis would have identified the joint signal found by iPCA.

\begin{figure}[!t]
\centering
\includegraphics[width =  1\linewidth]{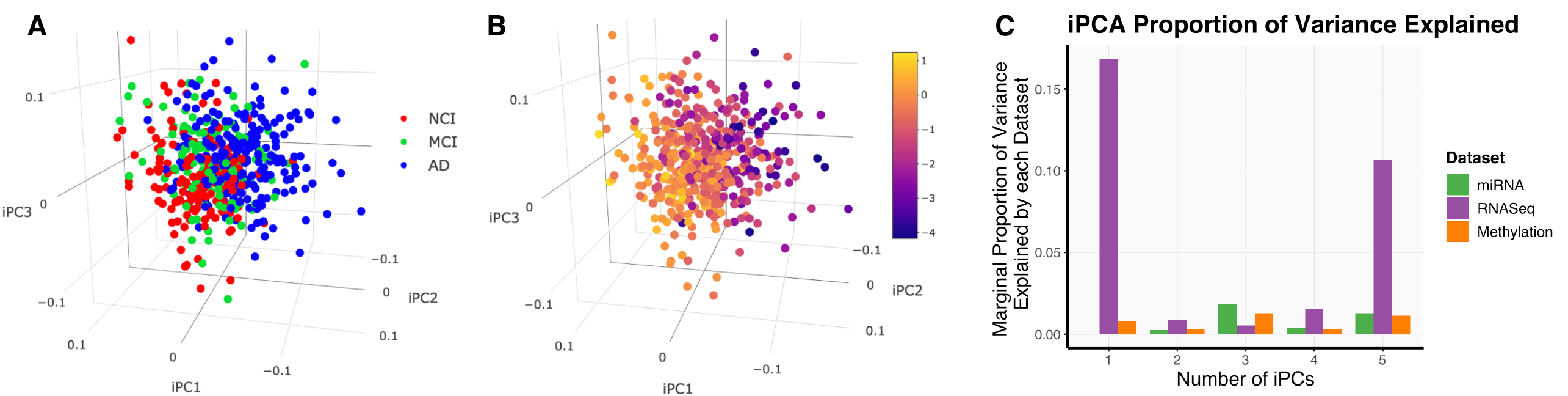}
\caption{\em \footnotesize (A)-(B) We show the top 3 iPCs obtained from iPCA with the multiplicative Frobenius estimator; the points are colored by clinician's diagnosis and global cognition score in (A) and (B), respectively. (C) We plot the marginal proportion of variance explained by the top iPCs in each data set in the ROSMAP analysis (using the multiplicative Frobenius iPCA estimator).}
\label{fig:rosmap_3d_var}
\end{figure}

We conclude our ROSMAP analysis by extracting the top genetic features which are associated to the joint patterns shown in Figure~\ref{fig:rosmap_3d_var}. Since iPCA provides an estimate of both $\Sig$ and $\Delt_k$, we can select the top features by applying sparse PCA to each $\hat{\Delt}_k$ obtained from iPCA. Table~\ref{tab:top_genes} lists the top miRNAs, genes, and genes affiliated with the selected CpG sites obtained from sparse PCA. Here, we used the \texttt{sparseEigen} R package \citep{benidis2016sparsepca} and chose the tuning parameter such that there were only $12$ non-zero features. 

Because the RNASeq data contributed most of the variation in iPC1, we did a literature search on the top five genes extracted by sparse PCA on $\hat{\Delt}_2$. Out of the top five genes, we found evidence in the biological literature, which links four of the five genes (the exception being SVOP) to AD \citep{carrette2003panel, li2017synaptic, han2014pituitary, espuny2017hallmarks}. While this is only a preliminary investigation into the importance of the genetic features obtained from iPCA, it is encouraging evidence and may potentially hint at candidate genes for future research.

\begin{table}[!t]
\centering
\footnotesize
\begin{tabular}{rlll}
  \hline
 & \textbf{miRNA} & \textbf{RNASeq} & \textbf{Methylation} \\ 
  \hline
1 & miR 216a & VGF & TMCO6 \\ 
  2 & miR 127 3p & SVOP & PHF3 \\ 
  3 & miR 124 & PCDHGC5 & BRUNOL4 \\ 
  4 & miR 30c & ADCYAP1 & OSCP1 \\ 
  5 & miR 143 & LINC01007 & GRIN2B \\ 
  6 & miR 27a & FRMPD2L1 & CASP9 \\ 
  7 & miR 603 & SLC30A3 & ZFP91; LPXN; ZFP91-CNTF \\ 
  8 & miR 423 3p & NCALD & CNP \\ 
  9 & miR 204 & S100A4 & YWHAE \\ 
  10 & miR 128 & AZGP1 & C11orf73 \\ 
  11 & miR 193a 3p & PAK1 & TMED10 \\ 
  12 & ebv miRBART14 & MAL2 & RELL1 \\ 
   \hline
\end{tabular}
\caption{Top genetic features obtained by applying Sparse PCA to each $\hat{\Delt}_k$ in ROSMAP analysis (using the multiplicative Frobenius iPCA estimator)} \label{tab:top_genes}
\end{table}

\section{Discussion} \label{sec:discussion}

As showcased in the simulations and the Alzheimer's disease case study, iPCA is not simply a theoretical construct that generalizes PCA to the integrated data setting. iPCA is also a useful and effective tool in practice to discover interesting joint patterns that are shared across multiple data sets. We believe that iPCA's strong empirical performance is due in part to its flexibility to handle a rich set of dependencies among samples and features concurrently, which is particularly appropriate and necessary for integrated data problems. This flexibility is inherently driven by the underlying matrix-variate normal model, and from a whitening perspective, we can view the iPCA model as a natural generalization of the classical PCA model and assumptions. More specifically, in relation to PCA, iPCA can be viewed as performing PCA on the concatenated feature-whitened data, having estimated the individual feature covariances $\Delt_k$ and the joint sample covariance $\Sig$ simultaneously.




While we discuss many potential penalized iPCA estimators for $\Sig$ and $\Delt_k$ in this work, we recommend that practitioners strongly consider using our new multiplicative Frobenius iPCA estimator. The simulations show that the Frobenius penalties are relatively robust to departures from model assumptions, and furthermore, the multiplicative Frobenius penalized estimator requires one less penalty parameter to tune and always converges to the global solution. Similar in spirit to other shrinkage penalties \citep{ledoit2004well}, the multiplicative Frobenius iPCA estimator is a shrinkage technique that induces a smoothness over the eigenvalues of the covariance matrices. However, its multiplicative form is especially unique and well-suited for integrated data problems as it performs automatic re-weighting of the integrated data sets and accounts for the concurrent estimation of $\Sig$ and $\Delt_k$ in each penalty term. Further investigation into the multiplicative Frobenius penalty and its statistical properties is left for future research.

Moreover, we believe that this work opens the door for research into the theoretical underpinnings of dimension reduction and data integration in ways that other non-model-based methods cannot. Building upon the model-based construction of iPCA, we show in Appendix~\ref{sec:s_consistency} that the additive $L_1$ correlation estimator satisfies one of the first statistical guarantees in the data integration context. However, this is only the beginning and certainly not meant to be the final investigation into theoretical properties of the iPCA estimators and data integration. Minimax results and consistency guarantees for the multiplicative Frobenius iPCA estimator in particular are challenging and will require careful follow up in future work.    


Still, there are many other open avenues for future exploration. Analogous to PCA, one might imagine similar fruitful extensions of iPCA to higher-order data, functional data, and other structured applications. One could continue exploring g-convex penalties in different contexts and problems. Another interesting area for future research would be to develop a general framework to prove the consistency of g-convex estimators using the intrinsic manifold space, rather than the Euclidean space. We believe this intersection of g-convexity and statistical theory is a particularly ripe area of future research, but overall, in this work, we developed a theoretically sound and practical tool for performing dimension reduction in the integrated data setting, thus facilitating holistic analyses at a large scale.

\acks{The authors acknowledge support from NSF DMS-1264058, NSF DMS-1554821 and NSF NeuroNex-1707400. T.T. also acknowledges support from the NSF Graduate Research Fellowship Program DGE-1752814. The authors thank Dr. Joshua Shulman and Dr. David Bennett for help acquiring the ROS/MAP data and acknowledge support from NIH P30AG10161, RF1AG15819, R01AG17917, and R01AG36042 for this data. The authors thank Dr. Zhandong Liu and Dr. Ying-Wooi Wan for help with processing and interpreting the ROS/MAP data. The authors also thank Cole Franks for help with improving and strengthening Theorem~\ref{nonexistence_thm}.}



\appendix
\section{Variance Explained by iPCA} \label{sec:s_var_explained}

To ensure that the cumulative proportion of variance explained from Definition~\ref{def:ipca_var_explained} is a well-defined concept, we check that $\text{PVE}_{k,m}$ is a proportion and is an increasing function as $m$ increases. This implies that the marginal proportion of variance explained given in Definition~\ref{def:mar_ipca_var_explained} is also a proportion.

\begin{proposition} \label{prop:prop_var_explained}
The cumulative proportion of variance explained in $\X_k$ by the top $m$ iPCs, as defined in \eqref{ipca_var_explained}, satisfies the following properties: for each $k = 1, \dots, K$ and $m = 1, \dots, \min\{n, p_k\}$,
\begin{enumerate}
\item [(i)] $0 \leq \text{PVE}_{k,m} \leq 1$;
\item [(ii)] $\text{PVE}_{k,m-1} \leq \text{PVE}_{k,m}$.
\end{enumerate}
\end{proposition}

\begin{proof}
(i) Since the Frobenius norm is always non-negative, it is clear that $\text{PVE}_{k,m} \geq 0$. So it suffices to show that $\text{PVE}_{k,m} \leq 1$, or equivalently, $\norm{ ( \smash{\U^{(m)}} )^T \X_k \smash{\V_k^{(m)}}}_F^2 \leq \norm{\X_k}_F^2$.

By definition of the Frobenius norm, we have that
\begin{align*}
\norm{ ( \smash{\U^{(m)}} )^T \X_k \smash{\V_k^{(m)}}}_F^2  &= \sum_{i=1}^{m} \sum_{j = 1}^{m} \left( ( \U^{(m)} )^T \X_k \V_k^{(m)} \right)_{ij}^2 \\
&= \sum_{i=1}^{m} \sum_{j = 1}^{m} \left( \sum_{q=1}^{n} \sum_{r=1}^{p_k} (\U^{(m)})^T_{iq} \X_{k,qr} \V^{(m)}_{k,rj} \right)^2 \\
&\stackrel{[1]}{\leq} \sum_{i=1}^{n} \sum_{j = 1}^{p_k} \left( \sum_{q=1}^{n} \sum_{r=1}^{p_k} \U_{iq}^T \X_{k,qr} \V_{k,rj} \right)^2 \\
&= \norm{\U^T \X_k \V_k}_F^2 \\
&\stackrel{[2]}{=} \norm{\X_k}_F^2.
\end{align*}
Here, [2] holds by the orthogonality of $\U$ and $\V_k$, and [1] follows from the facts that $m \leq \min\{ n, p_k\}$, $\U^{(m)}$ and $\V_k^{(m)}$ are precisely the first $m$ columns of $\U$ and $\V_k$ respectively, and the summand is non-negative. This concludes the proof of part (i).

\vspace{6pt}
\noindent
(ii) We follow a similar argument as part (i) to see that
\begin{align*}
\norm{ ( \smash{\U^{(m-1)}} )^T \X_k \smash{\V_k^{(m-1)}}}_F^2 &=  \sum_{i=1}^{m-1} \sum_{j = 1}^{m-1} \left( \sum_{q=1}^{n} \sum_{r=1}^{p_k} (\U^{(m-1)})^T_{iq} \X_{k,qr} \V^{(m-1)}_{k,rj} \right)^2 \\
&\leq \sum_{i=1}^{m} \sum_{j = 1}^{m} \left( \sum_{q=1}^{n} \sum_{r=1}^{p_k} (\U^{(m)})^T_{iq} \X_{k,qr} \V^{(m)}_{k,rj} \right)^2 \\
&= \norm{ ( \U^{(m)} )^T \X_k \V_k^{(m)}}_F^2.
\end{align*}
This implies that $\text{PVE}_{k,m-1} \leq \text{PVE}_{k,m}$.
\end{proof}

Next, we claim that Definition~\ref{def:ipca_var_explained} is a generalization of the cumulative proportion of variance explained in PCA. Recall that in PCA, if $\X = \U \D \V^T$ is the SVD of $\X$, then the cumulative proportion of variance explained by the top $m$ PCs is $(\sum_{i=1}^{m} d_i^2)/(\sum_{i=1}^{p} d_i^2)$, where $\D = \text{diag}(d_1, \dots, d_p)$. We can rewrite this as
\begin{align} \label{pca_var_explained}
\frac{\sum_{i=1}^{m} d_i^2}{\sum_{i=1}^{p} d_i^2} = \frac{\norm{\smash{\D^{(m)}}}_F^2}{\norm{\X}_F^2} = \frac{\norm{(\smash{\U^{(m)}})^T \X \smash{\V^{(m)}}}_F^2}{\norm{\X}_F^2}
\end{align}
using properties of the SVD. Since $\U$ are the PC scores, and $\V$ are the PC loadings from PCA, then Definition~\ref{def:ipca_var_explained} is indeed a natural definition in the sense that it generalizes the PCA cumulative proportion of variance explained. 

\section{Difference between iPCA and Concatenated PCA} \label{sec:s_cov_est}

To formalize the difference between iPCA and concatenated PCA, we examine the bias of the eigenvectors from concatenated PCA when sample and feature dependencies simultaneously obscure each other under the matrix-variate normal model.

\bias*

\begin{proof}
By construction of their methodology, we know that the iPCA scores $\U$ and the concatenated PCA scores $\tilde{\U}$ satisfy the following compact SVDs:
\begin{align*}
\X \Delt^{-1/2} &= \U \D \V^T, \\
\X &= \tilde{\U} \tilde{\D} \tilde{\V}^T. 
\end{align*}

We also know from \eqref{eq:ipca_model_equiv_z} that 
\begin{align}
\X = \Sig^{1/2} \Z \Delt^{1/2}. \label{eq:bias_x}
\end{align}

Using the orthogonality of the singular vectors and plugging in \eqref{eq:bias_x} for $\X$ in the compact SVD formulas then gives
\begin{align*}
\U &= \Sig^{1/2} \Z \V \D^{-1}, \\
\tilde{\U} &= \Sig^{1/2} \Z \Delt^{1/2} \tilde{\V} \tilde{\D}^{-1}.
\end{align*}

Therefore,
\begin{align*}
\U - \tilde{\U} = \Sig^{1/2} \Z \left(\V \D^{-1} - \Delt^{1/2} \tilde{\V} \tilde{\D}^{-1} \right).
\end{align*}
\end{proof}

\section{Covariance Estimation for iPCA} \label{sec:s_cov_est}

In this section, we provide the proofs and derivations related to the unpenalized and penalized maximum likelihood estimators under the iPCA model.

\subsection{Unpenalized Maximum Likelihood Estimators} \label{sec:s_unpenalized_mle}

We begin this section by deriving the log-likelihood equation associated with the iPCA population model \eqref{pop_model}.

Recall that the probability density function for each matrix-variate normal model $(k = 1, \dots, K)$ is given by
\begin{align*}
f\left(\X_{k} | \Sig, \Delt_k \right) = \left(2\pi\right)^{-\frac{np_k}{2}} | \Sig |^{-\frac{p_k}{2}} | \Delt_k |^{-\frac{n}{2}} \mathrm{exp}\left( -\frac{1}{2} \mathrm{tr} \left( \Sig^{-1} \X_{k} \Delt_k^{-1} \X_{k}^{T} \right) \right).
\end{align*}
Hence, the log-likelihood function is
\begin{align*}
\ell(\Sig^{-1}, \Delt_1^{-1}, \dots, \Delt_K^{-1}) &= \sum_{k = 1}^{K} \Big[ -\frac{np_k}{2} \log (2\pi) + \frac{p_k}{2} \log | \Sig^{-1} | + \frac{n}{2} \log | \Delt_k^{-1} | \\
& \qquad \qquad - \frac{1}{2} \mathrm{tr} \left( \Sig^{-1} \X_{k} \Delt_k^{-1} \X_{k}^{T} \right) \Big] \\
&\propto p \log | \Sig^{-1} | + n\sum_{k = 1}^{K} \log | \Delt_k^{-1} | - \sum_{k = 1}^{K} \mathrm{tr}\left( \Sig^{-1} \X_{k} \Delt_k^{-1} \X_{k}^{T} \right).
\end{align*}
The unpenalized MLEs are designed to solve the optimization problem
\begin{align} \label{s_ff_loglike}
\argmin_{\Sig^{-1}, \Delt_1^{-1}, \dots, \Delt_K^{-1}} -\ell(\Sig^{-1}, \Delt_1^{-1}, \dots, \Delt_K^{-1})
\end{align}
Taking partial derivatives with respect to each covariance parameter thus gives the following:

\unpenalizedmle*


\begin{proof}
Taking the partial derivatives of the log-likelihood equation with respect to $\Sig^{-1}$ and $\Delt_k^{-1}$ respectively yields the gradient equations:
\begin{align*}
\frac{\partial \ell}{\partial \Sig^{-1}} &= p \Sig - \sum_{k=1}^{K} \X_{k} \Delt_k^{-1} \X_{k}^{T}\\
\frac{\partial \ell}{\partial \Delt_k^{-1}} &= n \Delt_k - \X_{k}^{T} \Sig^{-1} \X_{k}.
\end{align*}
Setting the gradient equations equal to $\mathbf{0}$ gives the desired result.
\end{proof}

Assuming that the unpenalized MLEs exist, we can compute the unpenalized MLEs of $\Sig$ and $\Delt_1, \dots, \Delt_K$ via Algorithm~\ref{alg:unpenalized}, which is analogous to the Flip-Flop algorithm provided in \citet{dutilleul1999mle}.


\begin{algorithm}
\caption{Flip-Flop Algorithm for iPCA Unpenalized MLEs} \label{alg:unpenalized}
\begin{algorithmic} [1]
\State Assume that $\X_k$ has been centered appropriately and that the unpenalized MLEs exist.
\State Initialize $\hat{\Sig}$, $\hat{\Delt}_1, \dots, \hat{\Delt}_K$ to be symmetric positive definite.
\While{not converged}
\State Update $\hat{\Sig} = \frac{1}{p} \sum_{k=1}^{K} \X_k \smash[t] {\hat{\Delt}}^{-1}_k \X_k^T$
\For{$k = 1, \dots, K$}
\State Update $\hat{\Delt}_k = \frac{1}{n} \X_k^T \smash[t] {\hat{\Sig}}^{-1} \X_k$
\EndFor
\EndWhile
\end{algorithmic}
\end{algorithm}

The next theorem provides very restrictive conditions for which the unpenalized MLEs exist, but in almost all practical cases, the unpenalized maximum likelihood problem is ill-posed for iPCA.

\nonexistence*

\begin{proof}
$(ii)$ Suppose that $\X_k$ has not been column-centered but that $\mathrm{rank}(\tilde{\X}) = n$ and $\mathrm{rank}(\X_k) = p_k$ for $k = 1, \dots, K$. We will first prove part $(b)$, so assume also that the unpenalized log-likelihood is bounded. From the unpenalized MLEs in Lemma~\ref{unpenalized_mle}, it is easy to see that $\hat{\Sig}$ and $\hat{\Delt}_1, \dots, \hat{\Delt}_K$ are symmetric positive semidefinite. We next claim that $\hat{\Sig}$ and $\hat{\Delt}_1, \dots, \hat{\Delt}_K$ are full rank and hence positive definite. To prove this claim, we proceed by induction on the unpenalized Flip-Flop iteration counter $m$. Let $\hat{\Sig}^{m}$ and $\hat{\Delt}_k^{m}$ denote the $m^{th}$ Flip-Flop update of $\hat{\Sig}$ and $\hat{\Delt}_k$, respectively.

Clearly, the base case holds since $\Sig^{0}$ and $\Delt_1^{0}, \ldots, \Delt_K^{0}$ are initialized to be symmetric positive definite. So suppose $\Sig^{m}$ and $\Delt_1^{m} \ldots, \Delt_K^{m}$ are full rank. Then the unpenalized Flip-Flop iterates at the $(m+1)^{th}$-update step are
\begin{align*}
\hat{\Sig}^{m+1} &= \frac{1}{p} \sum_{k=1}^{K} \X_{k} (\hat{\Delt}_{k}^{m})^{-1} \X_{k}^{T} = \frac{1}{p} \tilde{\X} (\tilde{\Delt}^{m})^{-1} \tilde{\X}^T, \\
\hat{\Delt}_{k}^{m+1} &= \frac{1}{n} \X_{k}^{T} (\hat{\Sig}^{m+1})^{-1} \X_{k},
\end{align*}
where $\tilde{\X} = [\X_1, \dots, \X_K]$ and $\tilde{\Delt}^{m} = \mathrm{diag}(\hat{\Delt}_1^{m}, \dots, \hat{\Delt}_K^{m})$. 

Therefore, we have that
\begin{align*}
\mathrm{rank}\left(\hat{\Sig}^{m+1}\right) = \mathrm{rank}\left( \tilde{\X} (\tilde{\Delt}^{m})^{-1} \tilde{\X}^T \right) = \mathrm{rank} \left( \tilde{\X} (\tilde{\Delt}^{m})^{-\frac{1}{2}} \right) = \mathrm{rank} (\tilde{\X}) = n.
\end{align*}
Here, the second equality holds because $\mathrm{rank}(\A^T \A) = \mathrm{rank}(\A) = \mathrm{rank}(\A^T)$. The third equality holds because $(\tilde{\Delt}^{m})^{-\frac{1}{2}}$ is full rank (i.e. rank = $p$) by the inductive hypothesis, and the last equality holds by hypothesis. Thus, $\hat{\Sig}^{m+1}$ is positive definite. 

Similarly, for each $k = 1, \dots, K$,
\begin{align*}
\mathrm{rank} \left( \hat{\Delt}_{k}^{m+1} \right) = \mathrm{rank}\left( \X_k^T (\hat{\Sig}^{m+1})^{-1} \X_k\right) = \mathrm{rank}(\X_k) = p_k.
\end{align*}
So for each $k = 1, \dots, K$, $\hat{\Delt}_{k}^{m+1}$ is positive definite.

By induction, $\hat{\Sig}^m, \hat{\Delt}_{1}^m, \dots, \hat{\Delt}_{K}^m \succ 0$ for each iterate of Algorithm~\ref{alg:unpenalized}, and by Corollary~\ref{unpenalized_conv}, Algorithm~\ref{alg:unpenalized} converges to the global solution of \eqref{s_ff_loglike}. Thus, the unpenalized MLEs for $\Sig$ and $\Delt_1^{-1}, \dots, \Delt_K^{-1}$ exist under the assumptions given in part $(ii)(b)$.

To prove part $(a)$ of $(ii)$, let $\X_1, \dots, \X_K$ be given, and assume that $n \neq p_k$ for some $k = 1, \dots, K$. By the rank assumptions, it must be that $p_k \leq n \leq p$ for each $k = 1, \dots, K$. Define $\mathcal{K}^0 = \{ k \in \{1, \dots, K\} \mid p_k < n \}$ and $\mathcal{K}^1 = \{ k \in \{1, \dots, K\} \mid p_k = n \}$. Note that by assumption, $\mathcal{K}^0$ is not empty, and for each $k \in \mathcal{K}^0$, $\X_k^T$ has a kernel. We can then choose bases for $\mathbb{R}^n$ and $\mathbb{R}^{p_k}$ ($k = 1, \dots, K$) such that the first row of $\X_k$ is zero for each $k \in \mathcal{K}^0$ and the first row of $\X_k$ is the first standard basis vector, denoted $\ff_{k, 1}$, of $\mathbb{R}^{p_k}$ for each $k \in \mathcal{K}^1$. Let $\ee_1$ denote the first standard basis vector in $\mathbb{R}^n$. 

We will next construct a family of diagonal matrices for $\Sig, 
\Delt_1, \dots, \Delt_K$ that sends the unpenalized log-likelihood to infinity. Specifically, let $\Sig = e^{t \ee_1 \ee_1^T}$, $\Delt_k = \mathbf{I}_{p_k}$ for $k \in \mathcal{K}^0$, and $\Delt_k = e^{-t \ff_{k, 1} \ff_{k, 1}^T}$ for $k \in \mathcal{K}^1$, where $e$ denotes the matrix exponential. Then for $k \in \mathcal{K}^0$,
\begin{align*}
\mathrm{tr}\left( \Sig^{-1} \X_k \Delt_k^{-1} \X_k^T \right) = \mathrm{tr}\left( \Sig^{-1} \X_k \X_k^T \right) = \sum_{i = 2}^n \x_{k, i}^T \x_{k, i},
\end{align*}
and for $k \in \mathcal{K}^1$,
\begin{align*}
\mathrm{tr}\left( \Sig^{-1} \X_k \Delt_k^{-1} \X_k^T \right) = 1 + \sum_{i = 2}^n \left( (e^{t} - 1) (\x_{k, i}^T \ff_{k, 1})^2 + \x_{k, i}^T \x_{k, i} \right),
\end{align*}
where $\x_{k, i}$ denotes the $i^{th}$ row of $\X_k$.

Therefore, for every $t \in \mathbb{R}$, we have that
\begin{align*}
\ell(\Sig^{-1}, \Delt_1^{-1}, \dots, \Delt_K^{-1}) &= -p \log | \Sig | - n \sum_{k \in \mathcal{K}^0} \log | \Delt_k | - n \sum_{k \in \mathcal{K}^1} \log | \Delt_k | \\
&\qquad - \sum_{k \in \mathcal{K}^0} \mathrm{tr}\left( \Sig^{-1} \X_k \Delt_k^{-1} \X_k^T \right) - \sum_{k \in \mathcal{K}^1} \mathrm{tr}\left( \Sig^{-1} \X_k \Delt_k^{-1} \X_k^T \right) \\
&= -p \log(e^t) - n \sum_{k \in \mathcal{K}^0} 0 - n \sum_{k \in \mathcal{K}^1} \log(e^{-t})  \\
&\qquad - \sum_{k \in \mathcal{K}^0} \sum_{i = 2}^n \x_{k, i}^T \x_{k, i} - \sum_{k \in \mathcal{K}^1} \left( 1 + \sum_{i = 2}^n \left( (e^{t} - 1) (\x_{k, i}^T \ff_{k, 1})^2 + \x_{k, i}^T \x_{k, i} \right) \right) \\ 
&= t(n |\mathcal{K}^1| - p) - \sum_{k \in \mathcal{K}^0} \sum_{i = 2}^n \x_{k, i}^T \x_{k, i} \\
&\qquad - \sum_{k \in \mathcal{K}^1} \left( 1 + \sum_{i = 2}^n \left( (e^{t} - 1) (\x_{k, i}^T \ff_{k, 1})^2 + \x_{k, i}^T \x_{k, i} \right) \right).
\end{align*}

As $t \rightarrow -\infty$, the first term approaches $\infty$ since $n |\mathcal{K}^1| < p$ while the second and third terms approach a constant. Thus, $\ell(\Sig^{-1}, \Delt_1^{-1}, \dots, \Delt_K^{-1})$ is unbounded, as desired.

\vspace{4pt}
\noindent
$(i)$ Suppose that we have centered each data matrix $\X_k$ to have column means $\mathbf{0}$. Recall that the unpenalized MLEs for $\Sig$ and $\Delt_1, \dots, \Delt_K$ are obtained by
\begin{align*}
\hat{\Sig} &= \frac{1}{p} \sum_{k=1}^{K} \X_{k} \smash[t] {\hat{\Delt}}^{-1}_k \X_{k}^{T} = \frac{1}{p} \tilde{\X} \smash[t] {\tilde{\Delt}}^{-1} \tilde{\X}^T, \\
\hat{\Delt}_k &= \frac{1}{n} \X_{k}^{T} \smash[t] {\hat{\Sig}}^{-1} \X_{k},
\end{align*}
where $\smash[t] {\tilde{\Delt}} := \text{diag}(\smash[t] {\hat{\Delt}}_1, \dots, \smash[t] {\hat{\Delt}}_K)$.

For the sake of contradiction, suppose that there exists $n, p_1, \ldots, p_k$ such that $\Sig^{-1}$ and $\Delt_1^{-1}, \ldots, \Delt_K^{-1}$ are symmetric positive definite. By the same argument as in part $(ii)(b)$, $\mathrm{rank}(\hat{\Sig}) = \mathrm{rank}(\tilde{\X})$, but since each $\X_k$ has been centered to have column means $\mathbf{0}$, then the $n$ rows of $\tilde{\X}$ are linearly dependent. Hence, $\mathrm{rank}(\hat{\Sig}) = \mathrm{rank}(\tilde{\X}) < n$, which implies that $\hat{\Sig}$ can never be positive definite, a contradiction. Therefore, the unpenalized MLEs never exist if each $\X_k$ has be column-centered to have mean $\mathbf{0}$.
\end{proof}

\begin{remark}
Notice that if the unpenalized MLEs for $\Sig, \Delt_1, \dots, \Delt_K$ exist, then 
\begin{align*}
n &\stackrel{[1]}{=} \mathrm{rank}(\hat{\Sig}) \stackrel{[2]}{=} \mathrm{rank}(\tilde{\X}) \stackrel{[3]}{\leq} \min\{ n, p \} \\
p_k &\stackrel{[1]}{=} \mathrm{rank}(\hat{\Delt}_k) \stackrel{[2]}{=} \mathrm{rank}(\X_k) \stackrel{[3]}{\leq} \min \{ n, p_k \} \qquad \forall \: k = 1, \dots, K
\end{align*}
where [1] follows from positive definiteness of $\hat{\Sig}$ and $\hat{\Delt}_k$, [2] follows the same argument as in the proof of Theorem~\ref{nonexistence_thm}, and [3] holds by properties of rank and the dimensions of $\tilde{\X}$ and $\X_k$. By combining these rank constraints with the result of Theorem~\ref{nonexistence_thm}, we see that if the unpenalized MLEs for $\Sig, \Delt_1, \dots, \Delt_K$ exist, it must be that $p_k = n \leq p$ for each $k = 1, \dots, K$.
\end{remark}

Proposition~\ref{mle_naive} can be proved in the same way as Lemma~\ref{unpenalized_mle}, so we omit the proof.

\subsection{Penalized Maximum Likelihood Estimators} \label{sec:s_penalized_mle}

In this section, we develop Flip-Flop algorithms and analyze the convergence results for both Frobenius and $L_1$ penalties. For the sake of notation, let
\begin{align*}
-\ell_{P}(\Sig^{-1}, \Delt_1^{-1}, \ldots, \Delt_K^{-1}) &= -p \log | \Sig^{-1} | - n \sum_{k=1}^{K} \log | \Delt_k^{-1} | + \sum_{k = 1}^{K} \mathrm{tr} \left( \Sig^{-1} \X_{k} \Delt_k^{-1} \X_{k}^{T} \right) \nonumber \\
& \qquad \qquad + P(\Sig^{-1}, \Delt_1^{-1}, \dots, \Delt_K^{-1})
\end{align*}

We give the overarching framework of the Flip-Flop algorithms in Algorithm~\ref{alg:ff}, and we show in Theorem~\ref{converge_to_stationary} that Algorithm~\ref{alg:ff} can be used to find a local solution of \eqref{penalized_mle} for a certain class of penalties, which includes the additive Frobenius, multiplicative Frobenius, and additive $L_1$ penalties. The main idea behind the proof is to use convexity and view Algorithm~\ref{alg:ff} as a block coordinate descent algorithm so that each update of the Flip-Flop algorithm is a descent direction.  

\begin{algorithm}
\caption{Outline of Flip-Flop Algorithm for Penalized iPCA Covariance Estimators}\label{alg:ff}
\begin{algorithmic} [1]
\State Center the columns of $\X_1, \dots, \X_K$, and initialize $\hat{\Sig}$, $\hat{\Delt}_1, \dots, \hat{\Delt}_K$ to be positive definite.
\While{not converged}
\State Update $\Sig$ while fixing all other variables:

\setlength\parindent{2cm} $\displaystyle \smash[t] {\hat{\Sig}}^{-1} = \argmin_{\Sig^{-1} \succ 0} ~ -\ell_{P}(\Sig^{-1}, \smash[t] {\hat{\Delt}}^{-1}_1, \ldots, \smash[t] {\hat{\Delt}}^{-1}_K)$

\For{$k = 1, \dots, K$}
\State Update $\Delt_k$ while fixing all other variables:

\setlength\parindent{2cm} $\displaystyle \smash[t] {\hat{\Delt}}^{-1}_k = \argmin_{\Delt_k^{-1} \succ 0} ~ -\ell_{P}(\smash[t] {\hat{\Sig}}^{-1}, \smash[t] {\hat{\Delt}}^{-1}_1, \ldots, \smash[t] {\hat{\Delt}}^{-1}_{k-1}, \Delt_k^{-1}, \smash[t] {\hat{\Delt}}^{-1}_{k+1}, \ldots, \smash[t] {\hat{\Delt}}^{-1}_K)$
\EndFor
\EndWhile
\end{algorithmic}
\end{algorithm}

\localconvergence*

\begin{proof}
Suppose either $P(\Sig^{-1}, \Delt_1^{-1}, \dots, \Delt_K^{-1})$ is a differentiable convex function with respect to each coordinate or $P(\Sig^{-1}, \Delt_1^{-1}, \dots, \Delt_K^{-1}) = P_{0}(\Sig^{-1}) + \sum_{k=1}^{K} P_k(\Delt_k^{-1})$ where $P_i$ is a (non-differentiable) convex function for each $i = 1, \dots, K$. 

Let $\ell(\Sig^{-1}, \Delt_1^{-1}, \dots, \Delt_K^{-1}) = p \log | \Sig^{-1} | + n \sum_{k=1}^{K} \log | \Delt_k^{-1} | - \sum_{k = 1}^{K} \mathrm{tr} \left( \Sig^{-1} \X_{k} \Delt_k^{-1} \X_{k}^{T} \right)$. Since the domain of $-\ell$ is open and $-\ell$ is Gateaux-differentiable on its domain, then $-\ell_P$ is regular in the domain of $-\ell_P$ by Lemma 3.1 in \citet{tseng2001convergence}. 

Note also that since the log-determinant is a strictly concave function on the set of symmetric positive definite matrices, the trace function is linear, and the penalty term is convex with respect to each coordinate by hypothesis, then 
\begin{itemize}
\item $-\ell_P$ is strictly convex in $\Sig^{-1}$ with $\Delt_1^{-1}, \dots, \Delt_K^{-1}$ fixed, and
\item for each $k = 1, \dots, K$, $-\ell_P$ is strictly convex in in $\Delt_k^{-1}$ with $\Sig^{-1}, \Delt_j^{-1}$, $j \neq k$ fixed.
\end{itemize}

Because $-\ell_P$ is regular and strictly convex with respect to each coordinate, it follows that the Flip-Flop algorithm corresponding to \eqref{penalized_mle} converges to a stationary point of the objective function by Theorem 4.1(c) in \citet{tseng2001convergence}.
\end{proof}

In the following sections, we will derive the specific form of the Flip-Flop updates for each of the penalized iPCA estimators.

\subsubsection{Additive Frobenius Penalized Flip-Flop Estimator} \label{sec:s_addfrob}

To compute the additive Frobenius penalized estimator, we solve
\begin{align}\label{addfrob}
\smash[t] {\hat{\Sig}}^{-1}, \smash[t] {\hat{\Delt}}^{-1}_1, \dots, \smash[t] {\hat{\Delt}}^{-1}_K &=  \smash{\argmax_{\substack{\Sig^{-1} \succ 0 \\ \Delt_1^{-1}, \dots, \Delt_K^{-1} \succ 0}}} \Big \{ p \log | \Sig^{-1} | + n \sum_{k=1}^{K} \log | \Delt_k^{-1} | - \sum_{k = 1}^{K} \mathrm{tr}\left( \Sig^{-1} \X_{k} \Delt_k^{-1} \X_{k}^{T} \right) \nonumber \\
& \qquad \qquad \qquad \qquad \qquad - \lambda_{\Sig} \norm{\Sig^{-1}}_F^2 - \sum_{k=1}^{K} \lambda_k \norm{\Delt_k^{-1}}_F^2 \Big \}.
\end{align}
The gradient equations corresponding to \eqref{addfrob} are given by
\begin{align}
p \hat{\Sig} - \sum_{k=1}^{K} \X_{k} \smash[t] {\hat{\Delt}}^{-1}_k \X_{k}^{T} - 2 \lambda_{\Sig} \smash[t] {\hat{\Sig}}^{-1} &= 0, \label{add_grad_eq1} \\ 
n \hat{\Delt}_k - \X_{k}^{T} \smash[t] {\hat{\Sig}}^{-1} \X_{k} - 2 \lambda_{k} \smash[t] {\hat{\Delt}}^{-1}_k &= 0 \qquad \forall \: k = 1, \dots, K \label{add_grad_eq2}.
\end{align}

\begin{lemma} \label{lemma:addfrob_explicit}
If gradient equations \eqref{add_grad_eq1} and \eqref{add_grad_eq2} are satisfied, then
\begin{align}
\hat{\Sig} &= \left( \frac{2 \lambda_{\Sig}}{p} \I + \frac{1}{4 p^2} \bigg( \sum_{k=1}^{K} \X_{k} \smash[t] {\hat{\Delt}}^{-1}_k \X_{k}^{T} \bigg)^2 \right)^\frac{1}{2} + \frac{1}{2p} \sum_{k=1}^{K} \X_{k} \smash[t] {\hat{\Delt}}^{-1}_k \X_{k}^{T} \label{addfrob_explicit_sol1} \\
\hat{\Delt}_k &= \left( \frac{2 \lambda_k}{n} \I + \frac{1}{4 n^2} \left( \X_{k}^{T} \smash[t] {\hat{\Sig}}^{-1} \X_{k} \right)^2 \right)^\frac{1}{2} + \frac{1}{2n} \X_{k}^{T} \smash[t] {\hat{\Sig}}^{-1} \X_{k} \qquad \forall \: k = 1, \dots, K \label{addfrob_explicit_sol2}
\end{align}
\end{lemma}

\begin{proof}
Define $\hat{\mathbf{S}}_{\Sig} = \sum_{k=1}^{K} \X_{k} \smash[t] {\hat{\Delt}}^{-1}_k \X_{k}^{T}$, and right multiply \eqref{add_grad_eq1} by $\hat{\Sig}$ to obtain
\begin{align}
p \hat{\Sig}^2 - \hat{\mathbf{S}}_{\Sig} \hat{\Sig} - 2 \lambda_{\Sig} \I &= \mathbf{0} \\
\implies \hat{\Sig}^2 - \frac{1}{p}\hat{\mathbf{S}}_{\Sig} \hat{\Sig} + \frac{1}{4p^2} \hat{\mathbf{S}}_{\Sig}^2 &= \frac{2 \lambda_{\Sig}}{p} \I + \frac{1}{4p^2} \hat{\mathbf{S}}_{\Sig}^2. \label{addfrob_explicit_eq1}
\end{align}
On the other hand, we can multiply \eqref{add_grad_eq1} by $\hat{\Sig}$ on the left to obtain
\begin{align}
p \hat{\Sig}^2 - \hat{\Sig} \hat{\mathbf{S}}_{\Sig} - 2 \lambda_{\Sig} \I &= \mathbf{0} \\
\implies \hat{\Sig}^2 - \frac{1}{p} \hat{\Sig} \hat{\mathbf{S}}_{\Sig} + \frac{1}{4p^2} \hat{\mathbf{S}}_{\Sig}^2 &= \frac{2 \lambda_{\Sig}}{p} \I + \frac{1}{4p^2} \hat{\mathbf{S}}_{\Sig}^2. \label{addfrob_explicit_eq2}
\end{align}
Adding \eqref{addfrob_explicit_eq1} and \eqref{addfrob_explicit_eq2} and then dividing by 2 gives
\begin{align*}
\hat{\Sig}^2 - \frac{1}{2p}\hat{\mathbf{S}}_{\Sig} \hat{\Sig} - \frac{1}{2p} \hat{\Sig} \hat{\mathbf{S}}_{\Sig} + \frac{1}{4p^2} \hat{\mathbf{S}}_{\Sig}^2 &= \frac{2 \lambda_{\Sig}}{p} \I + \frac{1}{4p^2} \hat{\mathbf{S}}_{\Sig}^2 \\
\implies \left( \hat{\Sig} -  \frac{1}{2p} \hat{\mathbf{S}}_{\Sig} \right)^2 &= \frac{2 \lambda_{\Sig}}{p} \I + \frac{1}{4p^2} \hat{\mathbf{S}}_{\Sig}^2.
\end{align*}
Since $\frac{2 \lambda_{\Sig}}{p} \I + \frac{1}{4p^2} \hat{\mathbf{S}}_{\Sig}^2$ is positive definite, its has a unique square root. Thus,
\begin{align*}
\hat{\Sig} = \left( \frac{2 \lambda_{\Sig}}{p} \I + \frac{1}{4p^2} \hat{\mathbf{S}}_{\Sig}^2 \right)^{1/2} + \frac{1}{2p} \hat{\mathbf{S}}_{\Sig}.
\end{align*}

We can similarly rearrange \eqref{add_grad_eq2} and complete the square to obtain \eqref{addfrob_explicit_sol2}.
\end{proof}

We now have the machinery to prove Proposition~\ref{prop:addfrob_eig}, which gives us the form of each update in the additive Frobenius Flip-Flop algorithm.

\begin{proposition} \label{prop:addfrob_eig}
$\hat{\Sig}$ and $\hat{\Delt}_1, \dots, \hat{\Delt}_K$ are solutions to the gradient equations in \eqref{add_grad_eq1} and \eqref{add_grad_eq2} if and only if
\begin{align}
\hat{\Sig} &= \U \left[ \frac{1}{2p} \left( \boldsymbol{\Gamma} + \left( \boldsymbol{\Gamma}^2 + 8 \lambda_{\Sig}p \I \right)^{\frac{1}{2}} \right) \right] \U^T \label{prop:addfrob_eig_eq1} \\
\text{and } ~ \hat{\Delt}_k &= \V_k \left[ \frac{1}{2n} \left( \boldsymbol{\Phi}_k + \left( \boldsymbol{\Phi}_k^2 + 8 \lambda_{k} n \I \right)^{\frac{1}{2}} \right) \right] \V_k^T \qquad \forall \: k = 1, \dots, K, \label{prop:addfrob_eig_eq2}
\end{align}
where $\U, \V_k, \boldsymbol{\Gamma},$ and $\boldsymbol{\Phi}_k$ are defined by the eigendecompositions $\sum_{k=1}^{K} \X_k \smash[t] {\hat{\Delt}}^{-1}_k \X_k^T = \U \boldsymbol{\Gamma} \U^T$ and $\X_k^T \smash[t] {\hat{\Sig}}^{-1} \X_k = \V_k \boldsymbol{\Phi}_k \V_k^T$.
\end{proposition}

\begin{proof}
$(\Rightarrow)$ Suppose that $\hat{\Sig}$ and $\hat{\Delt}_1, \dots, \hat{\Delt}_K$ are solutions to the gradient equations in \eqref{add_grad_eq1} and \eqref{add_grad_eq2}. We will first show that the eigenvectors of $\hat{\Sig}$ and $\sum_{k=1}^{K} \X_{k} \smash[t] {\hat{\Delt}}^{-1}_k \X_{k}^{T}$ are equivalent. 

Let $\uu$ be an eigenvector of $\hat{\Sig}$ with the corresponding eigenvalue $\phi$. Then, by \eqref{add_grad_eq1}, 
\begin{align*}
\smash{\sum_{k=1}^{K} \X_{k} \smash[t] {\hat{\Delt}}^{-1}_k \X_{k}^{T} \uu} = ( p \hat{\Sig} - 2 \lambda_{\Sig} \smash[t] {\hat{\Sig}}^{-1} ) \uu = (p \phi - 2 \lambda_{\Sig} \phi^{-1}) \uu.
\end{align*} 
Therefore, $\uu$ is an eigenvector of $\sum_{k=1}^{K} \X_{k} \smash[t] {\hat{\Delt}}^{-1}_k \X_{k}^{T}$ with the eigenvalue $p \phi - 2 \lambda_{\Sig} \phi^{-1}$.

Conversely, suppose $\uu$ is an eigenvector of $\sum_{k=1}^{K} \X_{k} \smash[t] {\hat{\Delt}}^{-1}_k \X_{k}^{T}$ with eigenvalue $\gamma$. Then
\begin{align*}
\left( \frac{2 \lambda_{\Sig}}{p} \I + \frac{1}{4p^2} \left(\sum_{k=1}^{K} \X_{k} \smash[t] {\hat{\Delt}}^{-1}_k \X_{k}^{T} \right)^2 \right) \uu &= \left( \frac{2 \lambda_{\Sig}}{p} + \frac{1}{4p^2} \gamma^2 \right) \uu.
\end{align*}
This implies that
\begin{align}
\left( \frac{2 \lambda_{\Sig}}{p} \I + \frac{1}{4p^2} \left( \sum_{k=1}^{K} \X_{k} \smash[t] {\hat{\Delt}}^{-1}_k \X_{k}^{T} \right)^2 \right)^{\frac{1}{2}} \uu &= \left( \frac{2 \lambda_{\Sig}}{p} + \frac{1}{4p^2} \gamma^2 \right)^{\frac{1}{2}} \uu. \label{addfrob_eig_eq1}
\end{align}
So by Lemma~\ref{lemma:addfrob_explicit} and \eqref{addfrob_eig_eq1}, we have that
\begin{align}
\hat{\Sig} \uu &= \left[ \left( \frac{2 \lambda_{\Sig}}{p} \I + \frac{1}{4 p^2} \bigg( \sum_{k=1}^{K} \X_{k} \smash[t] {\hat{\Delt}}^{-1}_k \X_{k}^{T} \bigg)^2 \right)^{\frac{1}{2}} + \frac{1}{2p} \sum_{k=1}^{K} \X_{k} \smash[t] {\hat{\Delt}}^{-1}_k \X_{k}^{T} \right] \uu \\
&= \left[ \left( \frac{2 \lambda_{\Sig}}{p} + \frac{1}{4p^2} \gamma^2 \right)^{\frac{1}{2}} + \frac{1}{2p} \gamma \right] \uu \\
&= \left[ \frac{1}{2p} \left( \gamma + \sqrt{\gamma^2 + 8 \lambda_{\Sig} p} \right) \right] \uu, \label{addfrob_eig_eq2}
\end{align}
so $\uu$ is indeed an eigenvector of $\hat{\Sig}$ with the eigenvalue $\frac{1}{2p} \left( \gamma + \sqrt{\gamma^2 + 8 \lambda_{\Sig} p} \right)$.

Since the eigenvectors of $\hat{\Sig}$ and $\sum_{k=1}^{K} \X_k \smash[t] {\hat{\Delt}}^{-1}_k \X_k^T$ are equivalent and \eqref{addfrob_eig_eq2} gives us the exact connection between their eigenvalues, it follows that
\begin{align*}
\hat{\Sig} &= \U \left[ \frac{1}{2p} \left( \boldsymbol{\Gamma} + \left( \boldsymbol{\Gamma}^2 + 8 \lambda_{\Sig}p \I \right)^{\frac{1}{2}} \right) \right] \U^T,
\end{align*}
where $\U, \boldsymbol{\Gamma}$ are defined by the eigendecomposition $\sum_{k=1}^{K} \X_k \smash[t] {\hat{\Delt}}^{-1}_k \X_k^T = \U \boldsymbol{\Gamma} \U^T$.

This same logic can be used to show that for each $k = 1, \dots, K$,
\begin{align*}
\hat{\Delt}_k &= \V_k \left[ \frac{1}{2n} \left( \boldsymbol{\Phi}_k + \left( \boldsymbol{\Phi}_k^2 + 8 \lambda_{k} n \I \right)^{\frac{1}{2}} \right) \right] \V_k^T,
\end{align*}
where $\V_k, \boldsymbol{\Phi}_k$ are defined by the eigendecomposition $\X_k^T \smash[t] {\hat{\Sig}}^{-1} \X_k = \V_k \boldsymbol{\Phi}_k \V_k^T$. We omit the proof since it uses the same argument as above.

\vspace{6pt}
\noindent
$(\Leftarrow)$ Now suppose that $\hat{\Sig}$ and $\hat{\Delt}_1, \dots, \hat{\Delt}_K$ satisfy equations \eqref{prop:addfrob_eig_eq1} and \eqref{prop:addfrob_eig_eq2}. 

Since we know $\left[ \frac{1}{2p} \left( \gamma + \sqrt{\gamma^2 + 8 \lambda_{\Sig} p} \right) \right]^{-1} = \frac{1}{4 \lambda_{\Sig}} \left( \sqrt{\gamma^2 + 8 \lambda_{\Sig} p} - \gamma \right)$, it follows that $\left[ \frac{1}{2p} \left( \boldsymbol{\Gamma} + \left( \boldsymbol{\Gamma}^2 + 8 \lambda_{\Sig}p \I \right)^{\frac{1}{2}} \right) \right]^{-1} = \frac{1}{4 \lambda_{\Sig}} \left( \left( \boldsymbol{\Gamma}^2 + 8 \lambda_{\Sig} p \I \right)^{\frac{1}{2}} - \boldsymbol{\Gamma} \right)$.
Therefore,
\begin{align*}
p \hat{\Sig} - \sum_{k=1}^{K} \X_k \smash[t] {\hat{\Delt}}^{-1}_k \X_k^T - 2 \lambda_{\Sig} \smash[t] {\hat{\Sig}}^{-1} &= p \U \left[ \frac{1}{2p} \left( \boldsymbol{\Gamma} + \left( \boldsymbol{\Gamma}^2 + 8 \lambda_{\Sig}p \I \right)^{\frac{1}{2}} \right) \right] \U^T - \U \boldsymbol{\Gamma} \U^T \nonumber \\
&\qquad \qquad - 2 \lambda_{\Sig} \U \left[ \frac{1}{2p} \left( \boldsymbol{\Gamma} + \left( \boldsymbol{\Gamma}^2 +  8 \lambda_{\Sig}p \I \right)^{\frac{1}{2}} \right) \right]^{-1} \U^T  \\
&= \U \left[ \frac{1}{2} \left( \boldsymbol{\Gamma} + \left( \boldsymbol{\Gamma}^2 + 8 \lambda_{\Sig}p \I \right)^{\frac{1}{2}} \right) \right] \U^T - \U \boldsymbol{\Gamma} \U^T \nonumber \\
&\qquad \qquad - \U \left[ \frac{1}{2} \left( \left( \boldsymbol{\Gamma}^2 + 8 \lambda_{\Sig} p \I \right)^{\frac{1}{2}} - \boldsymbol{\Gamma} \right) \right] \U^T  \\
&= 0.
\end{align*}

Similarly, one can substitute in \eqref{prop:addfrob_eig_eq2} and follow the same argument to show that \eqref{add_grad_eq2} is also satisfied.
\end{proof}

As a consequence of Proposition~\ref{prop:addfrob_eig}, each update step in the additive Frobenius Flip-Flop algorithm can be solved by taking a full eigendecomposition and then regularizing the eigenvalues. This algorithm is given in Algorithm~\ref{alg:addfrob}.

\begin{algorithm}
\caption{Flip-Flop Algorithm for Additive Frobenius Penalized iPCA Estimators}\label{alg:addfrob}
\begin{algorithmic} [1]
\State Center the columns of $\X_1, \dots, \X_K$, and initialize $\hat{\Sig}$, $\hat{\Delt}_1, \dots, \hat{\Delt}_K$ to be positive definite.
\While{not converged}
\State Take eigendecomposition: $\sum_{k=1}^{K} \X_{k} \smash[t] {\hat{\Delt}}^{-1}_k \X_{k}^{T} = \U \boldsymbol{\Gamma} \U^T$~ \tikzmark{right1} \tikzmark{top1}
\State Regularize eigenvalues: $\Phi_{ii} = \frac{1}{2p}  \left( \Gamma_{ii} + \sqrt{\Gamma_{ii}^{2} + 8 \lambda_{\Sig} p} \right)$
\State Update $\smash[t] {\hat{\Sig}}^{-1} = \U \boldsymbol{\Phi}^{-1} \U^{T}$ \tikzmark{bottom1}

\For{$k = 1, \dots, K$}
\State Take eigendecomposition: $\X_{k}^{T}  \smash[t] {\hat{\Sig}}^{-1} \X_{k} = \V \boldsymbol{\Phi} \V^{T}$\tikzmark{top2}
\State Regularize eigenvalues: $\Gamma_{ii} = \frac{1}{2n} \left(\Phi_{ii} + \sqrt{\Phi_{ii}^{2} + 8 \lambda_{k} n} \right)$. 
\State Update $\smash[t] {\hat{\Delt}}^{-1}_k = \V \boldsymbol{\Gamma}^{-1} \V^{T}$ \tikzmark{bottom2}
\EndFor
\EndWhile
\vspace{-16pt}
\end{algorithmic}
\AddNote{top1}{bottom1}{right1}{Update $\Sig$}
\AddNote{top2}{bottom2}{right1}{Update $\Delt_k$}
\end{algorithm}

\subsubsection{Multiplicative Frobenius Penalized Flip-Flop Estimator} \label{sec:s_multfrob}
The derivation of the multiplicative Frobenius Flip-Flop algorithm is very similar to the previous derivation with the additive Frobenius penalty. Thus, we omit most of the details and simply provide a sketch of the derivation.

Recall that the multiplicative Frobenius penalized estimator solves
\begin{align}\label{multfrob}
\smash[t] {\hat{\Sig}}^{-1}, \smash[t] {\hat{\Delt}}^{-1}_1, \dots, \smash[t] {\hat{\Delt}}^{-1}_K &=  \smash{\argmax_{\substack{\Sig^{-1} \succ 0 \\ \Delt_1^{-1}, \dots, \Delt_K^{-1} \succ 0}}} \Big \{ p \log | \Sig^{-1} | + n \sum_{k=1}^{K} \log | \Delt_k^{-1} | - \sum_{k = 1}^{K} \mathrm{tr}\left( \Sig^{-1} \X_{k} \Delt_k^{-1} \X_{k}^{T} \right) \nonumber \\
& \qquad \qquad \qquad \qquad \qquad - \norm{\Sig^{-1}}_F^2 \sum_{k=1}^{K} \lambda_k \norm{\Delt_k^{-1}}_F^2 \Big \},
\end{align}
for which the gradient equations are
\begin{align}
p \hat{\Sig} - \sum_{k=1}^{K} \X_{k} \smash[t] {\hat{\Delt}}^{-1}_k \X_{k}^{T} - 2 \smash[t] {\hat{\Sig}}^{-1}  \sum_{k=1}^{K} \lambda_{k} \norm{\smash[t] {\hat{\Delt}}^{-1}_k}_{F}^{2} &= 0, \label{mult_grad_eq1} \\
n \hat{\Delt}_k - \X_{k}^{T} \smash[t] {\hat{\Sig}}^{-1} \X_{k} - 2 \lambda_{k} \smash[t] {\hat{\Delt}}^{-1}_k \norm{\smash[t] {\hat{\Sig}}^{-1}}_{F}^{2} &= 0 \qquad \forall \: k = 1, \dots, K. \label{mult_grad_eq2}
\end{align}

Assuming these gradient equations \eqref{mult_grad_eq1} and \eqref{mult_grad_eq2} are satisfied, we can write
\begin{align*}
\hat{\Sig} &= \left( \frac{2 \sum_{k=1}^{K} \lambda_{k} || \smash[t] {\hat{\Delt}}^{-1}_k ||_{F}^{2}}{p} \I + \frac{1}{4 p^2} \left( \sum_{k=1}^{K} \X_{k} \smash[t] {\hat{\Delt}}^{-1}_k \X_{k}^{T} \right)^2 \right)^\frac{1}{2} + \frac{1}{2p} \sum_{k=1}^{K} \X_{k} \smash[t] {\hat{\Delt}}^{-1}_k \X_{k}^{T}, \\
\hat{\Delt}_k &= \left( \frac{2 \lambda_{k} || \smash[t] {\hat{\Sig}}^{-1} ||_{F}^{2}}{n} \I + \frac{1}{4 n^2} \left( \X_{k}^{T} \smash[t] {\hat{\Sig}}^{-1} \X_{k} \right)^2 \right)^\frac{1}{2} + \frac{1}{2n} \X_{k}^{T} \smash[t] {\hat{\Sig}}^{-1} \X_{k} \qquad \forall \: k = 1, \dots, K.
\end{align*}

We then can follow the same argument in Proposition~\ref{prop:addfrob_eig} to show that $\hat{\Sig}$ and $\hat{\Delt}_1, \dots, \hat{\Delt}_K$ are solutions to the multiplicative Frobenius gradient equations \eqref{mult_grad_eq1} and \eqref{mult_grad_eq2} if and only if
\begin{align}
\hat{\Sig} &= \U \left[ \frac{1}{2p} \left( \boldsymbol{\Gamma} + \left( \boldsymbol{\Gamma}^2 + 8 p \sum_{k=1}^{K} \lambda_{k} || \smash[t] {\hat{\Delt}}^{-1}_k ||_{F}^{2} \I \right)^{\frac{1}{2}} \right) \right] \U^T \label{prop:multfrob_eig_eq1} \\
\text{and } ~ \hat{\Delt}_k &= \V_k \left[ \frac{1}{2n} \left( \boldsymbol{\Phi}_k + \left( \boldsymbol{\Phi}_k^2 + 8 n \lambda_{k} || \smash[t] {\hat{\Sig}}^{-1} ||_{F}^{2} \I \right)^{\frac{1}{2}} \right) \right] \V_k^T \qquad \forall \: k = 1, \dots, K, \label{prop:multfrob_eig_eq2}
\end{align}
where $\U, \V_k, \boldsymbol{\Gamma},$ and $\boldsymbol{\Phi}_k$ are defined by the eigendecompositions $\sum_{k=1}^{K} \X_k \smash[t] {\hat{\Delt}}^{-1}_k \X_k^T = \U \boldsymbol{\Gamma} \U^T$ and $\X_k^T \smash[t] {\hat{\Sig}}^{-1} \X_k = \V_k \boldsymbol{\Phi}_k \V_k^T$. This gives rise to the Flip Flop algorithm for solving the multiplicative Frobenius penalized estimators, as given in Algorithm~\ref{alg:multfrob}.

\subsubsection{Additive $L_1$ Penalized Flip-Flop Estimator} \label{sec:s_l1_ff}

If the inverse covariance matrices are known to have a sparse underlying structure, then we can apply an $L_1$ penalty, rather than the Frobenius penalty, to induce this sparsity. Note that while it is possible to use a multiplicative $L_1$ penalty, the multiplicative $L_1$ penalty is not geodesically convex and is not separable (as defined in \citet{tseng2001convergence}). Thus, we primarily consider the additive $L_1$ penalized iPCA estimator:
\begin{align}\label{glasso_off}
\smash[t] {\hat{\Sig}}^{-1}, \smash[t] {\hat{\Delt}}^{-1}_1, \dots, \smash[t] {\hat{\Delt}}^{-1}_K &=  \smash{\argmax_{\substack{\Sig^{-1} \succ 0 \\ \Delt_1^{-1}, \dots, \Delt_K^{-1} \succ 0}}} \Big \{ p \log | \Sig^{-1} | + n \sum_{k=1}^{K} \log | \Delt_k^{-1} | - \sum_{k = 1}^{K} \mathrm{tr}\left( \Sig^{-1} \X_{k} \Delt_k^{-1} \X_{k}^{T} \right) \nonumber \\
& \qquad \qquad \qquad \qquad \qquad - \lambda_{\Sig} \norm{\Sig^{-1}}_{1,\text{off}} - \sum_{k=1}^{K} \lambda_k \norm{\Delt_k^{-1}}_{1,\text{off}}\Big \}.
\end{align}
Note that $\norm{\cdot}_{1,\text{off}}$ penalizes the off-diagonal entries (i.e. $\norm{\A}_{1,\text{off}} = \sum_{i \neq j} | a_{ij} |$), but it is also possible to use the ordinary $L_1$ norm $\norm{\cdot}_1$.

For fixed $\hat{\Delt}_1, \dots, \hat{\Delt}_K$, the Flip-Flop update for $\hat{\Sig}$ is seen to be
\begin{align*}
\smash[t] {\hat{\Sig}}^{-1} &= \argmin_{\Sig^{-1} \succ 0} -p \log | \Sig^{-1} | +  \mathrm{tr}\left( \Sig^{-1} \left(\sum_{k = 1}^{K} \X_{k} \smash[t] {\hat{\Delt}}^{-1}_k \X_{k}^{T} \right)\right) + \lambda_{\Sig} \norm{\Sig^{-1}}_{1,\text{off}},
\end{align*}
and similarly for fixed $\hat{\Sig}$ and $\hat{\Delt}_j$, $j \neq k$, the update for $\hat{\Delt}_k$ is
\begin{align*}
\smash[t] {\hat{\Delt}}^{-1}_k &= \argmin_{\Delt_k^{-1} \succ 0} - n \log | \Delt_k^{-1} | + \mathrm{tr}\left( \Delt_k^{-1} \X_{k}^{T} \smash[t] {\hat{\Sig}}^{-1} \X_{k} \right) + \lambda_k \norm{\Delt_k^{-1}}_{1,\text{off}}.
\end{align*}
Both of which can be solved via the graphical lasso algorithm \citep{hsieh2011quic}. Plugging in these updates to the framework laid out in Algorithm~\ref{alg:ff}, we give the additive $L_1$ Flip-Flop algorithm in Algorithm~\ref{alg:glasso_off}.

\begin{algorithm}
\caption{Flip-Flop Algorithm for Additive $L_1$ Penalized iPCA Covariance Estimators}\label{alg:glasso_off}
\begin{algorithmic} [1]
\State Center the columns of $\X_1, \dots, \X_K$, and initialize $\hat{\Sig}$, $\hat{\Delt}_1, \dots, \hat{\Delt}_K$ to be positive definite.
\While{not converged}
\State Put $\A = \frac{1}{p}\sum_{k=1}^{K} \X_k \Delt_k^{-1} \X_k^T$ \tikzmark{top1}
\State Apply graphical lasso: 

\setlength\parindent{2cm} $\displaystyle \smash[t] {\hat{\Sig}}^{-1} = \argmin_{\Sig^{-1}} ~ -\log  |\Sig^{-1} | + \mathrm{tr} ( \A \Sig^{-1} ) + \frac{\lambda_{\Sig}}{p} \norm{\Sig^{-1}}_{1, \text{off}}$ \tikzmark{bottom1} 

\For{$k = 1, \dots, K$}
\State Put $\A_k := \frac{1}{n} \X_k^T \Sig^{-1} \X_k$ \tikzmark{top2}
\State Apply graphical lasso:
 
\setlength\parindent{2cm} $\displaystyle \smash[t] {\hat{\Delt}}^{-1}_k = \argmin_{\Delt_k^{-1}} ~ -\log | \Delt_k^{-1} | + \mathrm{tr} (\A_k \Delt_k^{-1}) + \frac{\lambda_{k}}{n} \norm{\Delt_k^{-1}}_{1, \text{off}}$ \tikzmark{bottom2} \tikzmark{right1}
\EndFor
\EndWhile \vspace{-12pt}
\end{algorithmic}
\AddNote{top1}{bottom1}{right1}{Update $\Sig$}
\AddNote{top2}{bottom2}{right1}{Update $\Delt_k$}
\end{algorithm}

\section{Geodesic Convexity and the Multiplciative Frobenius iPCA Estimator} \label{sec:s_geodesic}

Because of the central role that geodesic convexity plays in the multiplicative Frobenius iPCA estimator, we give a self-contained introduction to geodesic convexity in Appendix~\ref{sec:s_geodesic_intro}. This review provides the necessary concepts and tools to prove the theorems in Appendix~\ref{sec:s_multfrob_conv}. For a more comprehensive review, we refer to \citet{rapcsak1991gconvex} and \citet{vishnoi2018gconvex}.

\subsection{Introduction to Geodesic Convexity} \label{sec:s_geodesic_intro}

Convex optimization problems arise frequently in a variety of machine learning applications such as regression, matrix completion, and clustering, to name a few examples. Beyond its widespread applications, convex problems are well-understood theoretically and can be reliably solved in polynomial-time. As a result, machine learning tasks are often formulated as convex problems in Euclidean space to guarantee fast convergence to a global solution. Convexity, however, is not limited to Euclidean spaces. Many tools which we know and love from convex optimization can be extended to geodesic convexity (g-convexity) on Riemannian manifolds. In this general Riemannian setting, there are several applications in which we have g-convexity but not convexity \citep{zhang2016first}.

Before formally defining geodesic convexity, we first recall some useful concepts from metric geometry. A \textit{metric space} is a pair $(X,d)$ of a set $X$ and a distance function $d$ that satisfies positivity, symmetry, and the triangle inequality. A \textit{path} $\gamma$ is a continuous mapping from $[0,1]$ to $X$, and the \textit{length} $\ell$ of a path $\gamma$ is defined as $\ell(\gamma) := \sup \{\sum_{i=1}^{n} d(\gamma(t_{i-1}), \gamma(t_i)) \: : \: 0 = t_0 < \ldots < t_n = 1, \: n \in \mathbb{N} \}$. A metric space is a \textit{length space} if $d(x,y) = \inf \ell(\gamma)$ where the infimum is taken over all paths $\gamma : [0,1] \rightarrow X$ joining $x$ and $y$. 

\begin{definition}
Let $(X,d)$ be a length space. A path $\gamma: [0,1] \rightarrow X$ is a \textbf{geodesic} if for every $t \in [0,1]$ there exists an interval $[a,b] \subset [0,1]$ which contains a neighborhood of $t$ and $\gamma |_{[a,b]}$ is a shortest path from $\gamma(a)$ to $\gamma(b)$. Put simply, a geodesic is a path which locally minimizes length.
\end{definition}

Note that geodesics minimize length locally, but not globally. A canonical example of geodesics are the great circles on a sphere.

This concept of a geodesic generalizes the notion of a line in Euclidean space to general (nonlinear) length spaces. By replacing lines by geodesics in the definition of convexity, we can extend convexity to g-convexity in a straightforward manner.

\begin{definition}
Let $\mathcal{M}$ be a Riemannian manifold. A function $f: \mathcal{M} \rightarrow \mathbb{R}$ is \textbf{geodesically convex} if for any $x,y \in \mathcal{M}$, geodesic $\gamma$ such that $\gamma (0) = x$ and $\gamma (1) = y$, and $t \in [0,1],$ it holds that
\begin{align*}
    f(\gamma (t)) \leq (1-t) f(x) + tf(y).
\end{align*}
\end{definition}

The only caveat is that the underlying space must be a Riemannian manifold. To prevent a long winded detour into the details of Riemannian geometry, we avoid the full technical definition and simply think of a Riemannian manifold as a real differentiable manifold equipped with the notion of an inner product. We need the structure of a Riemannian manifold in order to meaningfully perform algebraic operations in our space. 

In summary, by extending the notion of a line to a geodesic, we can easily translate the notion of convexity on a Euclidean space to geodesic convexity on a Riemannian manifold. We next give a series of known properties regarding geodesic convexity that will be of use in the following proofs \citep{wiesel2012geodesic, vishnoi2018gconvex}. Many of these properties are analogous to the familiar convex setting. 

\begin{theorem}
Any local minima of a geodesically convex function is a global minima.
\end{theorem}

\begin{proposition}\label{prop:gconvex_operations}
The following operations preserve geodesic convexity:
\begin{enumerate}[label=(\roman*)]
\item \textit{(Addition)} If $f$ and $g$ are g-convex functions, then $f+g$ is g-convex.
\item \textit{(Kronecker Products)} Suppose $f$ is a real-valued, g-convex function on $\mathbb{S}^d_{++}$, and $\Q_j \in \mathbb{S}^{d_j}_{++}$ for each $j = 1, \dots, J$ such that $\prod_{j=1}^{J}d_j = d$. Then
\begin{align*}
g(\Q_1, \dots, \Q_K) = f(\Q_1 \otimes \dots \otimes \Q_K)
\end{align*}
is jointly g-convex in $\{ \Q_j \}_{j=1}^{J} \in \mathbb{S}^{d_1}_{++} \times \dots \times \mathbb{S}^{d_K}_{++}$.
\end{enumerate}
\end{proposition}

The proof of Proposition~\ref{prop:gconvex_operations}(i) is straightforward from the definition of g-convex functions, and Proposition~\ref{prop:gconvex_operations}(ii) is proved in \citet{wiesel2012geodesic}.

\begin{remark}
(Example 4.9 in \citet{vishnoi2018gconvex}) Consider the manifold of positive definite matrices $\mathbb{S}_{++}^n$. For $\Q_0, \Q_1 \in \mathbb{S}_{++}^n$, the geodesic joining $\Q_0$ to $\Q_1$ can be parameterized as
\begin{align}\label{eq:geodesic}
\Q_t = \Q_0^{\frac{1}{2}} \left( \Q_0^{-\frac{1}{2}}  \Q_1 \Q_0^{-\frac{1}{2}}  \right)^t \Q_0^{\frac{1}{2}} \qquad \forall \: t \in [0,1].
\end{align}
\end{remark}

\subsection{Global Convergence of the Multiplicative Frobenius iPCA Estimator} \label{sec:s_multfrob_conv}

We now have the tools to start proving global convergence of the multiplicative iPCA estimator. The roadmap of this proofs is as follows. First, we prove that the negative log-likelihood is g-convex. Then, we prove that the multiplicative Frobenius penalty is g-convex. Since the sum of g-convex functions is g-convex, this implies that the multiplicative iPCA estimator is a g-convex optimization problem. Furthermore, since the Flip-Flop algorithm was proven to converge in Theorem~\ref{converge_to_stationary}, it follows that the multiplicative Frobenius iPCA estimator converges to the global solution as a consequence of geodesic convexity!

Without loss of generality, we assume that each data set $\X_k$ has been column-centered for the remainder of this section.

\begin{lemma} \label{lemma_ll_gconvex}
The negative log-likelihood of our model, \eqref{ff_loglike}, is jointly
geodesically convex in $\Sig^{-1}$, $\Delt_1^{-1}, \ldots \Delt_K^{-1}$ with respect to $\mathbb{S}_{++}^{n} \times \mathbb{S}_{++}^{p_1} \times \dots \times \mathbb{S}_{++}^{p_K}$. 
\end{lemma}
\begin{proof}
Let $\tilde{\X} = [\X_1, \ldots, \X_K] \in \mathbb{R}^{n \times p}$ and define
\[ \tilde{\Delt} = \begin{bmatrix}
\Delt_1 &    0    & \cdots & 0 \\
   0    & \Delt_2 & \cdots & 0 \\
\vdots  & \vdots  & \ddots & \vdots\\
   0    &    0    & \cdots & \Delt_K
\end{bmatrix} \]
Since $\tilde{\Delt}$ is a block diagonal matrix, then $| \tilde{\Delt} | = \prod_{k = 1}^{K} | \Delt_k |$ and $| \smash[t] {\tilde{\Delt}}^{-1} | = \prod_{k = 1}^{K} | \Delt_k^{-1} |$. This implies that
\begin{align}
    \sum_{k = 1}^{K} \log | \Delt_k^{-1} | = \mathrm{log} (\prod_{k = 1}^{K} | \Delt_k^{-1} |) = \log | \smash[t] {\tilde{\Delt}}^{-1} |. \label{lemma_ll_gconvex_eq1}
\end{align}

Also, using this new parameterization,
\begin{align}
\sum_{k = 1}^{K} \mathrm{tr}(\X_k \Delt_k^{-1} \X_k^T \Sig^{-1}) = \mathrm{tr}(\tilde{\X} \smash[t] {\tilde{\Delt}}^{-1} \tilde{\X}^T \Sig^{-1}) = \mathrm{vec}(\tilde{\X})^T (\smash[t] {\tilde{\Delt}}^{-1} \otimes \Sig^{-1})\: \mathrm{vec}(\tilde{\X}). \label{lemma_ll_gconvex_eq2}
\end{align}

Using \eqref{lemma_ll_gconvex_eq1} and \eqref{lemma_ll_gconvex_eq2}, we rewrite the negative log-likelihood as
\begin{align*}
-\ell(\Sig^{-1}, \Delt_1^{-1}, \ldots, \Delt_K^{-1})  &\propto
-p \log | \Sig^{-1} | - n
\sum_{k=1}^{K} \log | \Delt_k^{-1} | + \sum_{k=1}^{K} \mathrm{tr}(\X_{k} \Delt_k^{-1} \X_{k}^{T} \Sig^{-1}) \\
&= -p \log | \Sig^{-1} | - n \log | \smash[t] {\tilde{\Delt}}^{-1} | + \mathrm{vec}(\tilde{\X})^T (\smash[t] {\tilde{\Delt}}^{-1} \otimes \Sig^{-1}) \: \mathrm{vec}(\tilde{\X}) \\
&= -\mathrm{log} ( | \Sig^{-1} |^p | \smash[t] {\tilde{\Delt}}^{-1} |^n ) + \mathrm{vec}(\tilde{\X})^T (\smash[t] {\tilde{\Delt}}^{-1} \otimes \Sig^{-1}) \: \mathrm{vec}(\tilde{\X}) \\
&= \mathrm{log} (| \smash[t] {\tilde{\Delt}}^{-1} \otimes \Sig^{-1} |^{-1} ) + \mathrm{vec}(\tilde{\X})^T (\smash[t] {\tilde{\Delt}}^{-1} \otimes \Sig^{-1}) \: \mathrm{vec}(\tilde{\X}).
\end{align*}

Next, consider the manifold $\mathbb{S}_{++}^{np}$, and let the function $f: \mathbb{S}_{++}^{np} \rightarrow \mathbb{R}$ be given by
\begin{align*}
    f(\Q) = \log(| \Q |^{-1}) + \mathrm{vec}(\tilde{\X})^T \Q \mathrm{vec}(\tilde{\X}).
\end{align*}

Since $f$ is the sum of geodesically convex functions on $\mathbb{S}_{++}^{np}$ \citep{wiesel2012geodesic}, and $-\ell(\Sig^{-1}, \smash[t] {\tilde{\Delt}}^{-1}) = f(\smash[t] {\tilde{\Delt}}^{-1} \otimes \Sig^{-1})$, then by Proposition~\ref{prop:gconvex_operations}, the negative log-likelihood is jointly geodesically convex in $\Sig^{-1}, \Delt_1^{-1}, \dots, \Delt_K^{-1}$ with respect to $\mathbb{S}_{++}^{n} \times \mathbb{S}_{++}^{p_1} \times \dots \times \mathbb{S}_{++}^{p_K}$.
\end{proof}

\begin{corollary} \label{unpenalized_conv}
Suppose the unpenalized log-likelihood in \eqref{s_ff_loglike} is bounded. If the Flip-Flop estimators for the unpenalized log-likelihood exist, then they converge to the global solution of \eqref{s_ff_loglike}.
\end{corollary}
\begin{proof}

Recall that the Flip-Flop algorithm yields the iterates:
\begin{enumerate}
    \item $\hat{\Sig} = \frac{1}{p} \sum_{k=1}^{K} \X_{k} \smash[t] {\hat{\Delt}}^{-1}_k \X_{k}^{T} = \argmax_{\Sig} \ell(\Sig, \hat{\Delt}_1, \ldots, \hat{\Delt}_K)$
    \item For each $k = 1, \dots, K$, \\
    $\hat{\Delt}_k = \frac{1}{n} \X_{k}^{T} (\hat{\Sig})^{-1} \X_{k} = \argmax_{\Delt_k} \ell(\hat{\Sig}, \hat{\Delt}_1, \ldots, \hat{\Delt}_{k-1}, \Delt_k, \hat{\Delt}_{k+1}, \ldots, \hat{\Delt}_K)$
\end{enumerate}
Thus, each update of the Flip-Flop algorithm monotonically increases the log-likelihood $\ell$. Assuming that the MLEs exist and are bounded, this Flip-Flop algorithm converges to a local maximum. Furthermore, since $\ell$ is jointly geodesically concave in $\Sig^{-1}$, $\Delt_1^{-1}, \ldots \Delt_K^{-1}$, then all local maxima are global maxima so that the Flip-Flop estimators for the unpenalized log-likelihood converge to the global solution.
\end{proof}

\begin{lemma}\label{lemma_multfrob_geodesic}
The multiplicative Frobenius norm penalty, $\norm{\Sig^{-1}}_{F}^{2}
\sum_{k=1}^{K}  \lambda_{k} \norm{\Delt_k^{-1}}_{F}^{2}$, is jointly
geodesically convex in $\Sig^{-1}$, $\Delt_1^{-1}, \ldots, \Delt_K^{-1}$ with respect to $\mathbb{S}_{++}^{n} \times \mathbb{S}_{++}^{p_1} \times \dots \times \mathbb{S}_{++}^{p_K}$. 
\end{lemma}
\begin{proof}
Since $\mathrm{tr}(\A \otimes \mathbf{B}) = \mathrm{tr}(\A) \mathrm{tr}(\mathbf{B})$ and $\norm{\A}_F^2 = \mathrm{tr}(\A^T \A)$, then 
\begin{align}
    P(\Sig^{-1}, \Delt_1^{-1}, \dots, \Delt_K^{-1}) &:= \norm{\Sig^{-1}}_{F}^{2}\sum_{k=1}^{K}  \lambda_{k} \norm{\Delt_K^{-1}}_{F}^{2}\\ 
    &= \sum_{k=1}^{K}\lambda_k \mathrm{tr}(\Sig^{-2}) \mathrm{tr}(\Delt_k^{-2}) \\
    &= \sum_{k = 1}^{K} \lambda_k \mathrm{tr}(\Delt_k^{-2} \otimes \Sig^{-2}) \\
    &= \sum_{k=1}^{K} \lambda_k \norm{\Delt_k^{-1} \otimes \Sig^{-1}}_F^2 \\
    &= \sum_{k=1}^{K} \norm{(\sqrt{\lambda_k}\Delt_k^{-1}) \otimes \Sig^{-1}}_F^2 \\
    &= \norm{\bar{\Delt}^{-1} \otimes \Sig^{-1}}_F^2, \label{multfrob_geodesic_eq1}
\end{align}
where
\[ \bar{\Delt} = \begin{bmatrix}
\frac{1}{\sqrt{\lambda_1}} \Delt_1 &    0    & \cdots & 0 \\
   0    & \frac{1}{\sqrt{\lambda_2}} \Delt_2 & \cdots & 0 \\
\vdots  & \vdots  & \ddots & \vdots\\
   0    &    0    & \cdots & \frac{1}{\sqrt{\lambda_K}} \Delt_K
\end{bmatrix}. \]

We will next show that the function $f: \mathbb{S}_{++}^{np} \rightarrow \mathbb{R}$ defined by 
\begin{align} \label{multfrob_geodesic_eq2}
    f(\Q^\alpha) := \norm{\Q^\alpha}_F^2 = \mathrm{tr}((\Q^\alpha)^T \Q^\alpha) = \mathrm{tr}((\Q^{\alpha})^2), \qquad \alpha \in \{\pm 1\}
\end{align}
 is geodesically convex in $\Q \in \mathbb{S}_{++}^{np}$. That is, let $\Q_0, \Q_1 \in \mathbb{S}_{++}^{np}$ be given, and let $\Q_t$ be the geodesic between $\Q_0$ and $\Q_1$ as given in \eqref{eq:geodesic}. We want to show that $f(\Q^{\alpha}_t)$ is a convex function with respect to $t$.

So consider the eigendecomposition $\Q_0^{-\frac{1}{2}} \Q_1 \Q_0^{-\frac{1}{2}} = \U \D \U^T$, where $\U$ is an orthogonal matrix and $\D$ is a diagonal matrix with diagonal entries $d_i$. Then for all $t \in [0,1]$, it follows from \eqref{multfrob_geodesic_eq2} that
 \begin{align}
 f(\Q_t^{\alpha}) &= \mathrm{tr}(\Q^{\alpha}_t \Q^{\alpha}_t) \\
&= \mathrm{tr}(\Q_0^{\frac{\alpha}{2}} ( \Q_0^{-\frac{1}{2}} \Q_1 \Q_0^{-\frac{1}{2}})^{\alpha t} \Q_0^\alpha ( \Q_0^{-\frac{1}{2}} \Q_1 \Q_0^{-\frac{1}{2}})^{\alpha t} \Q_0^{\frac{\alpha}{2}}) \\
&= \mathrm{tr}(\Q_0^{\frac{\alpha}{2}} ( \U \D \U^T )^{\alpha t} \Q_0^\alpha ( \U \D \U^T )^{\alpha t} \Q_0^{\frac{\alpha}{2}}) \\
&= \mathrm{tr}(\Q_0^{\frac{\alpha}{2}} \U \D^{\alpha t} \U^T  \Q_0^\alpha \U \D^{\alpha t} \U^T \Q_0^{\frac{\alpha}{2}}) \\
&= \mathrm{tr}(\D^{\alpha t} \U^T  \Q_0^\alpha \U \D^{\alpha t} \U^T \Q_0^\alpha \U) \\
&= \mathrm{tr}(\D^{\alpha t} \A \D^{\alpha t} \A), \label{multfrob_geodesic_eq3}
\end{align}
where $\A := \U^T \Q_0^\alpha \U$.

Because $\A$ is symmetric and $(\D^{\alpha t} \A)_{ij} = d_i^{\alpha t} a_{ij}$, then
\begin{align*}
    (\D^{\alpha t} \A \D^{\alpha t} \A)_{ij} = \sum_{l = 1}^{np}( (\D^{\alpha t} \A)_{il} (\D^{\alpha t} &\A)_{lj} ) = \sum_{l = 1}^{np} ( d_i^{\alpha t} a_{il} d_l^{\alpha t} a_{lj} ) = \sum_{l = 1}^{np} ( d_i^{\alpha t} a_{il} d_l^{\alpha t} a_{jl} ).
\end{align*}
Plugging this into \eqref{multfrob_geodesic_eq3} gives
\begin{align*}
f(\Q_t^{\alpha}) &= \sum_{i = 1}^{np} ( \D^{\alpha t} \A \D^{\alpha t} \A )_{ii} = \sum_{i = 1}^{np} \sum_{l = 1}^{np} ( d_i^{\alpha t} a_{il} d_l^{\alpha t} a_{il} ) = \sum_{i = 1}^{np} \sum_{l = 1}^{np} ( a_{il}^2 (d_i d_l)^{\alpha t} ).
\end{align*}

Note that since $\Q_0$ and $\Q_1$ are positive definite, then $\Q_0^{-\frac{1}{2}} \Q_1 \Q_0^{-\frac{1}{2}}$ is positive definite. Thus, $d_i > 0$ for each $i = 1, \dots, np$, and because $d_i > 0$ for each $i = 1, \dots, np$, $f(\Q_t^{\alpha}) = \sum_{i = 1}^{np} \sum_{l = 1}^{np} ( a_{il}^2 (d_i d_l)^{\alpha t} )$ is a convex function in $t$. This proves $f$ is g-convex in $\Q \in \mathbb{S}_{++}^{np}$.

Since $f$ is g-convex in $\Q \in \mathbb{S}_{++}^{np}$ and $P(\Sig^{-1}, \bar{\Delt}^{-1}) = f(\bar{\Delt}^{-1} \otimes \Sig^{-1})$ from \eqref{multfrob_geodesic_eq1}, then the multiplicative Frobenius norm penalty $P$ is g-convex in $\Sig^{-1}, \Delt_1^{-1}, \dots, \Delt_K^{-1}$ by Proposition~\ref{prop:gconvex_operations}(ii).
\end{proof}

\globalconvergence*
\begin{proof}
From Lemma~\ref{lemma_ll_gconvex} and Lemma~\ref{lemma_multfrob_geodesic}, we see that the objective function corresponding to the multiplicative Frobenius penalized estimator in \eqref{multfrob} is the sum of jointly g-convex functions. Thus, the multiplicative Frobenius penalized estimator is also jointly geodesically convex in $\Sig^{-1}$ and $\Delt_1^{-1}, \ldots, \Delt_K^{-1}$.

Recall we have already proved that Algorithm~\ref{alg:multfrob} converges to a stationary point in Theorem~\ref{converge_to_stationary}. Since the multiplicative Frobenius iPCA estimator is geodesically convex, then all local optima are global optima, and the Flip-Flop algorithm for the multiplicative Frobenius iPCA estimator converges to the global solution of \eqref{multfrob}.
\end{proof}

\begin{remark}
The multiplicative Frobenius iPCA estimator is unique in the sense that the Kronecker production solution $\hat{\Sig} \otimes \hat{\Delt}_k$ is unique.
\end{remark}

\section{Equivalence between PCA and iPCA with Frobenius Penalties when $K = 1$}\label{sec:s_equiv}

One major reason for using the Frobenius penalties is that in the case where we observe only one data set, iPCA with the Frobenius penalties and the classical PCA are equivalent in the sense that the PC scores and loadings are the same. Throughout this section, we will assume $K = 1$, and we let the SVD of $\X$ (which has been column-centered) be given by $\X = \U \D \V^T$, where $\U \in \mathbb{R}^{n \times n}$, $\D \in \mathbb{R}^{n \times p}$ with the diagonal elements $d_i$, $\V \in \mathbb{R}^{p \times p}$, and $r = \mathrm{rank}(\X) < \min \{ n, p\}$. Without loss of generality, suppose also that $p \geq n$. In PCA, it is well known that the PC scores are given by the columns of $\U$ and the PC loadings are given by the columns of $\V$. We will show that iPCA with the Frobenius penalties also yield the same PC scores $\U$ and loadings $\V$.

Let us first consider iPCA with the additive Frobenius penalty $\lambda_{\Sig} \norm{\Sig^{-1}}_F^2 + \lambda_{\Delt} \norm{\Delt^{-1}}_F^2$. By Theorem 1 in \citet{allen2010transposable}, there is a unique global solution maximizing the matrix-variate normal log-likelihood with additive Frobenius penalties \eqref{addfrob} when $K = 1$. This global solution is given by
\begin{align*}
\hat{\Sig} = \U \boldsymbol{\beta} \U^T \quad \mathrm{and}  \quad \hat{\Delt} = \V \boldsymbol{\theta} \V^T,
\end{align*}
where $\boldsymbol{\beta} = \mathrm{diag}(\beta_1, \dots, \beta_n)$ and $\boldsymbol{\theta} = \mathrm{diag}(\theta_1, \dots, \theta_p)$ are defined as follows:
\begin{align*}
\beta_i &= \begin{dcases}
\frac{d_i^2 \theta_i}{n \theta_i^2 - 2 \lambda_{\Delt}}, & i = 1, \dots, r \\
\sqrt{\frac{2 \lambda_{\Sig}}{p}}, & i = r+1, \dots, n
\end{dcases} \\
\theta_i &= \begin{dcases}
\sqrt{\frac{-c_{2, i} - \sqrt{c_{2, i}^2 - 4 c_1 c_{3, i}}}{2 c_1}}, & i = 1, \dots, r \\
\sqrt{\frac{2 \lambda_{\Delt}}{n}}, & i = r+1, \dots, p,
\end{dcases}
\end{align*}
with coefficients
\begin{align*}
c_1 &= - 2 \lambda_{\Sig} n^2, \\
c_{2, i} &= d_i^4 (p - n) + 8 n \lambda_{\Sig} \lambda_{\Delt}, \\
c_{3, i} &= 2 \lambda_{\Delt} (d_i^4 - 4 \lambda_{\Sig} \lambda_{\Delt}).
\end{align*}

Since the PC scores and loadings from iPCA are the eigenvectors of $\hat{\Sig}$ and $\hat{\Delt}$, respectively, the above result shows that the PC scores and loadings from iPCA with the additive Frobenius penalty are precisely $\U$ and $\V$, as in PCA.

We next investigate the equivalence of iPCA with the multiplicative Frobenius penalty $\lambda \norm{\Sig^{-1}}_F^2 \norm{\Delt^{-1}}_F^2$ and PCA when $K = 1$. While we can follow a similar argument as \citet{allen2010transposable}, the proof is complicated by the non-separable penalty terms. That is, each term in the multiplicative penalty depends on both $\hat{\Sig}$ and $\hat{\Delt}$. We will return to this complication in the proof of the following theorem.

\begin{theorem}\label{thm:multfrob1}
In the case when $K = 1$, the solution to iPCA with the multiplicative Frobenius penalty \eqref{multfrob} is given by $\hat{\Sig} = \U \boldsymbol{\beta} \U^T$ and $\hat{\Delt} = \V \boldsymbol{\theta} \V^T$, where $\boldsymbol{\beta} = \mathrm{diag}(\beta_1, \dots, \beta_n)$ and $\boldsymbol{\theta} = \mathrm{diag}(\theta_1, \dots, \theta_p)$ satisfy the system
\begin{align*}
\begin{dcases}
\theta_i = \sqrt{\frac{-c_{2, i}(T) - \sqrt{c_{2, i}^2(T) - 4 c_1(T) c_{3, i}(T)}}{2 c_1(T)}}, & i = 1, \dots, r \\
\theta_i = \sqrt{\frac{2 \lambda}{n}}, & i = r + 1, \dots, p \\
\beta_i = \frac{d_i^2 \theta_i}{n \theta_i^2 - 2 \lambda}, & i = 1, \dots, r \\
\beta_i = \sqrt{\frac{2 \lambda T}{p}}, & i = r+1, \dots, n \\
T = \sum_{k = 1}^{r} \theta_k^{-2} + (p - r) \frac{n}{2 \lambda},
\end{dcases}
\end{align*}
with $c_1(T) = -2 \lambda n^2 T$, $c_{2, i}(T) = d_i^4 (p - n) + 8 n \lambda^2 T$, and $c_{3, i}(T) = 2 \lambda (d_i^4 - 4 \lambda^2 T)$. This solution exists and is unique (up to a scaling factor).
\end{theorem}

\begin{proof}
As seen previously, the gradient equations from the matrix-variate normal log-likelihood with the multiplicative Frobenius penalty are
\begin{align*}
p \hat{\Sig} - 2 \lambda \hat{\Sig}^{-1} \norm{\hat{\Delt}^{-1}}_F^2 &= \X \hat{\Delt}^{-1} \X^T  \\
n \hat{\Delt} - 2 \lambda \hat{\Delt}^{-1} \norm{\hat{\Sig}^{-1}}_F^2 &= \X^T \hat{\Sig}^{-1} \X.
\end{align*}

Thus, the eigenvectors of $\hat{\Sig}$ and $\hat{\Delt}$ are equal to their respective quadratic forms. It follows that there is only one solution for the eigenvectors - namely, the left and right singular vectors of $\X = \U \D \V^T$. [Note: the last $n - r$ eigenvectors of $\U$ and the last $p - r$ eigenvectors of $\V$ are not unique.]

Put $\hat{\Sig} = \U \boldsymbol{\beta} \U^T$ and $\hat{\Delt} = \V \boldsymbol{\theta} \V^T$, where $\boldsymbol{\beta} = \mathrm{diag}(\beta_1, \dots, \beta_n)$ and $\boldsymbol{\theta} = \mathrm{diag}(\theta_1, \dots, \theta_p)$. Note that we have the implicit constraints $\beta_i > 0$ and $\theta_i > 0$ for each $i$ due to the positive definiteness of the covariances. Plugging in these decompositions and the SVD of $\X$ into the gradient equations gives
\begin{align*}
p \U \boldsymbol{\beta} \U^T - 2 \lambda \U \boldsymbol{\beta}^{-1} \U^T \norm{\V \boldsymbol{\theta}^{-1} \V^T}_F^2 &= \U \D \V^T \V \boldsymbol{\theta}^{-1} \V^T \V \D^T \U^T  \\
n \V \boldsymbol{\theta} \V^T - 2 \lambda \V \boldsymbol{\theta}^{-1} \V^T \norm{\U \boldsymbol{\beta}^{-1} \U^T}_F^2 &= \V \D^T \U^T \U \boldsymbol{\beta}^{-1} \U^T \U \D \V^T,
\end{align*}
or equivalently, using the orthogonality of $\U$ and $\V$,
\begin{align*}
p \boldsymbol{\beta} - 2 \lambda \boldsymbol{\beta}^{-1} \norm{ \boldsymbol{\theta}^{-1}}_F^2 &= \D \boldsymbol{\theta}^{-1} \D^T \\
n \boldsymbol{\theta} - 2 \lambda \boldsymbol{\theta}^{-1} \norm{ \boldsymbol{\beta}^{-1}}_F^2 &= \D^T \boldsymbol{\beta}^{-1} \D.
\end{align*}

We can write this element-wise as the following system of equations:
\begin{align}
p \beta_i - 2 \lambda \beta_i^{-1} \sum_{k = 1}^{p} \theta_k^{-2} &= d_i^2 \theta_i^{-1}, & i = 1, \dots, r \label{eq:multfrob_equiv_eq1} \\
p \beta_i - 2 \lambda \beta_i^{-1} \sum_{k = 1}^{p} \theta_k^{-2} &= 0, & i = r+1, \dots, n \label{eq:multfrob_equiv_eq2}\\
n \theta_i - 2 \lambda \theta_i^{-1} \sum_{k = 1}^{n} \beta_k^{-2} &= d_i^2 \beta_i^{-1}, & i = 1, \dots, r \label{eq:multfrob_equiv_eq3}\\
n \theta_i - 2 \lambda \theta_i^{-1} \sum_{k = 1}^{n} \beta_k^{-2} &= 0, & i = r+1, \dots, p. \label{eq:multfrob_equiv_eq4}
\end{align}

Now, to simplify this system of equations, we notice that if $(\boldsymbol{\beta}, \boldsymbol{\theta})$ solve this system, then for any positive scalar $c$, $(c\boldsymbol{\beta}, c^{-1}\boldsymbol{\theta})$ is also a solution. Without loss of generality, we may thus assume that $\boldsymbol{\beta}$ is normalized so that $\sum_{k = 1}^{n} \beta_k^{-2} = 1$.

It immediately follows from \eqref{eq:multfrob_equiv_eq4} that $\theta_i = \sqrt{\frac{2 \lambda}{n}}$ for $i = r + 1, \dots, p$. 

On the other hand, for $i = 1, \dots, r$, \eqref{eq:multfrob_equiv_eq3} gives us that
\begin{align*}
& n \theta_i^2 \beta_i - 2 \lambda \beta_i = d_i^2 \theta_i \\
\implies & \beta_i = \frac{d_i^2 \theta_i}{n \theta_i^2 - 2 \lambda}.
\end{align*}
Substituting this equation for $\beta_i$ into \eqref{eq:multfrob_equiv_eq1} then yields
\begin{align*}
& \frac{p d_i^2 \theta_i}{n \theta_i^2 - 2 \lambda} - \frac{2 \lambda (n \theta_i^2 - 2 \lambda)}{d_i^2 \theta_i} \sum_{k = 1}^{p} \theta_k^{-2} = \frac{d_i^2}{\theta_i} 
\end{align*}
Finding a common denominator and expanding all terms yields 
\begin{align*}
&[-2 \lambda n^2 T] \theta_i^4 + [d_i^4 (p - n) + 8 n \lambda^2 T] \theta_i^2 + [2 \lambda (d_i^4 - 4 \lambda^2 T)] = 0,
\end{align*}
where $T = \sum_{k = 1}^{p} \theta_k^{-2} = \sum_{k = 1}^{r} \theta_k^{-2} + (p - r) \frac{n}{2 \lambda}$. For the sake of notation, let us define $c_1(T) = -2 \lambda n^2 T$, $c_{2, i}(T) = d_i^4 (p - n) + 8 n \lambda^2 T$, and $c_{3, i}(T) = 2 \lambda (d_i^4 - 4 \lambda^2 T)$.

Summarizing what we have done so far, $\boldsymbol{\beta} \succ 0$ and $\boldsymbol{\theta} \succ 0$ must satisfy:
\begin{align}
\beta_i &= \frac{d_i^2 \theta_i}{n \theta_i^2 - 2 \lambda}, & i = 1, \dots, r \label{eq:multfrob_equiv_eq5} \\
\beta_i &= \sqrt{\frac{2 \lambda T}{p}}, & i = r+1, \dots, n \label{eq:multfrob_equiv_eq6} \\
0 &= c_1(T) \theta_i^4 + c_{2, i}(T) \theta_i^2 + c_{3, i}(T), & i = 1, \dots, r \label{eq:multfrob_equiv_eq7} \\
\theta_i &= \sqrt{\frac{2 \lambda}{n}}, & i = r + 1, \dots, p \label{eq:multfrob_equiv_eq8} \\
T &= \sum_{k = 1}^{r} \theta_k^{-2} + (p - r) \frac{n}{2 \lambda}. & \label{eq:multfrob_equiv_eq9}
\end{align}

Now, from the quartic polynomial in \eqref{eq:multfrob_equiv_eq7}, the four possible roots are
\begin{align*}
\theta_i = \pm \sqrt{\frac{-c_{2, i}(T) \pm \sqrt{c_{2, i}^2(T) - 4 c_1(T) c_{3, i}(T)}}{2 c_1(T)}}.
\end{align*}
In any case, $c_1(T) < 0$, $c_{2, i}(T) > 0$, and
\begin{align*}
c_{2, i}^2(T) - 4 c_1(T) c_{3, i}(T) = d_i^8 (p - n)^2 + 16 \lambda^2 n p d_i^4 T > 0,
\end{align*}
so that
\begin{align}
\theta_i = \sqrt{\frac{-c_{2, i}(T) - \sqrt{c_{2, i}^2(T) - 4 c_1(T) c_{3, i}(T)}}{2 c_1(T)}}. \label{eq:multfrob_equiv_troot}
\end{align}
is always a real positive root. Furthermore, if $\theta_i$ is given by \eqref{eq:multfrob_equiv_troot} for each $i = 1, \dots, r$, then
\begin{align*}
n \theta_i^2 - 2 \lambda &= \frac{1}{-4 \lambda n T} \left[ - d_i^4 (p - n) - 8 n \lambda^2 T - \sqrt{d_i^8 (p - n)^2 + 16 \lambda^2 n p d_i^4 T} \right] - 2 \lambda \\
&= 2 \lambda + \frac{1}{4 \lambda n T}\left[ d_i^4 (p - n) + \sqrt{d_i^8 (p - n)^2 + 16 \lambda^2 n p d_i^4 T} \right] - 2 \lambda \\
&\geq 0.
\end{align*}
Thus, the corresponding $\beta_i$, which we have already shown to be given by $\beta_i = \frac{d_i^2 \theta_i}{n \theta_i^2 - 2 \lambda}$, is also positive. This shows that \eqref{eq:multfrob_equiv_troot} yields a feasible solution for our system of equations \eqref{eq:multfrob_equiv_eq5}-\eqref{eq:multfrob_equiv_eq9}. We claim that this is the only choice of $\theta_i$, which yields a feasible solution. To see this, we first immediately eliminate the two negative square root solutions due to the positivity constraint on $\theta_i$. We next divide the argument into three cases.

\textbf{Case 1:} If $d_i^4 - 4 \lambda^2 T = 0$, then $c_{3, 1}(T) = 0$ and so
\begin{align*}
\theta_i = \sqrt{\frac{-c_{2, i}(T) \pm c_{2, i}(T)}{2 c_1(T)}},
\end{align*}
which equals $0$ if we take the positive sign. Thus, in this case, there is only one feasible root by choosing the negative sign.

\textbf{Case 2:} If $d_i^4 - 4 \lambda^2 T > 0$, then $c_{3, i}(T) > 0$. Additionally, $c_{1}(T) < 0$ and $c_{2, i}(T) > 0$, so by Descartes' rule of signs, there is at most one real positive solution to \eqref{eq:multfrob_equiv_eq7}. As shown earlier, the root given in \eqref{eq:multfrob_equiv_troot} is a real positive solution. This must be the unique real positive root by Descartes'.

\textbf{Case 3:} If $d_i^4 - 4 \lambda^2 T < 0$, then $c_{3, i}(T) < 0$, so by Descartes' rule of signs, there are at most two real positive solutions to \eqref{eq:multfrob_equiv_eq7}. We have already found one real positive root, given by \eqref{eq:multfrob_equiv_troot}. Suppose also that the other possible root
\begin{align}
\theta_i = \sqrt{\frac{-c_{2, i}(T) + \sqrt{c_{2, i}^2(T) - 4 c_1(T) c_{3, i}(T)}}{2 c_1(T)}}. \label{eq:multfrob_equiv_root2}
\end{align}
is a real positive root. Since the corresponding $\beta_i$ must also be positive (in order to be feasible) and we have already shown that $\beta_i = \frac{d_i^2 \theta_i}{n \theta_i^2 - 2 \lambda}$, it follows that the denominator $n \theta_i^2 - 2 \lambda$ must be positive. However, substituting \eqref{eq:multfrob_equiv_root2} into this denominator gives
\begin{align*}
n \theta_i^2 - 2 \lambda &= \frac{1}{-4 \lambda n T} \left[ - d_i^4 (p - n) - 8 n \lambda^2 T + \sqrt{d_i^8 (p - n)^2 + 16 \lambda^2 n p d_i^4 T} \right] - 2 \lambda \\
&= 2 \lambda + \frac{1}{4 \lambda n T}\left[ d_i^4 (p - n) - \sqrt{d_i^8 (p - n)^2 + 16 \lambda^2 n p d_i^4 T} \right] - 2 \lambda \\
&\leq 0.
\end{align*}
This contradicts the positivity of $\beta_i$ and hence is not a feasible solution.

Thus, in any case, we have shown that there is only one feasible root of \eqref{eq:multfrob_equiv_eq7}, which is given by \eqref{eq:multfrob_equiv_troot}.

The last step of this proof is to show that there exists a $T$ which solves our system of equations \eqref{eq:multfrob_equiv_eq5}-\eqref{eq:multfrob_equiv_eq9}. In particular, we can substitute \eqref{eq:multfrob_equiv_troot} into \eqref{eq:multfrob_equiv_eq9} to see that $T$ must satisfy
\begin{align}
T &= \sum_{k = 1}^{r} \frac{2 c_1(T)}{-c_{2, i}(T) - \sqrt{c_{2, i}^2(T) - 4 c_1(T) c_{3, i}(T)}} + (p - r) \frac{n}{2 \lambda} \\
&= \sum_{k = 1}^{r} \frac{4 \lambda n^2 T}{d_i^4 (p - n) + 8 n \lambda^2 T + \sqrt{d_i^8 (p - n)^2 + 16 \lambda^2 n p d_i^4 T}} + (p - r) \frac{n}{2 \lambda}. \label{eq:multfrob_equiv_T}
\end{align}
Let $f(T)$ denote the right hand side of the equation in \eqref{eq:multfrob_equiv_T}. It can be shown that $f'(T) > 0$ for all $T \geq 0$. Also, when $T = 0$, we have that $f(T) = (p - r) \frac{n}{2 \lambda} > 0 = T$, and as $T \rightarrow \infty$, we have $f(T) \rightarrow \frac{n p}{2 \lambda} < \infty$. Together, these observations imply that there exists a unique solution to the equation $T = f(T)$ for some $T > 0$. Thus, there exists a unique feasible $T$, and hence also $\boldsymbol{\beta}$ and $\boldsymbol{\theta}$, to our system of equations \eqref{eq:multfrob_equiv_eq5}-\eqref{eq:multfrob_equiv_eq9}.
\end{proof}

As a direct consequence, since the PC scores and loadings from iPCA are the eigenvectors of $\hat{\Sig}$ and $\hat{\Delt}$, respectively, Theorem~\ref{thm:multfrob1} shows that when $K = 1$, the PC scores and loadings from iPCA with the multiplicative Frobenius penalty are precisely $\U$ and $\V$, as in PCA. In other words, iPCA with the additive or multiplicative Frobenius penalty is a proper generalization of PCA to multiple data sets.

\section{Subspace Consistency of the Additive $L_1$ Correlation iPCA Estimator} \label{sec:s_consistency}

While the multiplicative Frobenius estimator appears superior from an optimization point-of-view, we will show in this section that the additive $L_1$ correlation estimator satisfies one of the first statistical guarantees in the data integration context. Specifically, our primary objective will be to prove that the additive $L_1$ correlation estimator after one Flip-Flop iteration, outlined in Algorithm~\ref{alg:glasso_off_cor}, is a consistent estimator of the true joint covariance matrix $\Sig$ and hence also consistently estimates the underlying joint subspace.

As mentioned in the main text, the additive $L_1$ correlation estimator applies the $L_1$ penalty to the correlation matrix, rather than the usual covariance matrix, and has been adopted previously in \citet{zhou2014gemini} and \citet{rothman2008sparse} for non-integrated data. In this non-integrated setting, \citet{zhou2014gemini} derived convergence rates for the one-step version of Algorithm~\ref{alg:glasso_off_cor}, assuming only one matrix instance was observed from the matrix-variate normal distribution. Motivated by this approach, we extend the proof idea and results in \citet{zhou2014gemini}, where $K=1$, to iPCA, where we observe one matrix instance for each of the $K \geq 1$ distinct matrix-variate normal models. Our proof closely resembles \citet{zhou2014gemini} with technical challenges due to the estimation of $\Sig$ from multiple $\Delt_k$'s with different $p_k$'s. We will see later that the consistency rate obtained in our main theorem, Theorem~\ref{step3_thm}, which holds for all $K \geq 1$, is equivalent to the rate given in \citet{zhou2014gemini} when $K = 1$.

Note however that this subspace consistency property is an incredibly unique property of the additive $L_1$ correlation estimator. When trying to prove a similar result for the other proposed estimators, we run into several difficulties. For instance, if we use the additive $L_1$ covariance estimator, the proof of our main result in Theorem~\ref{step3_thm} no longer goes through due to an additional $\sqrt{p}$ term in the graphical lasso convergence rate \citep{rothman2008sparse}. More specifically, the graphical lasso convergence rate when applied to the inverse covariance matrix, known to be $O\left( \sqrt{\frac{(p + s) \log(p)}{n}} \right)$, is not fast enough to guarantee statistical convergence of the Flip-Flop algorithm. However, by applying the graphical lasso to the inverse correlation matrix, we obtain a faster rate of $O\left( \sqrt{\frac{s \log(p)}{n}} \right)$, which then enables us to prove Theorem~\ref{step3_thm}. Other works, which have studied convergence rates of sparse penalties in the non-integrated setting, are also not applicable for integrated data problems. In particular, \citet{tsiligkaridis2013convergence} proved convergence rates for the additive $L_1$ covariance penalty but assumed multiple matrix observations per matrix-variate normal model. Since iPCA assumes only one matrix observation per model, the guarantees from \citet{tsiligkaridis2013convergence} do not hold. 

The problem of proving rates of convergence for the Frobenius estimators is even more difficult than the $L_1$ estimators. Without imposing some additional structure on the covariance matrices, we cannot even hope to prove statistical consistency in the $p_k > n$ setting. For the $L_1$ penalties, it is natural to impose a sparsity constraint, but with the dense Frobenius penalties, the appropriate underlying structure of the covariance matrices is unclear. One preliminary idea is to exploit the g-convexity of the log-likelihood function and impose some additional structure based upon the associated manifold. However, we leave this for future research as it requires developing a whole new set of tools, combined with Riemannian geometry, to study statistical properties in the manifold space, rather than the usual Euclidean space.

\begin{algorithm}
\caption{Flip-Flop Algorithm for Additive $L_1$ Correlation-Penalized iPCA Estimators}\label{alg:glasso_off_cor}
\begin{algorithmic} [1]
\State Center the columns of $\X_1, \dots, \X_K$, and initialize $\hat{\Sig}$, $\hat{\Delt}_1, \dots, \hat{\Delt}_K$ to be the identity matrix of the appropriate size.
\While{not converged}

\For{$k = 1, \dots, K$}
\State Compute sample covariance matrix: $\hat{\mathbf{S}}_k = \frac{1}{n} \X_k^T \smash[t] {\hat{\Sig}}^{-1} \X_k$ \tikzmark{top1} 
\State Get standard deviation estimate: $\hat{\W}_k = \text{diag}(\hat{\mathbf{S}}_k)^{1/2}$
\State Convert to sample correlation matrix: $\hat{\mathbf{S}}_{\rho,k} = \hat{\W}_k^{-1} \hat{\mathbf{S}}_k \hat{\W}_k^{-1}$
\State Apply graphical lasso to estimate correlation matrix:

\vspace{1mm}
\setlength\parindent{2cm} $\displaystyle \hat{\Delt}_{\rho, k}^{-1} =  \argmin_{\Delt^{-1}_{\rho, k}} -\log | \Delt^{-1}_{\rho, k} | + \mathrm{tr}(\hat{\mathbf{S}}_{\rho,k} \Delt^{-1}_{\rho, k}) + \lambda_k \norm{\Delt^{-1}_{\rho, k}}_{1, \text{off}}$ \label{deltkhat_rho_step} \tikzmark{right1}
\vspace{-1mm}
\State Convert back to covariance estimate: $\hat{\Delt}_k = \hat{\W}_k \hat{\Delt}_{\rho, k} \hat{\W}_k$ \label{deltkhat_step} \tikzmark{bottom1} 
\EndFor

\State Compute sample covariance matrix: $\hat{\mathbf{S}}_{\Sig} = \frac{1}{p} \sum_{k=1}^{K} \X_k \smash[t] {\hat{\Delt}}^{-1}_k \X_k^T$ \label{shat_sig_step} \tikzmark{top2}
\State Get standard deviation estimate: $\hat{\W}_{\Sig} = \text{diag}(\hat{\mathbf{S}}_{\Sig})^{1/2}$
\State Convert to sample correlation matrix: $\hat{\mathbf{S}}_{\rho, \Sig} = \hat{\W}_{\Sig}^{-1} \hat{\mathbf{S}}_{\Sig} \hat{\W}_{\Sig}^{-1}$ \label{shat_rho_sig_step}
\State Apply graphical lasso to estimate correlation matrix:

\vspace{1mm}
\setlength\parindent{2cm} $\displaystyle \hat{\Sig}^{-1}_{\rho} =  \argmin_{\Sig^{-1}_{\rho}} -\log | \Sig^{-1}_{\rho} | + \mathrm{tr}(\hat{\mathbf{S}}_{\rho,\Sig} \Sig^{-1}_{\rho}) + \lambda_{\Sig} \norm{\Sig^{-1}_{\rho}}_{1, \text{off}}$
\vspace{-1mm}
\State Convert back to covariance estimate: $\hat{\Sig} = \hat{\W}_{\Sig} \hat{\Sig}_{\rho} \hat{\W}_{\Sig}$ \label{alg:sig_noniterative} \tikzmark{bottom2}
\EndWhile \vspace{-12pt}
\end{algorithmic}
\AddNote{top1}{bottom1}{right1}{Update $\Delt_k$}
\AddNote{top2}{bottom2}{right1}{Update $\Sig$}
\end{algorithm}

We next collect notation. Let $\hat{\Sig}$ denote the additive $L_1$ correlation estimator obtained after one iteration of the while loop in Algorithm~\ref{alg:glasso_off_cor}. Assume that for each $k = 1, \dots, K$, the true population model is given by $\X_k \sim N_{n, p_k}(\mathbf{0}, \: \Sig \otimes \Delt_k)$, where $\Sig = (\sigma_{ij})$ and $\Delt_k = (\delta_{k,ij})$. For the purpose of identifiability, define $\Sig^* = (\sigma_{*, ij}) = n \Sig / \mathrm{tr}(\Sig)$ and $\Delt_k^* = (\delta_{k, ij}^*) = \mathrm{tr}(\Sig) \Delt_k \ n$ so that $\mathrm{tr}(\Sig^*) = n$ and $\Sig^* \otimes \Delt_k^* = \Sig \otimes \Delt_k$ for each $k$. Let $\rho(\Sig)$ and $\rho(\Delt_k)$ denote the true correlation matrices corresponding to $\Sig$ and $\Delt_k$, respectively. Define $\mathrm{vec}(\A)$ to be the vectorization operator which creates a column vector from the matrix $\A$ by stacking the columns of $\A$ below one another. Then for each $k = 1, \dots, K$, put $\tilde{\mathbf{S}}_k = \mathrm{vec}(\X_k) \mathrm{vec}(\X_k)^T$ and $\bar{\mathbf{S}}_k = \mathrm{vec}(\X_k)^T \mathrm{vec}(\X_k)$. Let $\tilde{\mathbf{S}}_{k}^{rq}$ denote the $r, q^{th}$ block of size $n \times n$ of $\tilde{\mathbf{S}}_k$, and let $\bar{\mathbf{S}}_k^{rq}$ denote the $r, q^{th}$ block of size $p_k \times p_k$ of $\bar{\mathbf{S}}_k$. If $\A$ is a matrix, let $\norm{\A}_2$ denote the operator norm or the maximum singular value of $\A$, let $\norm{\A}_F$ denote the Frobenius norm (i.e. $\norm{\A}_F^2 = \sum_{i,j} a_{ij}^2$), let $\norm{\A}_{0,\text{off}}$ denote the number of non-zero non-diagonal entries in $\A$, let $\norm{\A}_1 = \sum_{i,j} |a_{ij}|$, and let $\norm{\A}_{1,\text{off}} = \sum_{i\neq j} |a_{ij}|$. Denote the stable rank of $\A$ by $r(\A) = \norm{\A}_F^2/\norm{\A}_2^2$. Let us also write for a real symmetric matrix $\A$, $\phi_{\min}(\A)$ to be the minimum eigenvalue of $\A$. Define $\sigma_{\min} = \min_i \sigma_{ii}$, $\sigma_{\max} = \max_i \sigma_{ii}$, $\delta_{k,\min} = \min_i \delta_{k,ii}$, $\delta_{k,\max} = \max_i \delta_{k,ii}$, and similarly for $\sigma_{*, \min}$, $\sigma_{*, \max}$, $\delta_{k, \min}^*$, and $\delta_{k, \max}^*$. Also write $a \vee b = \max(a,b)$ and $a \wedge b = \min(a,b)$. If $a = o(b)$, then $\left| a/b \right| \rightarrow 0$ as $n, p_1, \dots, p_K \rightarrow \infty$. If $a \asymp b$, then there exists positive constants $c, C$ such that $cb \leq a \leq Cb$ as $n, p_1, \dots, p_K \rightarrow \infty$. Lastly, we also adopt the notation defined in Algorithm \ref{alg:glasso_off_cor} for the remainder of this section.


To establish the consistency of $\hat{\Sig}$, we require the following assumptions, which are generalizations of those in \citep{zhou2014gemini}:
\begin{enumerate} [label=(A\arabic*)]
\item Assume that $\Sig^{-1}$ and $\Delt^{-1}_k$ are sparse with respect to each other's dimensions: $s_{\Sig} := \norm{\Sig^{-1}}_{0, \text{off}} = o \left( \frac{p^2}{p_k \log (n \vee p_k)} \right)$ and $s_k := \norm{\Delt^{-1}_k}_{0, \text{off}} = o \left( \frac{n}{\log (n \vee p_k)} \right)$ for each $k = 1, \dots, K$.\label{a1}
\item Assume that we have uniformly bounded spectra: $0 < \phi_{\min}(\Sig) \leq \phi_{\max}(\Sig) < \infty$ and $0 < \phi_{\min}(\Delt_k) \leq \phi_{\max}(\Delt_k) < \infty$ for each $k = 1, \dots, K$. \label{a2}
\item Assume that the inverse correlation matrices satisfy $\norm{\rho(\Sig)^{-1}}_1 \asymp n$ and $\norm{\rho(\Delt_k)^{-1}}_1 \asymp p_k$ for each $k = 1, \dots, K$. \label{a3}
\item Assume that $K$ is finite and the growth rate of $n$ and $p_1, \dots, p_K$ satisfy $n \vee p_k = o(\mathrm{exp}(n \wedge p_k))$ for each $k = 1, \dots, K$. \label{a4}
\end{enumerate}

\begin{remark}
\begin{enumerate} [label=(\arabic*)]
\item Rather than verifying the sparsity assumption of $\Sig^{-1}$ in \ref{a1} and the growth rate \ref{a4} for each $k$, it is sufficient to check that $s_{\Sig} = o \left( \frac{p^2}{\max_k p_k \log (n \vee \max_k p_k)} \right)$ and $n \vee \max_k p_k = o(\mathrm{exp}(n \wedge \min_k p_k))$, respectively.
\item \ref{a1} implies that $\sqrt{\frac{s_k \log(n \vee p_k)}{n}} \rightarrow 0$ and $\sqrt{\frac{s_{\Sig} \log(n \vee p_k)}{p_k}} \rightarrow 0$ as $n, p_1, \dots, p_K \rightarrow \infty$.
\end{enumerate}
\end{remark}


Under these assumptions, we prove our main statistical consistency result in Theorem~\ref{step3_thm}. From this, we can easily establish subspace consistency for $\hat{\Sig}$ in Corollary~\ref{subspace_consistency_thm}. The overall idea of this convergence proof is to follow Algorithm~\ref{alg:glasso_off_cor} and sequentially bound each step in the algorithm. First, the error from line~\ref{deltkhat_step} of Algorithm~\ref{alg:glasso_off_cor} can be bounded by adapting results from \citet{rothman2008sparse}. Then, in Lemma~\ref{step1_lemma}, we bound the error between $\Sig^*$ and the sample covariance estimate $\hat{\mathbf{S}}_{\Sig}$ defined in the step~\ref{shat_sig_step}. Following the Flip-Flop algorithm, we next bound the error between $\rho(\Sig)$ and the sample correlation estimate $\hat{\mathbf{S}}_{\rho,\Sig}$ from step~\ref{shat_rho_sig_step} in Theorem~\ref{step2_thm}. Finally, in Theorem~\ref{step3_thm}, we prove the rate of convergence in the operator and Frobenius norms for $\smash[t] {\hat{\Sig}}^{-1}$ and $\hat{\Sig}$ as defined in step~\ref{alg:sig_noniterative} of the algorithm. Two direct consequences of convergence in the operator norm are consistent eigenvalues and eigenvectors of $\hat{\Sig}$, which in turn yield subspace consistency.

\subsection{Preliminaries} \label{sec:s_consistency_prelims}

The main driver behind Lemma~\ref{step1_lemma} is a large deviation inequality, namely Theorem 13.1 from \citet{zhou2014supplement}. As this is an important result used multiple times in the proof of Lemma~\ref{step1_lemma}, we state Theorem 13.1 from \citet{zhou2014supplement} below in a slightly different form for convenience.

\begin{theorem} \label{main_ci_thm}
Assume that $n \vee p_k \geq 2$. Let $\M$ be an $n \times n$ matrix and $\N$ be a $p_k \times p_k$ matrix such that $\frac{1}{n} \norm{\M}_F^2 < \infty$ and $\frac{1}{p_k} \norm{\N}_F^2 < \infty$. Define $\tau_k = 2C \tilde{K}^2 \log^{1/2}(n \vee p_k)$, where $C := \frac{1}{\sqrt{c}} \vee \frac{1}{c} \vee 1$ and $\tilde{K}$ and $c$ are the constants from Theorem 12.1 in \citet{zhou2014supplement}. 
\begin{enumerate}
\item [(i)] If the stable ranks satisfy $r(\Sig^{1/2} \M \Sig^{1/2}) \geq 4 \log(n \vee p_k)$ and $r(\Delt_k^{1/2} \N \Delt_k^{1/2}) \geq 4 \log(n \vee p_k)$, then with probability $1 - \frac{3}{(n \vee p_k)^2}$, we have that
\begin{align*}
\norm{\mathrm{diag}(\Delt_k)^{-1/2} \left( \frac{1}{n} \sum_{q=1}^{n} \sum_{r=1}^{n} \M_{qr} \bar{\mathbf{S}}_k^{rq} \right) \mathrm{diag}(\Delt_k)^{-1/2} - \frac{\mathrm{tr}(\Sig \M)}{n} \rho(\Delt_k)}_{\infty} &\leq D_k \tau_k \\
\text{and } \norm{\mathrm{diag}(\Sig)^{-1/2} \left( \frac{1}{p} \sum_{q=1}^{p_k} \sum_{r=1}^{p_k} \N_{qr} \tilde{\mathbf{S}}_k^{rq} \right) \mathrm{diag}(\Sig)^{-1/2} - \frac{\mathrm{tr}(\Delt_k \N)}{p} \rho(\Sig)}_{\infty} &\leq D'_k \tau_k,
\end{align*}
where $D_k = \frac{1}{n} \norm{\Sig^{1/2} \M \Sig^{1/2}}_F$ and $D_k' = \frac{1}{p} \norm{\Delt_k^{1/2} \N \Delt_k^{1/2}}_F$.
\item [(ii)] If $n \vee p_k = o(\mathrm{exp}{n \wedge p_k})$, then with probability $1 - \frac{3}{(n \vee p_k)^2}$, the above inequalities hold with $D_k = \frac{2}{\sqrt{n}} \norm{\Sig}_2 \norm{\M}_2$ and $D_k' = \frac{2 \sqrt{p_k}}{p} \norm{\Delt_k}_2 \norm{\N}_2$.
\end{enumerate}
\end{theorem}

We refer to \citet{zhou2014supplement} for the proof of this theorem, but if one looks at the interior of this proof, we obtain the following useful result, presented in Corollary~\ref{main_ci_cor}.

To simplify notation, let $\mathcal{E}_{\Delt}(k,\M)$ denote the event
\begin{align*}
\left \{ \norm{\mathrm{diag}(\Delt_k)^{-1/2} \left( \frac{1}{n} \sum_{q=1}^{n} \sum_{r=1}^{n} \M_{qr} \bar{\mathbf{S}}_k^{rq} \right) \mathrm{diag}(\Delt_k)^{-1/2} - \frac{\mathrm{tr}(\Sig \M)}{n} \rho(\Delt_k)}_{\infty} \leq \frac{2}{\sqrt{n}} \norm{\Sig}_2 \norm{\M}_2 \tau_k \right \},
\end{align*}
let $\mathcal{E}_{\Sig}(k,\N)$ denote the event
\begin{align*}
\left \{ \norm{\mathrm{diag}(\Sig)^{-1/2} \left( \frac{1}{p} \sum_{q=1}^{p_k} \sum_{r=1}^{p_k} \N_{qr} \tilde{\mathbf{S}}_k^{rq} \right) \mathrm{diag}(\Sig)^{-1/2} - \frac{\mathrm{tr}(\Delt_k \N)}{p} \rho(\Sig)}_{\infty} \leq \frac{2 \sqrt{p_k}}{p} \norm{\Delt_k}_2 \norm{\N}_2 \tau_k \right \},
\end{align*}
and define for each $k = 1, \dots, K$,
\begin{align*}
\nu_{n,k} &:= \frac{2}{\sqrt{n}} \: \tau_k = 4 C \tilde{K}^2 \sqrt{\frac{\log(n \vee p_k)}{n}} \\
\nu_{p_k} &:= \frac{2 \sqrt{p_k}}{p} \: \tau_k = 4 C \tilde{K}^2 \frac{\sqrt{p_k \log(n \vee p_k)}}{p} = 4 C \tilde{K}^2 \frac{p_k}{p} \sqrt{\frac{\log(n \vee p_k)}{p_k}}.
\end{align*}

\begin{corollary}\label{main_ci_cor}
Assume the same conditions and notation as Theorem~\ref{main_ci_thm}. Also suppose that $n \vee p_k = o(\mathrm{exp}{n \wedge p_k})$ Then under the event $\mathcal{E}_{\Delt}(k, \M)$, we have that 
\begin{align*}
\left | \left( \frac{1}{n} \sum_{q=1}^{n} \sum_{r=1}^{n} \M_{qr} \bar{\mathbf{S}}_k^{rq} \right)_{ij} - \frac{\mathrm{tr}(\Sig \M)}{n} \delta_{k,ij} \right | \leq \nu_{n,k} \norm{\Sig}_2 \norm{\M}_2 \sqrt{\delta_{k,ii} \delta_{k,jj}} \qquad \forall \: i,j,
\end{align*}
and under the event $\mathcal{E}_{\Sig}(k, \N)$, we have that
\begin{align*}
\left | \left( \frac{1}{p} \sum_{q=1}^{p_k} \sum_{r=1}^{p_k} \N_{qr} \tilde{\mathbf{S}}_k^{rq} \right)_{ij} - \frac{\mathrm{tr}(\Delt_k \N)}{p} \sigma_{ij} \right | \leq \nu_{p_k} \norm{\Delt_k}_2 \norm{\N}_2 \sqrt{\sigma_{ii} \sigma_{jj}} \qquad \forall \: i,j.
\end{align*}

\end{corollary}

We now have the necessary large deviation inequalities to prove Lemma~\ref{step1_lemma}, a generalization of Lemma 6.1 from \citet{zhou2014gemini}. Though the proof of Lemma~\ref{step1_lemma} closely resembles that of Lemma 6.1 from \citet{zhou2014gemini}, modifications must be made as iPCA considers multiple distinct matrix-variate normal models while \citet{zhou2014gemini} considers only one matrix-variate normal model. For clarity, we give our proof in its entirety and refer to results in \citet{zhou2014gemini} when necessary.

\begin{lemma}\label{step1_lemma}
Suppose that \ref{a1}-\ref{a4} hold. Let $\hat{\Delt}_{\rho,k}$ and $\hat{\Delt}_k$ be obtained as in steps~\ref{deltkhat_rho_step} and \ref{deltkhat_step} in Algorithm~\ref{alg:glasso_off_cor}, where we choose 
\begin{align}
\lambda_k = \frac{2 \alpha_k}{\epsilon (1-\alpha_k)} \geq \frac{3\alpha_k}{1-\alpha_k} \quad \text{for } \: \alpha_k = A \: \nu_{n,k} \: \text{where } A = \frac{\sqrt{n} \norm{\Sig}_F}{\mathrm{tr}(\Sig)} \label{lambda_k}
\end{align}
and $\epsilon \in (0, 2/3)$. Then on event $\mathcal{E}^*$, we have for $\hat{\mathbf{S}}_{\Sig}$ defined in step~\ref{shat_sig_step} of Algorithm~\ref{alg:glasso_off_cor},
\begin{align*}
\left | (\hat{\mathbf{S}}_{\Sig} - \Sig^*)_{ij} \right | &\leq \sum_{k=1}^{K} \frac{p_k}{p} \left[ 4C \tilde{K}^2 \sqrt{\sigma_{*,ii} \sigma_{*,jj}} \sqrt{\frac{\log(n \vee p_k)}{p_k}} (1 + o(1)) \right] + \sum_{k=1}^{K} | \sigma_{*,ij} | \tilde{\mu}_k \\
&= \sum_{k=1}^{K} \left[ \sqrt{\sigma_{*,ii} \sigma_{*,jj}} \: \nu_{p_k} (1 + o(1))  +  | \sigma_{*,ij} | \tilde{\mu}_k \right],
\end{align*}
where
\begin{align*}
\tilde{\mu}_k &:= \lambda_k \frac{\norm{\hat{\Delt}_{\rho,k}^{-1}}_{1, \text{off}}}{p} + \frac{\alpha_k}{1-\alpha_k} \frac{\norm{\hat{\Delt}_{\rho,k}^{-1}}_{1}}{p} \leq \mu_k, \\
\mu_k &:= \lambda_k \frac{\norm{\rho(\Delt_k)^{-1}}_{1, \text{off}}}{p} + \frac{\alpha_k}{1-\alpha_k} \frac{\norm{\rho(\Delt_k)^{-1}}_{1}}{p} + o(\lambda_k).
\end{align*}
Moreover, $\mathbb{P}(\mathcal{E}^*) \geq 1 - \sum_{k=1}^{K} \frac{8}{(n \vee p_k)^2}$.

\noindent
To put simply, with high probability,
\begin{align*}
\norm{\hat{\mathbf{S}}_{\Sig} - \Sig^*}_{\infty} &\leq \sum_{k=1}^{K} \frac{p_k}{p} \left[ 4 C \tilde{K}^2 \sigma_{*, \max} \sqrt{\frac{\log(n \vee p_k)}{p_k}} (1 + o(1)) \right] + \sum_{k=1}^{K} \sigma_{*, \max} \mu_k \\
&= \sigma_{*, \max} \sum_{k=1}^{K} \left[ \nu_{p_k} (1 + o(1)) + \mu_k \right].
\end{align*}
\end{lemma}

\begin{proof}

For each $k = 1, \dots, K$, define
\begin{align*}
\R_{\Sig, k} &:= [\delta_{k, 11} \mathrm{vec}(\Sig) \: \dots \: \delta_{k, 1p_k} \mathrm{vec}(\Sig) \: \dots \: \delta_{k, p_k p_k} \mathrm{vec}(\Sig)] \equiv \mathrm{vec}(\Sig) \otimes \mathrm{vec}(\Delt_k)^T \\
\hat{\R}_{\Sig, k} &:= \left[ \mathrm{vec}(\tilde{\mathbf{S}}_{k}^{11}) \: \dots \: \mathrm{vec}(\tilde{\mathbf{S}}_{k}^{1p_k}) \: \dots \: \mathrm{vec}(\tilde{\mathbf{S}}_{k}^{p_k p_k}) \right].
\end{align*}

Then one can verify as in \citet{zhou2014gemini} that
\begin{align}
\mathrm{vec}(\Sig^*) = \frac{1}{p} \sum_{k=1}^{K} \R_{\Sig, k} \mathrm{vec}((\Delt_k^*)^{-1}) \quad \text{and} \quad \mathrm{vec}(\hat{\mathbf{S}}_{\Sig}) = \frac{1}{p} \sum_{k=1}^{K} \hat{\R}_{\Sig, k} \mathrm{vec}(\smash[t] {\hat{\Delt}}^{-1}_k). \label{step1_eq1}
\end{align}

The equalities from \eqref{step1_eq1} thus yield
\begin{align}
\mathrm{vec}(\hat{\mathbf{S}}_{\Sig} - \Sig^*) &= \frac{1}{p} \sum_{k=1}^{K} \hat{\R}_{\Sig, k} \mathrm{vec}(\smash[t] {\hat{\Delt}}^{-1}_k) - \frac{1}{p} \sum_{k=1}^{K} \R_{\Sig, k} \mathrm{vec}((\Delt_k^*)^{-1}) \\
&= \frac{1}{p} \sum_{k=1}^{K} \hat{\R}_{\Sig, k} \mathrm{vec}(\smash[t] {\hat{\Delt}}^{-1}_k) - \frac{1}{p} \sum_{k=1}^{K} \R_{\Sig, k} \mathrm{vec}((\Delt_k^*)^{-1}) \nonumber \\
&\qquad + \frac{1}{p} \sum_{k=1}^{K} \hat{\R}_{\Sig, k} \mathrm{vec}((\Delt_k^*)^{-1}) - \frac{1}{p} \sum_{k=1}^{K} \hat{\R}_{\Sig, k} \mathrm{vec}((\Delt_k^*)^{-1}) \\
&= \frac{1}{p} \sum_{k=1}^{K} \left[ (\hat{\R}_{\Sig, k} - \R_{\Sig, k}) \mathrm{vec}((\Delt_k^*)^{-1}) +  \hat{\R}_{\Sig, k} \mathrm{vec}(\smash[t] {\hat{\Delt}}^{-1}_k - (\Delt_k^*)^{-1}) \right] \label{decomp1}.
\end{align}

For each $k = 1, \dots, K$, define $\boldsymbol{\Theta}_k := \hat{\Delt}_k - \Delt_k^*$ and $\tilde{\boldsymbol{\Theta}}_k := \smash[t] {\hat{\Delt}}^{-1}_k - (\Delt_k^*)^{-1}$, and notice that $\tilde{\boldsymbol{\Theta}}_k = - \smash[t] {\hat{\Delt}}^{-1}_k \boldsymbol{\Theta}_k (\Delt_k^*)^{-1}$. Then\begin{align}
\hat{\R}_{\Sig, k} \mathrm{vec}(\tilde{\boldsymbol{\Theta}}_k) &= \hat{\R}_{\Sig, k} \mathrm{vec}(\tilde{\boldsymbol{\Theta}}_k) + \R_{\Sig, k} \mathrm{vec}(\tilde{\boldsymbol{\Theta}}_k) - \R_{\Sig, k} \mathrm{vec}(\tilde{\boldsymbol{\Theta}}_k) \\
&= \R_{\Sig, k} \mathrm{vec}(\tilde{\boldsymbol{\Theta}}_k) + (\hat{\R}_{\Sig, k} - \R_{\Sig, k}) \mathrm{vec}(\tilde{\boldsymbol{\Theta}}_k) \\
&= \R_{\Sig, k} \mathrm{vec}(- \smash[t] {\hat{\Delt}}^{-1}_k \boldsymbol{\Theta}_k (\Delt_k^*)^{-1}) + (\hat{\R}_{\Sig, k} - \R_{\Sig, k}) \mathrm{vec}(- \smash[t] {\hat{\Delt}}^{-1}_k \boldsymbol{\Theta}_k (\Delt_k^*)^{-1}) \label{decomp2}.
\end{align}

Putting \eqref{decomp1} and \eqref{decomp2} together gives
\begin{align*}
\mathrm{vec}(\hat{\mathbf{S}}_{\Sig} - \Sig^*) &= \sum_{k=1}^{K} \Big[ \frac{1}{p} (\hat{\R}_{\Sig, k} - \R_{\Sig, k}) \mathrm{vec}((\Delt_k^*)^{-1}) +  \frac{1}{p} \R_{\Sig, k} \mathrm{vec}(- \smash[t] {\hat{\Delt}}^{-1}_k \boldsymbol{\Theta}_k (\Delt_k^*)^{-1}) \nonumber \\
& \qquad \qquad + \frac{1}{p} (\hat{\R}_{\Sig, k} - \R_{\Sig, k}) \mathrm{vec}(- \smash[t] {\hat{\Delt}}^{-1}_k \boldsymbol{\Theta}_k (\Delt_k^*)^{-1}) \Big] \\
& =: \sum_{k=1}^{K} (\U_{1,k} + \U_{2,k} + \U_{3,k}),
\end{align*}
where the matrix correspondent for each of the above terms will be denoted by $\M_{1,k}$, $\M_{2,k}$, and $\M_{3,k}$, respectively. We will proceed to bound each of the terms separately.

In order to bound $\U_{1,k}$, notice that by definition of $\R_{\Sig, k}$ and $\hat{\R}_{\Sig, k}$, 
\begin{align}
\U_{1,k} &= \frac{1}{p} (\hat{\R}_{\Sig, k} - \R_{\Sig, k}) \mathrm{vec}((\Delt_k^*)^{-1}) \label{u1_first}\\
&= \frac{1}{p} \sum_{q=1}^{p_k} \sum_{r=1}^{p_k} \mathrm{vec}(\tilde{\mathbf{S}}_k^{qr} - \delta_{k, qr} \Sig) [(\Delt_k^*)^{-1}]_{qr} \\
&= \frac{1}{p} \sum_{q=1}^{p_k} \sum_{r=1}^{p_k} \mathrm{vec}(\tilde{\mathbf{S}}_k^{qr}) [(\Delt_k^*)^{-1}]_{qr} - \frac{\mathrm{tr}(\Delt_k (\Delt_k^*)^{-1})}{p} \mathrm{vec}(\Sig) \\
&= \frac{1}{p} \sum_{q=1}^{p_k} \sum_{r=1}^{p_k} \mathrm{vec}(\tilde{\mathbf{S}}_k^{qr}) [(\Delt_k^*)^{-1}]_{qr} - \frac{p_k}{p} \mathrm{vec}(\Sig^*) \label{u1_last} \\
\implies \M_{1,k} &= \frac{1}{p} \sum_{q=1}^{p_k} \sum_{r=1}^{p_k} \tilde{\mathbf{S}}_k^{qr} [(\Delt_k^*)^{-1}]_{qr} - \frac{p_k}{p} \Sig^*.
\end{align}

Define $\mathcal{E}_{0,k}$ to be the event $\mathcal{E}_{\Delt}(k,(\Sig^*)^{-1}) \cap \mathcal{E}_{\Sig}(k, (\Delt_k^*)^{-1})$. Then by Corollary~\ref{main_ci_cor}, under the event $\mathcal{E}_{0,k}$, we have $| (\M_{1,k})_{ij} | \leq \nu_{p_k} \sqrt{\sigma_{*,ii} \sigma_{*,jj}}. $ Moreover, by Theorem~\ref{main_ci_thm}, $\mathbb{P}(\mathcal{E}_{0,k}) \geq 1 - \frac{3}{(n \vee p_k)^2}$.

Next, we will bound the second term $\U_{2,k}$. As in \citet{zhou2014supplement}, we can write
\begin{align*}
\U_{2,k} &= \frac{1}{p} \R_{\Sig, k} \mathrm{vec}(- \smash[t] {\hat{\Delt}}^{-1}_k \boldsymbol{\Theta}_k (\Delt_k^*)^{-1}) = \frac{1}{p} \mathrm{tr}(- \smash[t] {\hat{\Delt}}^{-1}_k \boldsymbol{\Theta}_k) \mathrm{vec}(\Sig^*) , \\
\text{and } \M_{2,k} &= \frac{1}{p} \mathrm{tr}(- \smash[t] {\hat{\Delt}}^{-1}_k \boldsymbol{\Theta}_k) \Sig^*.
\end{align*}

By Claim 17.3 in \citet{zhou2014supplement}, it follows that that under event $\mathcal{X}_{0,k} := \mathcal{E}_{\Sig}(k, \I) \cap \mathcal{E}_{\Delt}(k, \I)$,
\begin{align*}
\lambda_{k} \norm{\hat{\Delt}_{\rho,k}^{-1}}_{1, \text{off}} - \frac{\alpha_k}{1-\alpha_k} \norm{\hat{\Delt}_{\rho,k}^{-1}}_1 \leq \mathrm{tr}(\boldsymbol{\Theta}_k \smash[t] {\hat{\Delt}}^{-1}_k) \leq \lambda_{k} \norm{\hat{\Delt}_{\rho,k}^{-1}}_{1, \text{off}} + \frac{\alpha_k}{1-\alpha_k} \norm{\hat{\Delt}_{\rho,k}^{-1}}_1
\end{align*}

Thus,
\begin{align*}
| \mathrm{tr}(-\boldsymbol{\Theta}_k \smash[t] {\hat{\Delt}}^{-1}_k) | \leq \lambda_{k} \norm{\hat{\Delt}_{\rho,k}^{-1}}_{1, \text{off}} + \frac{\alpha_k}{1-\alpha_k} \norm{\hat{\Delt}_{\rho,k}^{-1}}_1,
\end{align*}
which implies that on $\mathcal{X}_{0,k}$,
\begin{align*}
| (\M_{2,k})_{ij} | \leq \frac{| \sigma_{*,ij} |}{p} \left( \lambda_{k} \norm{\hat{\Delt}_{\rho,k}^{-1}}_{1, \text{off}} + \frac{\alpha_k}{1-\alpha_k} \norm{\hat{\Delt}_{\rho,k}^{-1}}_1 \right) = | \sigma_{*,ij} | \tilde{\mu}_k.
\end{align*}
Additionally, by Theorem~\ref{main_ci_thm}, $\mathbb{P}(\mathcal{X}_{0,k}) \geq 1 - \frac{3}{(n \vee p_k)^2}$.

To bound the final term $\U_{3,k}$, we follow the same logic as \eqref{u1_first}-\eqref{u1_last} to obtain
\begin{align*}
\U_{3,k} &= \frac{1}{p} (\hat{\R}_{\Sig, k} - \R_{\Sig, k}) \mathrm{vec}(\tilde{\boldsymbol{\Theta}}_k) \\
&= \frac{1}{p} \sum_{q=1}^{p_k} \sum_{r=1}^{p_k} \mathrm{vec}(\tilde{\mathbf{S}}_k^{qr})[\tilde{\boldsymbol{\Theta}}_k]_{qr} - \frac{\mathrm{tr}(\Delt_k \tilde{\boldsymbol{\Theta}}_k)}{p} \mathrm{vec}(\Sig), \\
\text{and } \M_{3,k} &= \frac{1}{p} \sum_{q=1}^{p_k} \sum_{r=1}^{p_k} \tilde{\mathbf{S}}_k^{qr} [\tilde{\boldsymbol{\Theta}}_k]_{qr} - \frac{\mathrm{tr}(\Delt_k \tilde{\boldsymbol{\Theta}}_k)}{p} \Sig.
\end{align*}

Define $\mathcal{E}_{1,k} := \mathcal{E}_{\Sig}(k, \tilde{\boldsymbol{\Theta}}_k)$. By the proof of Theorem~\ref{main_ci_thm}, $\mathbb{P}(\mathcal{E}_{1,k} | \mathcal{X}_{0,k}) \geq 1 - \frac{2}{(n \vee p_k)^2}$. On the other hand, under the event $\mathcal{E}_{1,k}$, Corollary~\ref{main_ci_cor} gives
$| (\M_{3,k})_{ij} | \leq \nu_{p_k} \norm{\Delt_k}_2 \norm{\tilde{\boldsymbol{\Theta}}_k}_2 \sqrt{\sigma_{ii} \sigma_{jj}}.$

We next bound $\norm{\tilde{\boldsymbol{\Theta}}_k}_2$ using Corollary 10.1 from \citet{zhou2014supplement}, so assuming that $\mathcal{X}_{0,k}$ holds, then
\begin{align*}
\norm{\tilde{\boldsymbol{\Theta}}_k}_2 \leq C' \lambda_k \frac{\sqrt{s_k \vee 1}}{\delta^*_{k,\min} \phi^2_{\min}(\rho(\Delt_k))}.
\end{align*}

In summary, under the event $\mathcal{E}^* := \cap_{k=1}^{K} (\mathcal{E}_{0,k} \cap \mathcal{E}_{1,k} \cap \mathcal{X}_{0,k})$, we have that
\begin{align*}
| (\hat{\mathbf{S}}_{\Sig} - \Sig^*)_{ij} | &\leq \sum_{k=1}^{K} \left( | (\M_{1,k})_{ij} | + | (\M_{2,k})_{ij} | + | (\M_{3,k})_{ij} | \right) \\
&\leq \sum_{k=1}^{K} \bigg[\nu_{p_k} \sqrt{\sigma_{*,ii} \sigma_{*,jj}} + | \sigma_{*,ij} | \tilde{\mu}_k \nonumber\\
&\qquad + \nu_{p_k} \norm{\Delt_k}_2 \left( C' \lambda_k  \frac{\sqrt{s_k \vee 1}}{\delta^*_{k,\min} \phi^2_{\min}(\rho(\Delt_k))} \right) \sqrt{\sigma_{ii} \sigma_{jj}} \bigg] \\
&= \sum_{k=1}^{K} \bigg[\nu_{p_k} \sqrt{\sigma_{*,ii} \sigma_{*,jj}} + | \sigma_{*,ij} | \tilde{\mu}_k \nonumber\\
&\qquad + \nu_{p_k} \norm{\Delt_k}_2 \left( C' \lambda_k  \frac{\sqrt{s_k \vee 1}}{\delta_{k,\min} \phi^2_{\min}(\rho(\Delt_k))} \right) \sqrt{\sigma_{*,ii} \sigma_{*,jj}} \bigg] \\
&\leq \sum_{k=1}^{K} \bigg[\nu_{p_k} \sqrt{\sigma_{*,ii} \sigma_{*,jj}} (1 + o(1)) + | \sigma_{*,ij} | \tilde{\mu}_k \bigg]
\end{align*}
under the assumptions. 

Furthermore, applying the union bound implies that $\mathbb{P}(\mathcal{E}^*) \geq 1 - \sum_{k=1}^{K} \frac{8}{(n \vee p_k)^2}$. Assuming that event $\mathcal{X}_{0,k}$ holds, $\tilde{\mu}_k \leq \mu_k$ is a consequence of Corollary 17.4 from  \citet{zhou2014supplement}, and this concludes the proof.
\end{proof}

It is important to point out that in our choice of $\lambda_k$ in \eqref{lambda_k}, $A$ can be considered a constant under the bounded spectrum assumption \ref{a2}. In addition, \ref{a1} implies that $\sqrt{\frac{\log(n \vee p_k)}{n}} \rightarrow 0$ as $n, p_k \rightarrow \infty$. Therefore, since $\lambda_k$ is on the order of $A \sqrt{\frac{\log(n \vee p_k)}{n}}$, $\lambda_k \rightarrow 0$ as $n, p_k \rightarrow \infty$ (for $k = 1, \dots, K$).

Next, stepping through Algorithm~\ref{alg:glasso_off_cor}, we bound the error between the correlation estimate $\hat{\mathbf{S}}_{\rho, \Sig}$ and the true correlation matrix $\rho(\Sig)$.

\begin{theorem}\label{step2_thm}
Suppose the conditions in Lemma~\ref{step1_lemma} hold. Define $\tilde{\eta}_k := \nu_{p_k} (1 + o(1)) + \tilde{\mu}_k$. Then under event $\mathcal{E}^*$, we have for $\hat{\mathbf{S}}_{\rho, \Sig}$ defined in step~\ref{shat_rho_sig_step} in Algorithm~\ref{alg:glasso_off_cor} and $i \neq j$,
\begin{align}
\left | \left( \hat{\mathbf{S}}_{\rho, \Sig} - \rho(\Sig) \right)_{ij} \right | &\leq \sum_{k=1}^{K} \frac{p_k}{p} \left[ \frac{4 C \tilde{K}^2 \sqrt{\frac{\log(n \vee p_k)}{p_k}} (1 + o(1)) (1 + | \rho(\Sig)_{ij} |)}{1 - \sum_{k=1}^{K} \tilde{\eta}_k} \right] \nonumber \\
&\qquad \qquad + \frac{2 | \rho(\Sig)_{ij} | \sum_{k=1}^{K} \tilde{\mu}_k}{1 - \sum_{k=1}^{K} \tilde{\eta}_k} \label{step2_thm_eq1} \\
&= \sum_{k=1}^{K} \left[ \frac{\nu_{p_k}(1 + o(1))(1 + | \rho(\Sig)_{ij} |)}{1 - \sum_{k=1}^{K} \tilde{\eta}_k} + \frac{2 | \rho(\Sig)_{ij} | \tilde{\mu}_k}{1 - \sum_{k=1}^{K} \tilde{\eta}_k} \right] \\
&\leq \frac{2 \sum_{k=1}^{K} \eta_k}{1 - \sum_{k=1}^{K} \eta_k}, \qquad \text{where } \: \eta_k := \nu_{p_k} (1 + o(1)) + \mu_k.
\end{align}

\end{theorem}

\begin{proof}
Because $\tilde{\mu}_k \leq \mu_k$ by Lemma~\ref{step1_lemma} and $| \rho(\Sig)_{ij} | \leq 1$ for all $i,j$, it is clear that
\begin{align*}
\sum_{k=1}^{K} \left[ \frac{\nu_{p_k}(1 + o(1))(1 + | \rho(\Sig)_{ij} |)}{1 - \sum_{k=1}^{K} \tilde{\eta}_k} + \frac{2 | \rho(\Sig)_{ij} | \tilde{\mu}_k}{1 - \sum_{k=1}^{K} \tilde{\eta}_k} \right] \leq \frac{2 \sum_{k=1}^{K} \tilde{\eta}_k}{1 - \sum_{k=1}^{K} \tilde{\eta}_k} \leq \frac{2 \sum_{k=1}^{K} \eta_k}{1 - \sum_{k=1}^{K} \eta_k}.
\end{align*}
Therefore, it suffices to show \eqref{step2_thm_eq1}.

Assume throughout this proof that event $\mathcal{E}^*$ holds. Then by Lemma~\ref{step1_lemma}, 
\begin{align*}
\left | \frac{[\hat{\mathbf{S}}_{\Sig}]_{ii}}{\sigma_{*,ii}} - 1 \right |  \leq \sum_{k=1}^{K} \left[ \nu_{p_k} (1 + o(1)) + \tilde{\mu}_k \right] = \sum_{k=1}^{K} \tilde{\eta}_k,
\end{align*}
which implies
\begin{align}
\frac{[\hat{\mathbf{S}}_{\Sig}]_{ii}}{\sigma_{*,ii}} \geq 1 - \sum_{k=1}^{K} \tilde{\eta}_k \quad \implies \quad \sqrt{\frac{\sigma_{*,ii}}{[\hat{\mathbf{S}}_{\Sig}]_{ii}}} \leq \sqrt{\frac{1}{1 - \sum_{k=1}^{K} \tilde{\eta}_k}} \label{step2_eq1}
\end{align}
for all $i$. On the other hand, for $i \neq j$, Lemma~\ref{step1_lemma} gives
\begin{align}
\left | \left(\frac{\hat{\mathbf{S}}_{\Sig}}{\sqrt{\sigma_{*,ii} \sigma_{*,jj}}} - \rho(\Sig) \right)_{ij} \right | &= \left | \left(\frac{\hat{\mathbf{S}}_{\Sig}}{\sqrt{\sigma_{*,ii} \sigma_{*,jj}}} - \frac{\Sig^*}{\sqrt{\sigma_{*,ii} \sigma_{*,jj}}} \right)_{ij} \right | \\
&\leq \sum_{k=1}^{K} \left( \nu_{p_k} (1 + o(1)) + \frac{| \sigma_{*,ij} |}{\sqrt{\sigma_{*,ii} \sigma_{*,jj}}} \tilde{\mu}_k \right) \\
&= \sum_{k=1}^{K} \left( \nu_{p_k} (1 + o(1)) + | \rho(\Sig)_{ij} | \tilde{\mu}_k \right), \label{step2_eq2}
\end{align}
Thus, for $i \neq j$,
\begin{align}
\left| \left(\hat{\mathbf{S}}_{\rho, \Sig} - \rho(\Sig)\right)_{ij} \right | &= \left | \frac{[\hat{\mathbf{S}}_{\Sig}]_{ij}}{[\hat{\mathbf{S}}_{\Sig}]_{ii}^{1/2} [\hat{\mathbf{S}}_{\Sig}]_{jj}^{1/2}} - \rho(\Sig)_{ij} \right | \\
&= \left | \frac{\frac{[\hat{\mathbf{S}}_{\Sig}]_{ij}}{\sqrt{\sigma_{*,ii} \sigma_{*,jj}}}}{\frac{[\hat{\mathbf{S}}_{\Sig}]_{ii}^{1/2}}{\sqrt{\sigma_{*,ii}}} \frac{[\hat{\mathbf{S}}_{\Sig}]_{jj}^{1/2}}{\sqrt{\sigma_{*,jj}}}} - \rho(\Sig)_{ij} \right | \\
&\leq \left | \frac{\frac{[\hat{\mathbf{S}}_{\Sig}]_{ij}}{\sqrt{\sigma_{*,ii} \sigma_{*,jj}}} - \rho(\Sig)_{ij}}{\frac{[\hat{\mathbf{S}}_{\Sig}]_{ii}^{1/2}}{\sqrt{\sigma_{*,ii}}} \frac{[\hat{\mathbf{S}}_{\Sig}]_{jj}^{1/2}}{\sqrt{\sigma_{*,jj}}}} \right | + \left | \frac{ \rho(\Sig)_{ij}}{\frac{[\hat{\mathbf{S}}_{\Sig}]_{ii}^{1/2}}{\sqrt{\sigma_{*,ii}}} \frac{[\hat{\mathbf{S}}_{\Sig}]_{jj}^{1/2}}{\sqrt{\sigma_{*,jj}}}} - \rho(\Sig)_{ij} \right | \\
&= \left | \frac{\sqrt{\sigma_{*,ii}}}{[\hat{\mathbf{S}}_{\Sig}]_{ii}^{1/2}} \frac{\sqrt{\sigma_{*,jj}}}{[\hat{\mathbf{S}}_{\Sig}]_{jj}^{1/2}} \right | \left |\frac{[\hat{\mathbf{S}}_{\Sig}]_{ij}}{\sqrt{\sigma_{*,ii} \sigma_{*,jj}}} - \rho(\Sig)_{ij} \right | \nonumber \\
& \qquad \qquad + | \rho(\Sig)_{ij} | \left | \frac{\sqrt{\sigma_{*,ii}}}{[\hat{\mathbf{S}}_{\Sig}]_{ii}^{1/2}} \frac{\sqrt{\sigma_{*,jj}}}{[\hat{\mathbf{S}}_{\Sig}]_{jj}^{1/2}} - 1 \right | \\
&\leq \frac{1}{1 - \sum_{k=1}^{K} \tilde{\eta}_k} \left( \sum_{k=1}^{K} \left( \nu_{p_k} (1 + o(1)) + | \rho(\Sig)_{ij}| \tilde{\mu}_k \right) \right) \nonumber \\
&\qquad \qquad + | \rho(\Sig)_{ij} | \left | \frac{1}{1 - \sum_{k=1}^{K} \tilde{\eta}_k} - 1 \right | \label{step2_eq3} \\
&= \sum_{k=1}^{K} \frac{\nu_{p_k} (1 + o(1)) (1 + | \rho(\Sig)_{ij} |) + 2 | \rho(\Sig)_{ij} | \tilde{\mu}_k}{1 - \sum_{k=1}^{K} \tilde{\eta}_k}.
\end{align}
as desired. Note that \eqref{step2_eq3} follows from \eqref{step2_eq1} and \eqref{step2_eq2}.
\end{proof}

\subsection{Main Result} \label{sec:s_consistency_main_results}

We can now build on top of Theorem~\ref{step2_thm} and existing results to prove our main convergence result in Theorem~\ref{step3_thm}, which is a generalization of Corollary 17.2 from \citet{zhou2014supplement}. Indeed, when $K = 1$, Theorem~\ref{step3_thm} gives the same rates as \citet{zhou2014gemini}. Moreover, as a direct consequence of Theorem~\ref{step3_thm}, we obtain Corollary~\ref{subspace_consistency_thm}, which establishes bounds on the eigenvectors and eigenvalues of the additive $L_1$ correlation iPCA estimator, thereby giving the desired subspace consistency.

\begin{theorem}\label{step3_thm}
Suppose that \ref{a1}-\ref{a4} hold and that $\sqrt{n} \geq \frac{p}{\sqrt{p_k}}$ for each $k = 1, \dots, K$. Assume also that $\eta \leq 1/4$, and $\lambda_{\Sig}$ is chosen to be
\begin{align}
\lambda_{\Sig} = \frac{2 \sum_{k=1}^{K} \tilde{\eta}_k}{\epsilon_1 (1 - \sum_{k=1}^{K} \tilde{\eta}_k)}, \qquad \text{for some } \epsilon_1 \in (0,1).  \label{lambda_Sig}
\end{align}
Then on the event $\mathcal{E}^*$,
\begin{align}
\norm{\hat{\Sig} - \Sig^*}_2 &\leq 2 \tilde{C} \lambda_{\Sig} \sigma_{*, \max} \kappa(\rho(\Sig))^2 \sqrt{s_{\Sig} \vee 1}, \\
\norm{\hat{\Sig} - \Sig^*}_F &\leq 2 \tilde{C} \lambda_{\Sig} \sigma_{*, \max} \kappa(\rho(\Sig))^2 \sqrt{s_{\Sig} \vee n}, \\
\norm{\smash[t] {\hat{\Sig}}^{-1} - (\Sig^*)^{-1}}_2 &\leq \frac{\tilde{C} \lambda_{\Sig} \sqrt{s_{\Sig} \vee 1}}{\sigma_{*, \min} \phi^2_{\min}(\rho(\Sig))}, \label{step3_thm_eq1} \\
\norm{\smash[t] {\hat{\Sig}}^{-1} - (\Sig^*)^{-1}}_2 &\leq \frac{\tilde{C} \lambda_{\Sig} \sqrt{s_{\Sig} \vee n}}{\sigma_{*, \min} \phi^2_{\min}(\rho(\Sig))} \label{step3_thm_eq2}
\end{align}
for some constant $\tilde{C}$.

\end{theorem}

\begin{proof}
Assume that $\mathcal{E}^*$ holds throughout this proof.

Next, recall that from Theorem~\ref{step2_thm},
\begin{align*}
\max_{i \neq j} \left | \left( \hat{\mathbf{S}}_{\rho, \Sig} - \rho(\Sig) \right)_{ij} \right | \leq \omega := \frac{2 \sum_{k=1}^{K} \tilde{\eta}_k}{1 - \sum_{k=1}^{K} \tilde{\eta}_k} \leq \frac{2 \sum_{k=1}^{K} {\eta}_k}{1 - \sum_{k=1}^{K} {\eta}_k}.
\end{align*}
By defining $C_{p,k} := \frac{2}{\epsilon} \frac{\norm{\rho(\Delt_k)^{-1}}_{1,\text{off}}}{p} + \frac{\norm{\rho(\Delt_k)^{-1}}_{1}}{p}$, it follows that $\eta_k = \left( \nu_{p_k} + \frac{\alpha_k}{1 - \alpha_k} C_{p,k} \right) (1 + o(1))$. Then because $C_{p,k} \asymp 1$ under \ref{a3} and $\nu_{p_k} \rightarrow 0$ as $n, p_k \rightarrow \infty$ under \ref{a1} and \ref{a2},  
\begin{align*}
\frac{2 \sum_{k=1}^{K} {\eta}_k}{1 - \sum_{k=1}^{K} {\eta}_k} \asymp \sum_{k=1}^{K} \nu_{p_k} \rightarrow 0 \quad \text{as } n, p_k \rightarrow \infty.
\end{align*}
Moreover, under \ref{a1}, we have $\omega \sqrt{s_{\Sig} \vee 1} = o(1)$. Thus, we can apply Theorem 4.5 from \citet{zhou2014gemini} to obtain for some constant $\tilde{C}$,
\begin{align*}
\norm{\hat{\Sig}_{\rho} - \rho(\Sig)}_2 &\leq \norm{\hat{\Sig}_{\rho} - \rho(\Sig)}_F \leq \tilde{C} \kappa(\rho(\Sig))^2 \lambda_{\Sig} \sqrt{s_{\Sig} \vee 1} \\
\text{and } \norm{\smash[t] {\hat{\Sig}}^{-1}_{\rho} - \rho(\Sig)^{-1}}_2 &\leq \norm{\smash[t] {\hat{\Sig}}^{-1}_{\rho} - \rho(\Sig)^{-1}}_F \leq \tilde{C} \lambda_{\Sig} \frac{\sqrt{s_{\Sig} \vee 1}}{2 \phi^2_{\min}(\rho(\Sig))}.
\end{align*}

Now since we have a bound on the correlation estimates, the next step is to consider the covariance estimates. Let us define $\W_{\Sig} = \text{diag}(\Sig^*)^{1/2} = \sqrt{\frac{n}{\mathrm{tr}(\Sig)}}\text{diag}(\Sig)^{1/2}$. By Lemma~\ref{step1_lemma}, 
\begin{align*}
| \hat{\W}_{\Sig,ii}^2 - \W_{\Sig,ii}^2 | = | \hat{\W}_{\Sig,ii}^2 - \sigma_{*,ii} | \leq \sigma_{*,ii} \sum_{k=1}^{K} \tilde{\eta}_k \quad \forall \: i. 
\end{align*}

Therefore,
\begin{align}
\norm{\hat{\W}_{\Sig} - \W_{\Sig}}_2 &\leq \sqrt{\sigma_{*,\max}} \left[ \left(\sqrt{1 + \sum_{k=1}^{K} \tilde{\eta}_k} - 1 \right) \vee \left(1 - \sqrt{1 - \sum_{k=1}^{K} \tilde{\eta}_k} \right) \right] \\
&\leq \sqrt{\sigma_{*,\max}} \sum_{k=1}^{K} \tilde{\eta}_k \label{step3_eq1} \\ 
\text{and } \norm{\hat{\W}_{\Sig}^{-1} - \W_{\Sig}^{-1}}_2 &\leq \frac{1}{\sqrt{\sigma_{*,\max}}} \left[ \frac{\sqrt{1 + \sum_{k=1}^{K} \tilde{\eta}_k} - 1}{\sqrt{1 + \sum_{k=1}^{K} \tilde{\eta}_k}} \vee \frac{1 - \sqrt{1 - \sum_{k=1}^{K} \tilde{\eta}_k}}{\sqrt{1-\sum_{k=1}^{K} \tilde{\eta}_k}}  \right] \\
&\leq \frac{1}{\sqrt{\sigma_{*,\max}}} \frac{\sum_{k=1}^{K} \tilde{\eta}_k}{\sqrt{1 - \sum_{k=1}^{K} \tilde{\eta}_k}}. \label{step3_eq2}
\end{align}

Using Proposition 15.2 in \citet{zhou2014supplement}, \eqref{step3_eq1}, and \eqref{step3_eq2}, we obtain
\begin{align*}
\norm{\hat{\Sig} - \Sig^*}_2 &= \norm{\hat{\W}_{\Sig} \hat{\Sig}_{\rho} \hat{\W}_{\Sig} - \W_{\Sig} \rho(\Sig) \W_{\Sig}}_2 \\
&\leq \left( \norm{\hat{\W}_{\Sig} - \W_{\Sig}}_2 + \norm{\W_{\Sig}}_2 \right)^2 \norm{\hat{\Sig}_{\rho} - \rho(\Sig)}_2 \nonumber \\
& \qquad + \norm{\hat{\W}_{\Sig} - \W_{\Sig}}_2 \left( \norm{\hat{\W}_{\Sig} - \W_{\Sig}}_2 + 2 \right) \norm{\rho(\Sig)}_2 \\
&\leq \left( \tilde{C} \kappa(\rho(\Sig))^2 \lambda_{\Sig} \sqrt{s_{\Sig} \vee 1} \right) \sigma_{*,\max} \left( 1 + \sum_{k=1}^{K} \tilde{\eta}_k \right)^2 \nonumber \\
&\qquad + \sigma_{*,\max} \sum_{k=1}^{K} \tilde{\eta}_k \left( \sum_{k=1}^{K} \tilde{\eta}_k + 2 \right) \norm{\rho(\Sig)}_2.
\end{align*}

And because $\lambda_{\Sig}$ was chosen to satisfy $\sum_{k=1}^{K} \tilde{\eta}_k < \lambda_{\Sig} (1 - \sum_{k=1}^{K} \tilde{\eta}_k)/2$ where  $\sum_{k=1}^{K} \tilde{\eta}_k \leq \sum_{k=1}^{K} \eta_k \leq 1/4$, we have
\begin{align*}
\norm{\hat{\Sig} - \Sig^*}_2 &\leq \left( \tilde{C} \kappa(\rho(\Sig))^2 \lambda_{\Sig} \sqrt{s_{\Sig} \vee 1} \right) \sigma_{*,\max} \left( 1 + \sum_{k=1}^{K} \tilde{\eta}_k \right)^2 \nonumber \\
&\qquad + \sigma_{*,\max} \frac{\lambda_{\Sig}}{2} \left( 1 - \sum_{k=1}^{K} \tilde{\eta}_k \right) \left( \sum_{k=1}^{K} \tilde{\eta}_k + 2 \right) \norm{\rho(\Sig)}_2 \\
&\leq 2\tilde{C} \kappa(\rho(\Sig))^2 \sigma_{*,\max} \lambda_{\Sig} \sqrt{s_{\Sig} \vee 1}.
\end{align*}

We can also bound the error on the Frobenius norm similarly. Using Proposition 15.2 in \citet{zhou2014supplement}, \eqref{step3_eq1}, \eqref{step3_eq2}, $\sum_{k=1}^{K} \tilde{\eta}_k < \lambda_{\Sig} (1 - \sum_{k=1}^{K} \tilde{\eta}_k)/2$, and $\sum_{k=1}^{K} \tilde{\eta}_k \leq \sum_{k=1}^{K} \eta_k \leq 1/4$ we see that
\begin{align*}
\norm{\hat{\Sig} - \Sig^*}_F &\leq \left( \norm{\hat{\W}_{\Sig} - \W_{\Sig}}_2 + \norm{\W_{\Sig}}_2 \right)^2 \norm{\hat{\Sig}_{\rho} - \rho(\Sig)}_F \nonumber \\
& \qquad + \norm{\hat{\W}_{\Sig} - \W_{\Sig}}_2 \left( \norm{\hat{\W}_{\Sig} - \W_{\Sig}}_2 + 2 \right) \norm{\rho(\Sig)}_F \\
&\leq \left( \tilde{C} \kappa(\rho(\Sig))^2 \lambda_{\Sig} \sqrt{s_{\Sig} \vee 1} \right) \sigma_{*,\max} \left( 1 + \sum_{k=1}^{K} \tilde{\eta}_k \right)^2 \nonumber \\
&\qquad + \sigma_{*,\max} \sum_{k=1}^{K} \tilde{\eta}_k \left( \sum_{k=1}^{K} \tilde{\eta}_k + 2 \right) \sqrt{n} \norm{\rho(\Sig)}_2 \\
&\leq 2 \tilde{C} \kappa(\rho(\Sig))^2 \sigma_{*,\max} \lambda_{\Sig} \sqrt{s_{\Sig} \vee n}.
\end{align*}

The same logic can be used to prove \eqref{step3_thm_eq1} and \eqref{step3_thm_eq2}, so we omit the details.
\end{proof}

To summarize the convergence results from Theorem~\ref{step3_thm}, if we set 
\begin{align}
\lambda_k  &= \frac{2 \alpha_k}{\epsilon (1 - \alpha_k)} \asymp \sqrt{\frac{\log(n \vee p_k)}{n}} \quad \forall \: k = 1, \dots, K, \label{lambda_k_order} \\
\text{and } \lambda_{\Sig} &= \frac{2 \sum_{k=1}^{K} \tilde{\eta}_k}{\epsilon_1 (1 - \sum_{k=1}^{K} \tilde{\eta}_k)} \asymp \sum_{k=1}^{K} \frac{p_k}{p} \sqrt{\frac{\log(n \vee p_k)}{p_k}}  \label{lambda_Sig_order},
\end{align}
then according to Theorem~\ref{step3_thm}, with probability $1 - \sum_{k=1}^{K} \frac{8}{(n \vee p_k)^2}$,
\begin{align}
\norm{\hat{\Sig} - \Sig^*}_2 &= O \left( \sum_{k=1}^{K}\frac{p_k}{p} \sqrt{\frac{(s_{\Sig} \vee 1) \log(n \vee p_k)}{p_k}} \right), \label{rate2} \\
\norm{\smash[t] {\hat{\Sig}}^{-1} - (\Sig^{-1})^*}_2 &= O \left( \sum_{k=1}^{K}\frac{p_k}{p} \sqrt{\frac{(s_{\Sig} \vee 1) \log(n \vee p_k)}{p_k}} \right), \label{rate2_inv} \\
\norm{\hat{\Sig} - \Sig^*}_F &= O \left( \sum_{k=1}^{K}  \frac{p_k}{p} \sqrt{\frac{(s_{\Sig} \vee n) \log(n \vee p_k)}{p_k}} \right), \label{rateF} \\
\norm{\smash[t] {\hat{\Sig}}^{-1} - (\Sig^{-1})^*}_F &= O \left( \sum_{k=1}^{K}  \frac{p_k}{p} \sqrt{\frac{(s_{\Sig} \vee n) \log(n \vee p_k)}{p_k}} \right) \label{rateF_inv}.
\end{align}

Subspace consistency for the additive $L_1$ correlation estimator, $\hat{\Sig}$, follows as a direct corollary of Theorem~\ref{step3_thm}.

\begin{corollary}\label{subspace_consistency_thm}
Let $\uu_i$ denote the eigenvector of $\Sig^*$ corresponding to the eigenvalue $\phi_i$ and $\hat{\uu}_i$ be the eigenvector of $\hat{\Sig}$ with eigenvalue $\hat{\phi}_i$, where the eigenvalues are sorted in descending order. Under the conditions stated in Theorem~\ref{step3_thm}, we have the following:
\begin{enumerate}[label=(\roman*)]
\item For each $i$,
\begin{align*}
| \hat{\phi}_i - \phi_i | = O\left( \sum_{k=1}^{K} \frac{p_k}{p} \sqrt{\frac{(s_{\Sig} \vee 1) \log (n \vee p_k)}{p_k}} \right).
\end{align*}
\item For each $i$ such that $\min (\phi_{i-1} - \phi_i, \: \phi_i - \phi_{i+1}) \neq 0$,
\begin{align*}
\norm{\hat{\uu}_i - \uu_i}_2 = O\left( \sum_{k=1}^{K} \frac{p_k}{p} \sqrt{\frac{(s_{\Sig} \vee 1) \log (n \vee p_k)}{p_k}} \right).
\end{align*}
\end{enumerate}
Thus, the eigenvectors and eigenvalues of the one-step additive $L_1$ correlation estimator, $\hat{\Sig}$, are consistent.
\end{corollary}

\begin{proof}
$(i)$ By Weyl's theorem \citep{horn2012matrix}, we have for each $i$,
\begin{align*}
| \hat{\phi}_i - \phi_i | \leq \norm{\hat{\Sig} - \Sig^*}_2.
\end{align*}
Thus, it follows from Theorem~\ref{step3_thm} that
\begin{align*}
| \hat{\phi}_i - \phi_i | = O\left( \sum_{k=1}^{K} \frac{p_k}{p} \sqrt{\frac{(s_{\Sig} \vee 1) \log (n \vee p_k)}{p_k}} \right).
\end{align*}

\noindent
$(ii)$ By a variant of the Davis-Kahan sin $\theta$ theorem given in \citet{yu2015daviskahan}, we have for each $i$,
\begin{align*}
\norm{\hat{\uu}_i - \uu_i}_2 \leq \frac{2^{3/2} \norm{\hat{\Sig} - \Sig^*}_2}{\min (\phi_{i-1} - \phi_i, \: \phi_i - \phi_{i+1})},
\end{align*}
assuming that $\min (\phi_{i-1} - \phi_i, \: \phi_i - \phi_{i+1}) \neq 0$ and $\hat{\vv}^T \vv \geq 0$. Note however that if $\hat{\uu}_i^T \uu_i < 0$, we can simply take the negative of $\hat{\uu}_i$ and apply the theorem to $-\hat{\uu}_i$ and $\uu_i$. Thus, for each $i$ such that $\min (\phi_{i-1} - \phi_i, \: \phi_i - \phi_{i+1}) \neq 0$, it follows from Theorem~\ref{step3_thm} that
\begin{align*}
\norm{\hat{\uu}_i - \uu_i}_2 = O\left( \sum_{k=1}^{K} \frac{p_k}{p} \sqrt{\frac{(s_{\Sig} \vee 1) \log (n \vee p_k)}{p_k}} \right).
\end{align*}

Since $\sum_{k=1}^{K} \frac{p_k}{p} \sqrt{\frac{(s_{\Sig} \vee 1) \log (n \vee p_k)}{p_k}}$ converges to $0$ as $n, p_1, \dots, p_K \rightarrow \infty$ under \ref{a1}, this gives us the consistency of the eigenvalues and eigenvectors of $\hat{\Sig}$.
\end{proof}

Note though that the consistency statements above assume that $\sqrt{n} \geq \frac{p}{\sqrt{p_k}} ~ \forall\: k = 1, \dots, K$ (i.e. the ``large $n$'' setting). If instead $\sqrt{n} < \frac{p}{\sqrt{p_k}} ~ \forall\: k = 1, \dots, K$ (i.e. the ``large $p$'' setting), then we can modify Algorithm~\ref{alg:glasso_off_cor} to first initialize an estimate of $\Sig$ assuming $\hat{\Delt}_k = \I$ for each $k$. Call this initial estimate $\hat{\Sig}^1$. Then estimate $\Delt_k$ given $\hat{\Sig}^1$. Call this estimate $\hat{\Delt}_k$. Lastly, obtain the final estimate of $\Sig$ given $\hat{\Delt}_1, \dots, \hat{\Delt}_K$. Denote this final estimate of $\Sig$ by $\hat{\Sig}$. Similar convergence rates can be obtained for the ``large $p$'' setting by using this modified algorithm. Namely, if the penalty parameters $\lambda_k$ and $\lambda_{\Sig}$ are chosen on the order of \eqref{lambda_k_order} and \eqref{lambda_Sig_order}, we can obtain the same rates as \eqref{rate2}-\eqref{rateF_inv}. The only additional assumption required here is a bound on $| \rho(\Sig)_{ij} |$, namely, $| \rho(\Sig)_{ij} | = O(\frac{\sqrt{n p_k}}{p})$ for each $k = 1, \dots, K$ and $i \neq j$. We omit this proof as it is a repetition of previous arguments with slight differences. A more thorough discussion of this scenario is presented in \citet{zhou2014gemini}.

Thus, in either the large $n$ or large $p$ case, we have a variant of the additive $L_1$ correlation Flip-Flop algorithm such that under certain assumptions, the estimate of $\Sig$ converges at a rate of $\displaystyle O \left( \sum_{k=1}^{K} \frac{p_k}{p} \sqrt{\frac{(s_{\Sig} \vee 1) \log (n \vee p_k)}{p_k}} \right)$ in the operator norm. As a result of convergence in the operator norm, we obtain consistency of the eigenvalues and eigenvectors of $\hat{\Sig}$ (see Corollary~\ref{subspace_consistency_thm}). Since eigenvectors of $\hat{\Sig}$ define the estimated iPCA subspace, this in turn implies subspace consistency of the one-step additive $L_1$ correlation estimator.

\section{Selecting Penalty Parameters} \label{sec:s_selecting_penalty_parameters}

Our missing data imputation framework for selecting penalty parameters is given in Algorithm~\ref{alg:select_pp}. In the subsequent sections, we discuss algorithms to impute the missing data in step~\ref{impute_step} of this algorithm. As one option, one could use the multi-cycle expectation-conditional maximization (MCECM) \citep{meng1993ecm} algorithm, which iterates between taking conditional expectations in the E-step and maximizing with respect to one variable at a time in the M-step. We derive the full MCECM algorithm for iPCA in Appendix~\ref{sec:s_full_mcecm}. This method generalizes the MCECM imputation method proposed in  \citet{allen2010transposable}, which only considered the $K=1$ case. However, as \citet{allen2010transposable} pointed out, the full MCECM algorithm is computationally expensive, so in practice, we also advocate using a faster one-step approximation to the MCECM algorithm, which we discuss in Appendix~\ref{sec:s_one_step_ecm}.

\begin{algorithm}\caption{Selecting Penalty Parameters via Missing Imputation Framework}\label{alg:select_pp}
\textbf{Given:} data $\X_1, \dots, \X_K$, space of penalty parameters $\Lambda$, type of penalized iPCA estimator
\begin{algorithmic}[1]
\For{$k = 1, \dots, K$}
\State \multiline{Randomly leave out $5\%$ of the elements in $\X_k$; denote these scattered missing elements by $\X_k^{m}$ }
\EndFor

\For{$\lambda$ in $\Lambda$}
\State \multiline{Impute missing values (preferably by Algorithm~\ref{alg:one_step}); denote these imputed values by $\Xhat_k^m$} \label{impute_step}
\EndFor

\State Select $\lambda$ which minimizes $\displaystyle \sum_{k=1}^{K} \frac{\norm{\Xhat_k^{m} - \X_k^{m}}_F^2}{\norm{\X_k^{m} - \Xbar_k^{m}}_F^2}$ , where $\Xbar_k^{m}$ are the values of the column mean matrix $\Xbar_k$ at the missing indicies.
\end{algorithmic}
\end{algorithm}

Following the notation in \citet{allen2010transposable}, we write $\X^{o} = (\X^o_1, \dots, \X^o_K)$ to denote the totality of observed entries of $\X = (\X_1, \dots, \X_K)$, and $\X^m = (\X^m_1, \dots, \X^m_K)$ to denote the missing entries of $\X$. Also define $\boldsymbol{\Theta} = (\mmu_1, \dots, \mmu_K, \Sig^{-1}, \Delt_1^{-1}, \dots, \Delt_K^{-1})$, and let $\boldsymbol{\Theta}'$ be the current estimates of $\boldsymbol{\Theta}$. Note that the MCECM and its one-step approximation for iPCA (where $K \geq 1$) are generalizations of the imputation algorithms from \citet{allen2010transposable} (where $K = 1$).

\subsection{Multi-Cycle Expectation-Conditional Maximization Algorithm} \label{sec:s_full_mcecm}

For concreteness, we will work with the multiplicative Frobenius penalty. The other penalized methods are very similar. We will proceed to derive the E-steps and M-steps with respect to each variable for the MCECM algorithm. Note here that we will estimate the column means $\mmu_k$ adaptively within the MCECM algorithm, rather than fixing them a priori. We find that this strategy works better in practice.

So first, in order to compute the E-steps, note that the $Q$ function is
\begin{align*}
Q(\boldsymbol{\Theta}; \boldsymbol{\Theta}') &:= \mathbb{E}_{\X^{m} | \X^{o}, \boldsymbol{\Theta}'}[\ell(\boldsymbol{\Theta} | \X)] \\
&= p \log | \Sig^{-1} | + n \sum_{k=1}^{K} \log | \Delt_k^{-1} | \nonumber \\
& \qquad - \mathbb{E}_{\X^{m} | \X^{o}, \boldsymbol{\Theta}'}\left[ \sum_{k = 1}^{K} \mathrm{tr}\left( \Sig^{-1} \left(\X_{k} - \mathbf{1}_{n} \mmu_{k}^{T} \right) \Delt_k^{-1} \left(\X_{k} - \mathbf{1}_{n} \mmu_{k}^{T} \right)^{T} \right) \right ] \nonumber \\ 
& \qquad - \norm{\Sig^{-1}}_F^2 \sum_{k=1}^{K} \rho_k \norm{\Delt_k^{-1}}_F^2
\end{align*}

Thus, for the E-step with respect to $\Sig$, we use linearity of the expectation and trace operators to obtain
\begin{align*}
&\mathbb{E}_{\X^{m} | \X^{o}, \boldsymbol{\Theta}'}\left[ \sum_{k = 1}^{K} \mathrm{tr} \left( \Sig^{-1} \left(\X_{k} - \mathbf{1}_{n} \mmu_{k}^{T} \right) \Delt_k^{-1} \left(\X_{k} - \mathbf{1}_{n} \mmu_{k}^{T} \right)^{T} \right) \right] \\
& \qquad = \mathrm{tr} \left( \Sig^{-1} \sum_{k = 1}^{K} \mathbb{E}_{\X^{m} | \X^{o}, \boldsymbol{\Theta}'}\left[ \left(\X_{k} - \mathbf{1}_{n} \mmu_{k}^{T} \right) \Delt_k^{-1} \left(\X_{k} - \mathbf{1}_{n} \mmu_{k}^{T} \right)^{T} \right] \right).
\end{align*}
Using the notation and proof of Proposition 3 in \citet{allen2010transposable}, the conditional expectation reduces to
\begin{align}
\mathbb{E}_{\X^{m} | \X^{o}, \boldsymbol{\Theta}'}\left[ \left(\X_{k} - \mathbf{1}_{n} \mmu_{k}^{T} \right) \Delt_k^{-1} \left(\X_{k} - \mathbf{1}_{n} \mmu_{k}^{T} \right)^{T} \right] = \sum_{k=1}^{K} \left( \Xhat_k \Delt_k^{-1} \Xhat_k^T + F_k(\Delt_k^{-1}) \right). \label{e_step_sig}
\end{align}

For the E-Step with respect to $\Delt_k$, a similar argument shows that 
\begin{align*}
& \mathbb{E}_{\X^{m} | \X^{o}, \boldsymbol{\Theta}'}\left[ \sum_{k = 1}^{K} \mathrm{tr} \left( \Sig^{-1} \left(\X_{k} - \mathbf{1}_{n} \mmu_{k}^{T} \right) \Delt_k^{-1} \left(\X_{k} - \mathbf{1}_{n} \mmu_{k}^{T} \right)^{T} \right) \right] \\
& \qquad = \sum_{k = 1}^{K} \mathrm{tr} \left( \mathbb{E}_{\X^{m} | \X^{o}, \boldsymbol{\Theta}'}\left[ \left(\X_{k} - \mathbf{1}_{n} \mmu_{k}^{T} \right)^{T} \Sig^{-1} \left(\X_{k} - \mathbf{1}_{n} \mmu_{k}^{T} \right) \right] \Delt_k^{-1} \right),
\end{align*}
and again from \citet{allen2010transposable}, we have that
\begin{align}
\mathbb{E}_{\X^{m} | \X^{o}, \boldsymbol{\Theta}'}\left[ \left(\X_{k} - \mathbf{1}_{n} \mmu_{k}^{T} \right)^{T} \Sig^{-1} \left(\X_{k} - \mathbf{1}_{n} \mmu_{k}^{T} \right) \right] = \Xhat_k^T \Sig^{-1} \Xhat_k + G(\Sig^{-1}). \label{e_step_deltk}
\end{align}

We next plug \eqref{e_step_sig} and \eqref{e_step_deltk} back into the $Q$ function and take partial derivatives to compute the M steps.

For the M-step with respect to $\Sig$, we have that
\begin{align*}
\frac{\partial Q}{\partial \Sig^{-1}} &= p \Sig - \sum_{k=1}^{K} \left( \Xhat_k \Delt_k^{-1} \Xhat_k^T + F_k(\Delt_k^{-1}) \right) - 2 \Sig^{-1} \sum_{k=1}^{K} \lambda_k \norm{\Delt_k^{-1}}_F^2 = 0,
\end{align*}
so given $\Delt_k^{-1}$ from the previous iteration, we can update $\Sig$ via an eigendecomposition of $\sum_{k=1}^{K} \left( \Xhat_k \Delt_k^{-1} \Xhat_k^T + F_k(\Delt_k^{-1}) \right)$. (The form of this update is analogous to the Flip-Flop updates in the multiplicative Frobenius Flip-Flop algorithm.)

Similarly, for the M-step with respect to $\Delt_k$,
\begin{align*}
\frac{\partial Q}{\partial \Delt_k^{-1}} = n \Delt_k - \left( \Xhat_k^T \Sig^{-1} \Xhat_k + G(\Sig^{-1}) \right) - 2 \lambda_k \Delt_k^{-1} \norm{\Sig^{-1}}_F^2 = 0,
\end{align*}
so given $\Sig^{-1}$ from the previous iteration, we can update $\Delt_k$ by an eigendecomposition of $\Xhat_k^T \Sig^{-1} \Xhat_k + G(\Sig^{-1})$.

Putting these E-steps and M-steps together, we provide the full MCECM algorithm to impute missing values in Algorithm~\ref{alg:full_mcecm}.

\begin{algorithm}
\caption{Full MCECM Algorithm for iPCA} \label{alg:full_mcecm}
\begin{algorithmic}[1]

\State Set $\hat{\mmu}_k$ to be the column means of $\X_k^o$  for each $k = 1, \dots, K$ \tikzmark{top1}
\State If $x^{k}_{ij}$ is missing, set $x^{k}_{ij} = \hat{\mmu}_k^{j}$.
\State Initialize $\smash[t] {\hat{\Sig}}^{-1}$ and $\smash[t] {\hat{\Delt}}^{-1}_1 \ldots \smash[t] {\hat{\Delt}}^{-1}_K$ to be symmetric positive definite. \tikzmark{bottom1}

\While{not converged}
\State Compute $\sum_{k=1}^{K} \left[ \Xhat_k \smash[t] {\hat{\Delt}}^{-1}_k \Xhat_k^T + F(\smash[t] {\hat{\Delt}}^{-1}_k) \right ]$ \Comment{E-Step ($\Sig$)} 
\State Update $\hat{\mmu}_k$ to be the column means of $\Xhat_k$ \tikzmark{top3}
\State Take eigendecomposition: $\sum_{k=1}^{K} \left[ \Xhat_k \smash[t] {\hat{\Delt}}^{-1}_k \Xhat_k^T + F(\smash[t] {\hat{\Delt}}^{-1}_k) \right ] = \U \boldsymbol{\Gamma} \U^T$ \tikzmark{right1}
\State Regularize eigenvalues: $\phi_i = \frac{1}{2p} \left( \gamma_i + \sqrt{\gamma_i^2 + 8p \sum_{k=1}^{K} \lambda_k \norm{\smash[t] {\hat{\Delt}}^{-1}_k}_F^2} \right)$
\State Update $\smash[t] {\hat{\Sig}}^{-1} = \U \boldsymbol{\Phi}^{-1} \U^T$ \tikzmark{bottom3}

\For{$k = 1, \dots, K$}
\State Compute $\Xhat_k^T \smash[t] {\hat{\Sig}}^{-1} \Xhat_k + G_k(\smash[t] {\hat{\Sig}}^{-1})$ \Comment{E-Step ($\Delt_k$)} 
\State Update $\hat{\mmu}_k$ to be the column means of $\Xhat_k$ \tikzmark{top5}
\State Take eigendecomposition: $\Xhat_k^T \smash[t] {\hat{\Sig}}^{-1} \Xhat_k + G_k(\smash[t] {\hat{\Sig}}^{-1}) = \V \boldsymbol{\Phi} \V^T$
\State Regularize eigenvalues: $\gamma_i = \frac{1}{2n} \left( \phi_i + \sqrt{\phi_i^2 + 8n \lambda_k \norm{\smash[t] {\hat{\Sig}}^{-1}}_F^2} \right)$
\State Update $\smash[t] {\hat{\Delt}}^{-1}_k = \V \boldsymbol{\Gamma}^{-1} \V^T$ \tikzmark{bottom5}
\EndFor
\EndWhile
\vspace{-16pt}
\end{algorithmic}
\AddNote{top1}{bottom1}{right1}{Initialization}
\AddNote{top3}{bottom3}{right1}{M-Step ($\Sig$)}
\AddNote{top5}{bottom5}{right1}{M-Step ($\Delt_k$)}
\end{algorithm}

\subsection{One-Step Approximation} \label{sec:s_one_step_ecm}

Algorithm~\ref{alg:full_mcecm} is a generalization of the TRCMAimpute algorithm from \citet{allen2010transposable}, and as discussed in \citet{allen2010transposable}, it is computationally expensive to compute $F(\smash[t] {\hat{\Delt}}^{-1}_k)$ and $G(\smash[t] {\hat{\Delt}}^{-1}_k)$. Hence, rather than using the full MCECM algorithm to impute the missing values in step~\ref{impute_step} of Algorithm~\ref{alg:select_pp}, we advocate, as in \citet{allen2010transposable}, using a one-step approximation to the MCECM algorithm. We will empirically show that the one-step approximation is both a good approximation to the full MCECM algorithm and works well in practice.

The idea behind the one-step approximation is that since the first step of the MCECM algorithm typically gives the steepest decrease in the objective function, we will quickly approximate the MCECM algorithm by first obtaining a decent initial imputation and then stopping the algorithm after one M-step and one E-step. We detail the initial imputation step, M-step, and E-step as follows.

For the initial imputation step, we impute missing values assuming $\Sig = \I$. If we assume $\Sig = \I$, then $\X_k \sim \N_{n,p_k}(\mathbf{1}_n \mmu_k^T, \I \otimes \Delt_k)$ for each $k = 1, \dots, K$, or equivalently, $\x_1^k, \dots, \x^k_n \stackrel{iid}{\sim} N(\mmu_k, \Delt_k)$, where $x^k_i$ is the $i^{th}$ row of $\X_k$. Since this reduces to the familiar multivariate case, we can initially impute the missing values in $\X_k$ using any (regularized) multivariate normal imputation method such as RCMimpute from \citet{allen2010transposable} for each $k = 1, \dots, K$.

Given the initial imputation for $\X_1, \dots, \X_K$ from the previous step, we next compute the M-step and update $\mmu_1, \dots, \mmu_K, \Sig, \Delt_1, \dots, \Delt_K$ using the penalized Flip-Flop algorithms derived in Appendix~\ref{sec:s_cov_est}.

In the next and final step, we take an E-step to impute the missing values by $\Xhat^m_k = \mathbb{E}[\X^m_k | \X^o_k, \hat{\mmu}_k, \hat{\Sig}, \hat{\Delt}_k]$ for each $k = 1, \dots, K$. This can be done using the Alternating Conditional Expectations Algorithm from \citet{allen2010transposable}, applied to each $\X_k$ separately. The only difference is that we specialize to the case where the mean matrix is $\M_k = \mathbf{1}_n \mmu_k^T$. We summarize this one-step approximation method for missing data imputation in Algorithm~\ref{alg:one_step}.

\begin{algorithm}\caption{One-Step MCECM Approximation}\label{alg:one_step}
\begin{algorithmic}[1]

\For{$k = 1, \dots, K$} \Comment{Initial Imputation // E-Step}
\State Impute missing values in $\X_k$ assuming $\Sig = \I$; call it $\Xhat_{k,0}$
\Statex \hspace{8mm} \textbullet~Use any regularized multivariate normal imputation method
\EndFor

\State \multiline{Estimate $\mmu_1, \dots, \mmu_K, \Sig, \Delt_1, \dots, \Delt_K$ given $\Xhat_{1,0}, \dots, \Xhat_{K,0}$ via penalized MLE Flip-Flop algorithm\Comment{M-Step}} 
\For{$k = 1, \dots, K$} \Comment{E-Step}
\State \multiline{Set the missing values $\Xhat_k^{m} = \mathbb{E}[\X_k^{m} | \X_k^{o}, \hat{\mmu}_k, \hat{\Sig}, \hat{\Delt}_k]$ using the alternating conditional expectations algorithm as in \citet{allen2010transposable}}
\EndFor

\end{algorithmic}
\end{algorithm}

In Figure~\ref{fig:em_conv}, we compare the numerical convergence of the one-step approximation and the full MCECM algorithm for a small simulation. From this plot, we note two important observations. First, the initialization (assuming $\Sig = \I$) in the one-step approximation algorithm makes a significant difference, compared to initializing missing elements to their respective column means as in the full MCECM algorithm. Second, the first update of the MCECM algorithm results in the largest increase in the log-likelihood function. As discussed in \citet{allen2010transposable}, these two observations motivate the one-step approximation algorithm, which takes advantage of a good initialization and one update step to main sufficient accuracy while reducing the computational workload. Figure~\ref{fig:em_conv} also shows that the likelihood function after a good initialization and one iteration is on par with the full MCECM algorithm after 15 iterations. For a more detailed discussion on computation and timing comparisons between the one-step approximation and the full MCECM algorithm, we refer to \citet{allen2010transposable}. Note that though \citet{allen2010transposable} only treats the $K=1$ case, the results are applicable to the $K>1$ case due to the separability of the log-likelihood function.

\begin{figure}[h]
\centering
\includegraphics[width =  .5\linewidth]{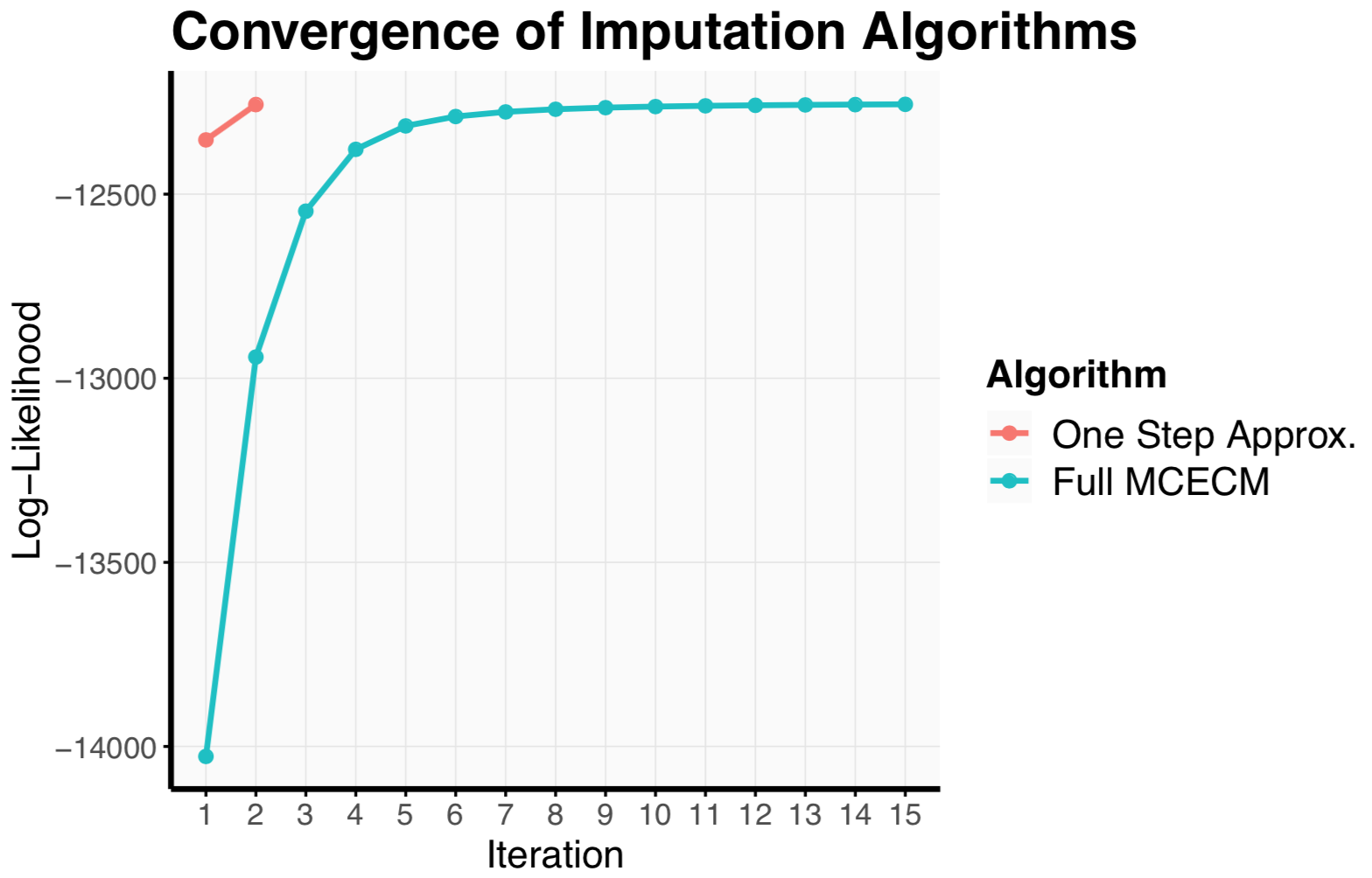}
\caption{\em \footnotesize  We use the same simulation as that in Figure~\ref{fig:cv} and randomly leave out $5\%$ of the entries in each data set. We impute the missing values using the full MCECM and one-step approximation algorithms with the multiplicative Frobenius penalty $(\lambda = (1,1))$, and we plot the log-likelihood value over each iterate. The log-likelihood obtained by the one-step approximation is on par with the log-likelihood after 15 iterations of the MCECM algorithm.}
\label{fig:em_conv}
\end{figure}

Figure~\ref{fig:cv} compares the average imputation errors from the full MCECM and the one-step approximation for a small simulation. In this case, for both the one-step approximation and the full MCECM, $\lambda = (10^{-.5}, 100) \approx (0.32, 100)$ gave the lowest average imputation error, and hence, both imputation methods selected $\lambda = (10^{-.5}, 100)$ for the multiplicative Frobenius iPCA estimator. This further supports the use of the one-step algorithm as an approximation to the full MCECM in practice. Moreover, we verified that the minimum subspace recovery error of $2.00$ was obtained at $\lambda^* = (0.01, 10^{1.5}) \approx (0.01, 31.62)$ and that $\lambda = (10^{-.5}, 100)$ yielded a similar subspace recovery error of $2.08$. This preliminary empirical evidence leads us to believe that the one-step imputation algorithm is indeed a good approximation to the full MCECM algorithm.

\begin{figure}[h]
\centering
\includegraphics[width =  1\linewidth]{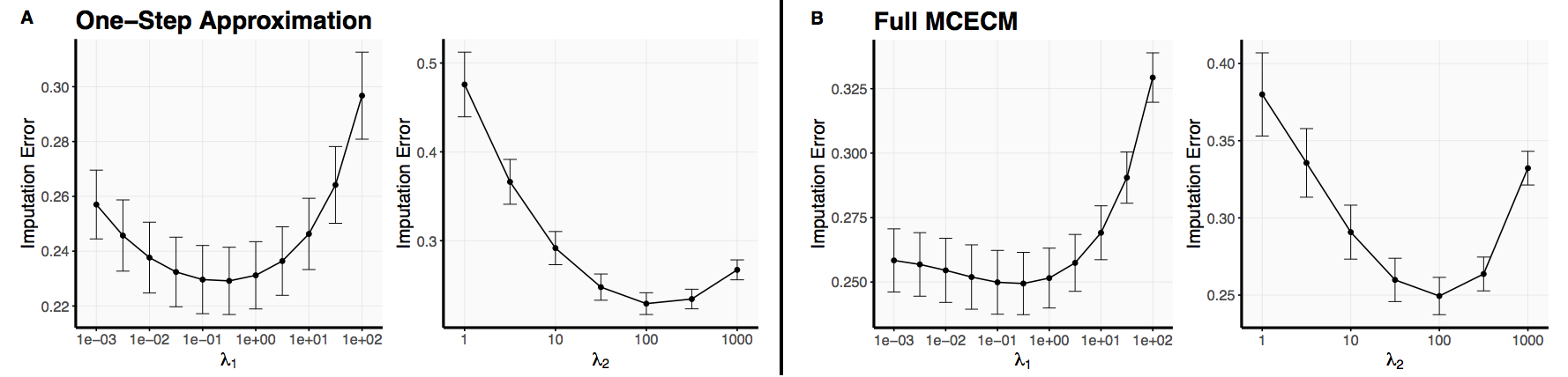}
\caption{\em \footnotesize We simulated two coupled data matrices $\X_1, \X_2$ with $n = 50$, $p_1 = 60, p_2 = 70$ according to the iPCA model \eqref{pop_model}. Here, we took $\Sig$ to be as in the base simulation described in Section~\ref{sec:sims}, $\Delt_1$ to be an autoregressive Toeplitz matrix with the $(ij)^{th}$ entry given by $.9^{|i-j|}$, and $\Delt_2$ to be a block-diagonal matrix with five equal-sized blocks. Then, we randomly removed $5\%$ of the elements in $\X_1$ and $\X_2$ and imputed these missing values using the full MCECM and the one-step approximation algorithms with the multiplicative Frobenius penalty. We plot the average imputation error $\sum_{k=1}^{K} \norm{\Xhat_k^{m} - \X_k^{m}}_F^2/\norm{\X_k^{m} - \Xbar_k^{m}}_F^2$ plus or minus one standard error, taken over 10 trials. In the left graph of each panel (A and B), $\lambda_1$ varies while $\lambda_2$ is fixed at its optimal value (i.e. $\lambda_2 = 100$), and in the right graph of each panel, $\lambda_2$ varies while $\lambda_1$ is fixed at its optimal value (i.e. $\lambda_1 = 10^{-.5}$). The minimum average imputation error is achieved at $\lambda = (10^{-.5}, 100)$ for both imputation algorithms.}
\label{fig:cv}
\end{figure}

\section{Simulations} \label{sec:s_sims}

In order to check that the simulation results in Figure~\ref{fig:sims} are not heavily dependent on our choice of $\Sig$ and $\Delt_1, \Delt_2, \Delt_3$, we ran additional simulations, varying the dimension of the true underlying subspace $\U$ and the number of data sets $K$. These simulation results are shown in Figure~\ref{fig:sims2}.

\begin{figure}
\centering
\includegraphics[width =  .8\linewidth]{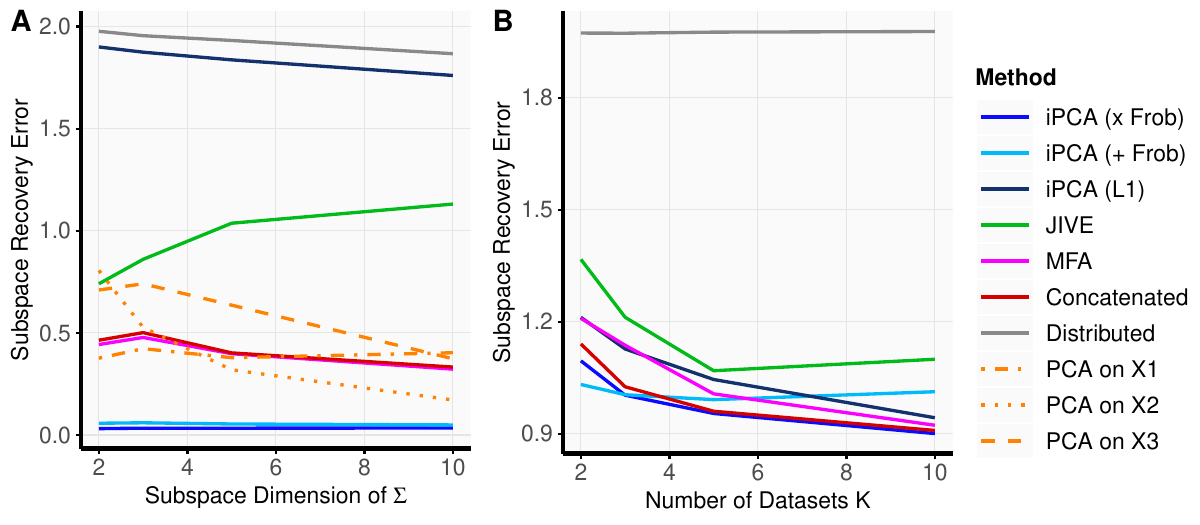}
\caption{\em \footnotesize Additional Simulations: (A) The Frobenius iPCA estimators yield the lower subspace recovery error regardless of the subspace dimension of $\Sig$; (B) As the number of integrated data sets increases, most methods tend to do better, with the multiplicative Frobenius iPCA penalty slightly outperforming the others when $K=10$. Note that we did not run individual PCAs in (B) because the number of data sets is changing.}
\label{fig:sims2}
\end{figure}

For the simulations in Figure~\ref{fig:sims2}A, we took $\Delt_1, \Delt_2, \Delt_3$ to be the same as in the base simulation, and we put $\Sig$ to be of the form $\U \D \U^T$, where $\U$ was a random $n \times n$ orthogonal matrix, and $\D = \text{diag}(d_1, \dots, d_n)$ was simulated by $d_i \sim \text{Unif}(5,75)$ for $i = 1, \dots, D$, and $d_i = 1$ otherwise. Here, $D$ is the dimension of the true underlying subspace of $\Sig$ (e.g. $D$ was taken to be 2 in the base simulation). Then, like in the base simulation, we generated $\X_k$ for each $k = 1, 2, 3$ by $\X_k \sim N(\mathbf{0}, \Sig \otimes \Delt_k)$, or equivalently $\X_k = \Sig^{1/2} \boldsymbol{\Omega}_k \Delt_k^{1/2}$, where $\boldsymbol{\Omega}_k$ is an $n \times p_k$ random matrix with $i.i.d.$ $N(0,1)$ entries.

For the simulations in Figure~\ref{fig:sims2}B, we took $\Sig$ to be the same as in the base simulation, and we generated $\Delt_k$ from a random choice among the covariance types: 
\begin{enumerate}[label=(\roman*)]
\item Autoregressive Toeplitz matrix with the $(ij)^{th}$ entry given by $\rho^{|i-j|}$, where $\rho \sim \text{Unif}(-.9,.9)$;
\item Block diagonal matrix with $B$ blocks of the entries $q_1, \dots, q_B$, where $B \sim \text{Unif}(3,10)$ and $q_i \sim \text{Unif}(0, .9)$;
\item Spiked covariance matrix $\U \D \U^T$, where $\U$ is a random orthogonal matrix, $d_i \sim \text{Unif}(5,75)$ for $i = 1, \dots, D$, $d_i = 1$ for $i = D+1, \dots, p_k$, and $D \sim \text{Unif}(5,50)$;
\item Observed covariance matrix of real miRNA data from TCGA Ovarian Cancer. (Note that this covariance matrix can only be used for one $k \in \{ 1, \dots, K\}$)
\end{enumerate}
The number of features was randomly selected by $p_k \sim \text{Unif}(200, 500)$, and we ensured that $\norm{\Delt_k}$ was larger than that of $\Sig$ so that the individual signal was larger than the joint signal. 

As conveyed in Figure~\ref{fig:sims2}A, the Frobenius iPCA estimators outperformed its competitors regardless of the subspace dimension of $\Sig$. It is also encouraging to see that the strong performance of the Frobenius iPCA estimators is not dependent on the cluster model for $\Sig$, which was used in the simulations in Figure~\ref{fig:sims} but not in Figure~\ref{fig:sims2}A. Note that since the simulated $\Sig$ was dense, the additive $L_1$ penalized iPCA estimator should perform poorly, as it does. We also point out that JIVE tends to do worse as the subspace dimension of $\Sig$ increases because JIVE tends to underestimate the rank of the joint variation matrix when the subspace dimension of $\Sig$ is larger. This is one of the disadvantages of matrix factorizations, as they require the rank of the factorized matrices to be chosen a priori.

In Figure~\ref{fig:sims2}B, we see that most of the methods perform better as the number of integrated data sets increases and that the multiplicative Frobenius iPCA estimator slightly outperforms all the other methods when $K = 10$. We speculate that the additive Frobenius iPCA estimator does worse for large values of $K$ because the grid of possible penalty parameters was too course, and as $K$ increase, so does the number of penalty parameters. Hence, it becomes more difficult to choose appropriate penalty parameters when $K$ is large.

For the Laplacian simulations in Figure~\ref{fig:sims_robust}A, we generated $\X_k$ for each $k = 1, 2, 3$ by $\X_k = \Sig^{1/2} \boldsymbol{\Omega} \Delt_k^{1/2} + \mathbf{E}_k$, where $\Sig$ and $\Delt_k$ were taken as in the base simulation, and $\mathbf{E}_k$ was an $n \times p_k$ random matrix with $i.i.d.$ $\text{Laplace}(0,b)$ entries. 

For the simulations in Figure~\ref{fig:sims_robust}B, we simulated three coupled data sets from an instance of the JIVE model: for each $k = 1, 2, 3$, $\X_k = \mathbf{J}_k + \A_k + \mathbf{E}_k$, where $\mathbf{J} = [\mathbf{J}_1, \mathbf{J}_2, \mathbf{J}_3]$ is the joint variation matrix of rank $r = 5$, $\A_1, \A_2, \A_3$ are the individual variation matrices of rank $r_1 = 10$, $r_2 = 15$, $r_3 = 20$, respectively, and $\mathbf{E}_k$ are error matrices with independent entries from $N(0, \sigma^2)$. Similar to the simulations in \citet{lock2013joint}, we set $\mathbf{J}$ and $\A_k$ by $\mathbf{J} = \U \V_k$ and $\A_k = \U_k \W_k$, where $\U \in \mathbb{R}^{n \times r}$, $\V_k \in \mathbb{R}^{r \times p_k}$, $\U_k \in \mathbb{R}^{n \times r_k}$, and $\W_k \in \mathbb{R}^{r_k \times p_k}$. Here, $\U, \U_k, \V_k, \W_k$ are matrices whose entries are randomly generated $i.i.d.$ from one of the following distributions: $N(0,1)$, $\text{Unif}(0,1)$, $\text{Exp}(1)$, and the discrete random variable $\{-2,-1,0,1,2\}$ with uniform probabilities. Note that under this JIVE model, the true joint covariance matrix is given by $\Sig = \mathbf{J} \mathbf{J}^T$.

In addition to the robustness simulations in Figure~\ref{fig:sims_robust}, we also ran simulations under the CMF model and simulations with uncommon row covariance matrices. When simulating from the CMF model, we generated 3 coupled data matrices with $n = 150$ and $p_1 = 300$, $p_2 = 500$, $p_3 = 400$ via the model $X_k = \U \V_k^T + \mathbf{E}_k$. Here, $\U \in \mathbb{R}^{n \times 2}$ was taken to be a random two-dimensional subspace from a cluster model with three clusters (as shown in Figure~\ref{fig:illustrative}), each $\V_k \in \mathbb{R}^{p_k \times 2}$ was taken to be the two top eigenvectors from the base simulation's $\Delt_k$ (e.g., $\V_1$ was taken to be the top two eigenvectors of the autoregressive Toeplitz matrix with entry $(i,j)$ given by $.9^{|i-j|}$), and $\mathbf{E}_k$ is a random matrix with $i.i.d.$ $N(0, \sigma^2)$ entries. We summarize the simulation results from the CMF model with increasing levels of noise $\sigma$ in Figure~\ref{fig:cmf}. From this figure, we see that when the additive noise is small, CMF (and concatenated PCA) yield slightly lower subspace recovery errors than the multiplicative Frobenius iPCA estimator, but as $\sigma$ increases, this improvement over the multiplicative Frobenius iPCA estimator diminishes.

\begin{figure}
\centering
\includegraphics[width =  .5\linewidth]{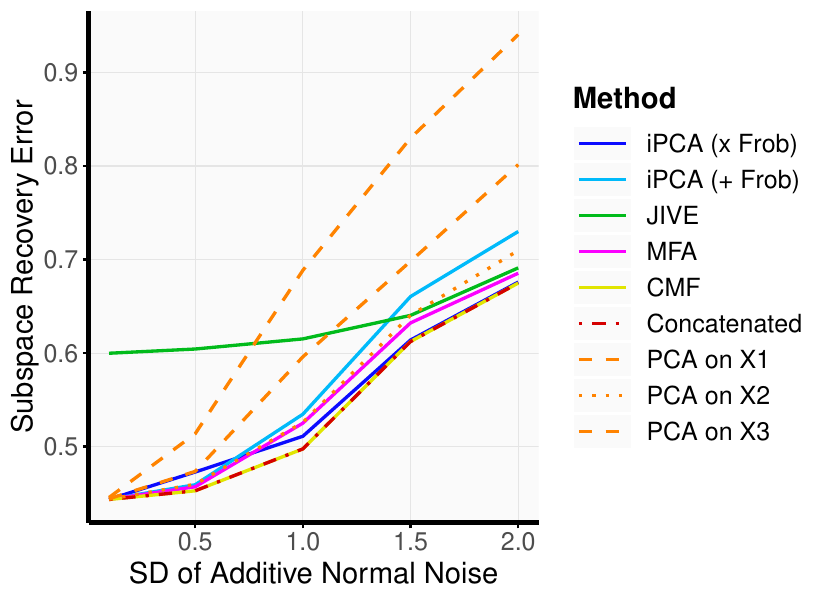}
\caption{\em \footnotesize CMF Model Simulations: As the noise level $\sigma$ of the CMF model increases, the multiplicative Frobenius iPCA estimator performs on par with CMF (and concatenated PCA).}
\label{fig:cmf}
\end{figure}

For the simulations with uncommon row covariance matrices, we modified the base simulation so that two out of the three data sets arose from the model $\X_k \sim N_{n, p_k}(\mathbf{0}, \Sig \otimes \Delt_k)$ while the final data set arose from the model $\X_k \sim N_{n, p_k}(\mathbf{0}, \tilde{\Sig} \otimes \Delt_k)$. Here, $\Sig$ and $\Delt_1, \Delt_2, \Delt_3$ are as in the base simulation while $\tilde{\Sig}$ is a rank-2 spiked covariance with eigenvectors generated from $i.i.d.$ random normal entries. We summarize the simulation results from this uncommon row covariance model in Figure~\ref{fig:uncommon_sig}. In Figure~\ref{fig:uncommon_sig}A, we see that regardless of which $\X_k$ is generated from the uncommon $\tilde{\Sig}$, the Frobenius iPCA estimators are relatively robust to the model misspecification. Figure~\ref{fig:uncommon_sig}B shows the difference in subspace recovery error when applying the methods to all three data sets (i.e. mixture) versus applying them to the two data sets generated by the common $\Sig$ (i.e. oracle). Here, we see that while the Frobenius iPCA estimators perform better with oracle knowledge of the two data sets generated by the common $\Sig$, the decline in performance is relatively small when adding the outlying data set, again demonstrating the robustness of the Frobenius iPCA estimators.

\begin{figure}
\centering
\includegraphics[width =  1\linewidth]{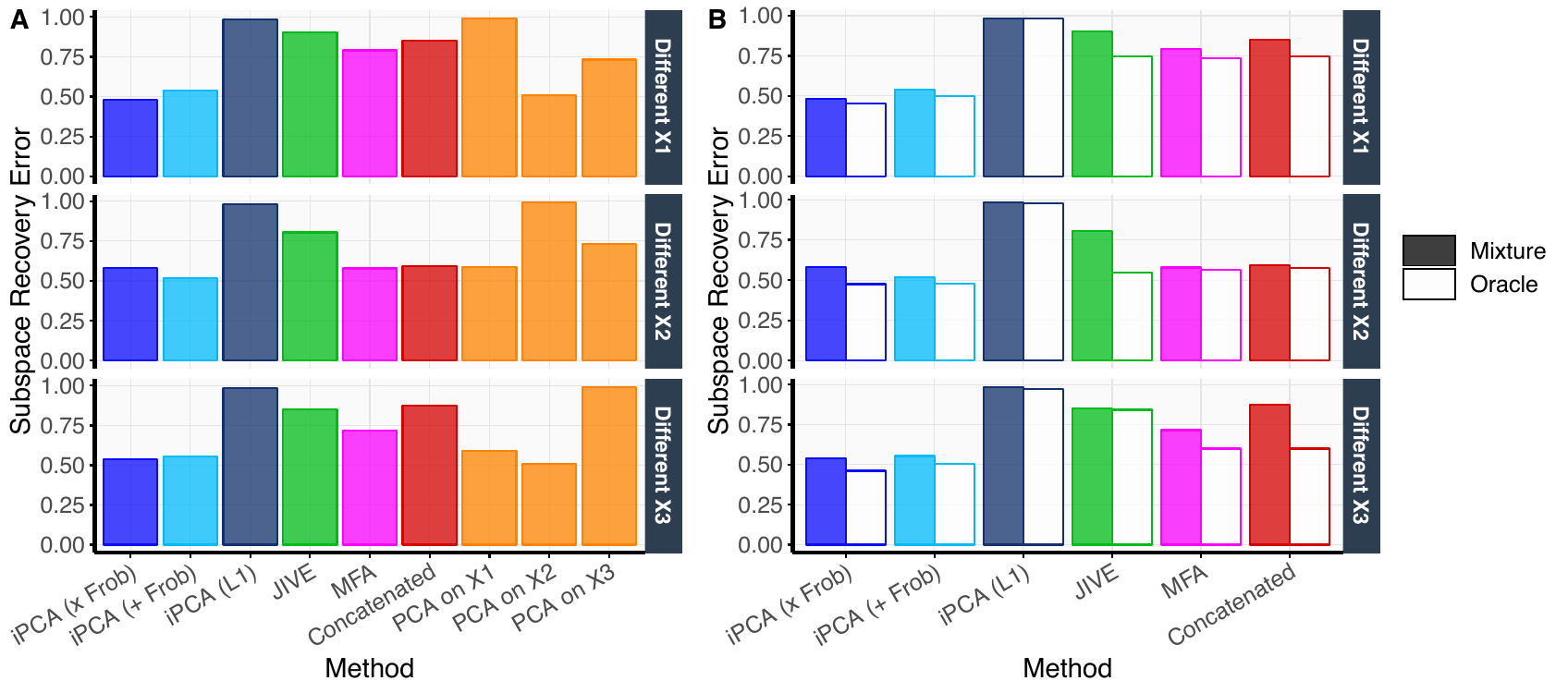}
\caption{\em \footnotesize Uncommon Row Covariance Simulations: (A) We plot the subspace recovery errors for various method under the uncommon row covariance model. In the top panel, $\X_1$ was generated from the uncommon $\tilde{\Sig}$. In the middle panel, $\X_2$ was generated differently, and in the bottom panel, $\X_3$ was different. (B) We compare the subspace recovery errors from various methods when applied to the mixture of all three data sets (i.e. mixture) versus when applied to the two data sets which were generated by the common $\Sig$ (i.e. oracle).}
\label{fig:uncommon_sig}
\end{figure}

In the last set of simulations, we empirically study the iPCA estimators under the sparse setting. For these sparse simulations, we generated two data sets with $n = 50$ and $p_1 = 50$, $p_2 = 100$ according to the Kronecker product model $\X_k \sim N_{n, p_k}(\mathbf{0}, \Sig \otimes \Delt_k)$. Because we are interested in uncovering the low-rank sparse structure when performing dimension reduction, we generate $\Sig$ as follows: Let $\mathbf{U}$ denote the eigenvectors of a block diagonal matrix with $B$ equally-sized blocks, and put $\mathbf{D} = \text{diag}(25, 12.5, 1, \dots, 1) \in \mathbb{R}^{n \times n}$. Then simulate $\Sig$ = $\mathbf{U} \mathbf{D} \mathbf{U}^T$ so that $\Sig$ is a low-rank block diagonal matrix with $B$ blocks. To generate $\Delt_1$, we use the \texttt{huge} package \citep{hugeR} in \texttt{R} to obtain the sparse covariance matrix associated with multivariate normal data being generated from a sparse banded graph (with bandwidth = 4). Lastly, we took $\Delt_2$ to be the block diagonal matrix (with 5 equally-sized blocks), created by taking the observed covariance matrix of miRNA data from TCGA ovarian cancer \citep{cancer2011integrated} and zeroing out the entries off of the block diagonal. The results from this sparse simulation study as we increase the number of blocks in $\Sig$ (and hence increase the sparsity in $\Sig$) are shown in Figure~\ref{fig:sparse}. From this study, we see that when the amount of sparsity is relatively low, the multiplicative Frobenius iPCA estimator and the additive $L_1$ iPCA covariance estimator perform the best while the additive $L_1$ iPCA correlation estimator performs poorly. We believe that in this low sparsity scenario, the additive $L_1$ iPCA correlation estimator struggles with choosing the appropriate penalty parameters and thus does not perform as well. However, as the sparsity level increases, both the additive $L_1$ iPCA covariance and correlation estimators outperform the dense Frobenius iPCA estimators.

\begin{figure}
\centering
\includegraphics[width =  .5\linewidth]{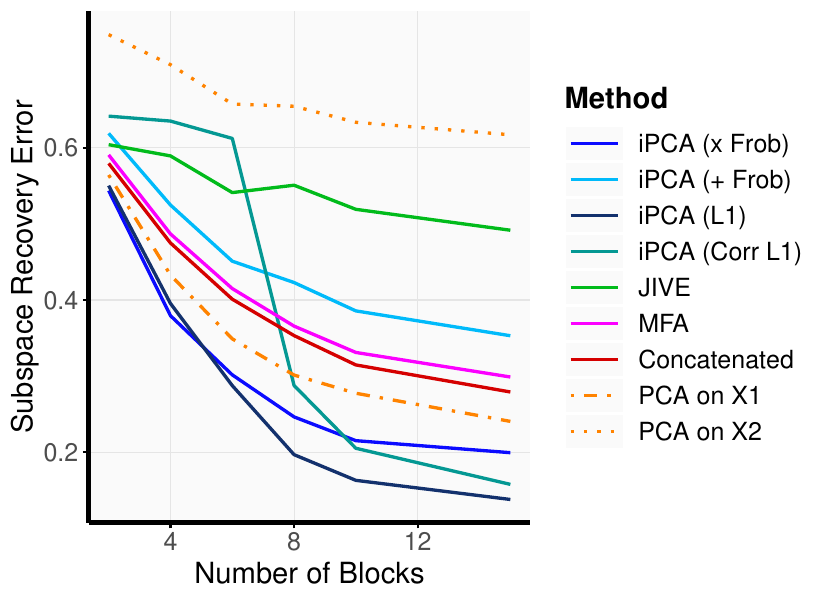}
\caption{\em \footnotesize Sparse Simulations: As we increase the number of blocks in $\Sig$ (and hence also the sparsity of $\Sig$), the sparse iPCA estimators improve over the Frobenius iPCA estimators. However, when the sparsity level is relatively low, the multiplicative Frobenius iPCA estimator performs similarly to the additive $L_1$ iPCA covariance estimator.}
\label{fig:sparse}
\end{figure}

We finally note that other metrics such as canonical angles, which also quantify the distance between subspaces, these metrics behave similarly to the subspace recovery error, so we omitted the results for brevity. Other common metrics such as $| \norm{\hat{\Sig}}_2 - \norm{\Sig}_2 |$ or $\norm{\hat{\Sig} - \Sig}_F^2$ are not appropriate for our study because we are interested in the distance between subspaces of eigenvectors, and eigenvectors are scale-invariant while these metric are not.

\section{Case Study: Integrative Genomics of Alzheimer's Disease}\label{sec:s_rosmap}

The ROSMAP data originally contained $309$ miRNAs, $41,809$ genes, and $420,132$ CpG sites, so we aggressively preprocessed the number of features in the RNASeq and methylation data sets to manageable sizes. First, we transformed the methylation data to m-values and log-transformed the RNASeq counts, as is common in most analyses for these data types. Then, we removed batch (experimental) effects from both data sets via ComBat \citep{johnson2007combat}. We next filtered the features by taking those with the highest variance (top 20,000 genes for RNASeq and top 50,000 CpG sites for methylation). Then, we performed univariate filtering and kept the features with the highest association to clinician's diagnosis. This left us with $p_1 = 309$, $p_2 = 900$, $p_3 = 1250$ in the miRNA, RNASeq, and methylation data sets, respectively. 

R code can be found at \url{https://github.com/DataSlingers/iPCA}.

\vskip 0.2in
\bibliography{ipca}

\begin{thebibliography}{42}
\providecommand{\natexlab}[1]{#1}
\providecommand{\url}[1]{\texttt{#1}}
\expandafter\ifx\csname urlstyle\endcsname\relax
  \providecommand{\doi}[1]{doi: #1}\else
  \providecommand{\doi}{doi: \begingroup \urlstyle{rm}\Url}\fi

\bibitem[Abdi et~al.(2013)Abdi, Williams, and Valentin]{abdi2013mfa}
H.~Abdi, L.~J. Williams, and D.~Valentin.
\newblock Multiple factor analysis: principal component analysis for multitable
  and multiblock data sets.
\newblock \emph{Wiley Interdisciplinary Reviews: Computational Statistics},
  5\penalty0 (2):\penalty0 149--179, 2013.

\bibitem[Acar et~al.(2014)Acar, Papalexakis, G{\"u}rdeniz, Rasmussen, Lawaetz,
  Nilsson, and Bro]{acar2014coupled}
E.~Acar, E.~E. Papalexakis, G.~G{\"u}rdeniz, M.~A. Rasmussen, A.~J. Lawaetz,
  M.~Nilsson, and R.~Bro.
\newblock Structure-revealing data fusion.
\newblock \emph{BMC Bioinformatics}, 15\penalty0 (1):\penalty0 239, Jul 2014.
\newblock ISSN 1471-2105.
\newblock \doi{10.1186/1471-2105-15-239}.
\newblock URL \url{https://doi.org/10.1186/1471-2105-15-239}.

\bibitem[Allen and Tibshirani(2010)]{allen2010transposable}
G.~I. Allen and R.~Tibshirani.
\newblock Transposable regularized covariance models with an application to
  missing data imputation.
\newblock \emph{The Annals of Applied Statistics}, 4\penalty0 (2):\penalty0
  764--790, 2010.

\bibitem[Alter et~al.(2003)Alter, Brown, and Botstein]{alter2003generalized}
O.~Alter, P.~O. Brown, and D.~Botstein.
\newblock Generalized singular value decomposition for comparative analysis of
  genome-scale expression data sets of two different organisms.
\newblock \emph{Proceedings of the National Academy of Sciences}, 100\penalty0
  (6):\penalty0 3351--3356, 2003.

\bibitem[Benidis et~al.(2016)Benidis, Sun, Babu, and
  Palomar]{benidis2016sparsepca}
K.~Benidis, Y.~Sun, P.~Babu, and D.~P. Palomar.
\newblock Orthogonal sparse pca and covariance estimation via procrustes
  reformulation.
\newblock \emph{IEEE Transactions on Signal Processing}, 64\penalty0
  (23):\penalty0 6211--6226, 2016.
\newblock URL \url{http://www.danielppalomar.com/publications.html}.

\bibitem[Carrette et~al.(2003)Carrette, Demalte, Scherl, Yalkinoglu, Corthals,
  Burkhard, Hochstrasser, and Sanchez]{carrette2003panel}
O.~Carrette, I.~Demalte, A.~Scherl, O.~Yalkinoglu, G.~Corthals, P.~Burkhard,
  D.~F. Hochstrasser, and J.C. Sanchez.
\newblock A panel of cerebrospinal fluid potential biomarkers for the diagnosis
  of alzheimer's disease.
\newblock \emph{Proteomics}, 3\penalty0 (8):\penalty0 1486--1494, 2003.

\bibitem[Dawid(1981)]{dawid1981some}
A.P. Dawid.
\newblock Some matrix-variate distribution theory: notational considerations
  and a bayesian application.
\newblock \emph{Biometrika}, 68\penalty0 (1):\penalty0 265--274, 1981.

\bibitem[Dutilleul(1999)]{dutilleul1999mle}
P.~Dutilleul.
\newblock The mle algorithm for the matrix normal distribution.
\newblock \emph{Journal of Statistical Computation and Simulation}, 64\penalty0
  (2):\penalty0 105--123, 1999.

\bibitem[Escofier and Pages(1994)]{escofier1994mfa}
B.~Escofier and J.~Pages.
\newblock Multiple factor analysis.
\newblock \emph{Computational Statistics \& Data Analysis}, 18\penalty0
  (1):\penalty0 121--140, 1994.

\bibitem[Espuny-Camacho et~al.(2017)Espuny-Camacho, Arranz, Fiers, Snellinx,
  Ando, Munck, Bonnefont, Lambot, Corthout, Omodho,
  et~al.]{espuny2017hallmarks}
I.~Espuny-Camacho, A.~M. Arranz, M.~Fiers, A.~Snellinx, K.~Ando, S.~Munck,
  J.~Bonnefont, L.~Lambot, N.~Corthout, L.~Omodho, et~al.
\newblock Hallmarks of alzheimer's disease in stem-cell-derived human neurons
  transplanted into mouse brain.
\newblock \emph{Neuron}, 93\penalty0 (5):\penalty0 1066--1081, 2017.

\bibitem[{Fan} et~al.(2017){Fan}, {Wang}, {Wang}, and {Zhu}]{fan2017pca}
J.~{Fan}, D.~{Wang}, K.~{Wang}, and Z.~{Zhu}.
\newblock Distributed estimation of principal eigenspaces.
\newblock \emph{ArXiv e-prints}, February 2017.

\bibitem[Ghil and Malanotte-Rizzoli(1991)]{ghil1991meteorology}
M.~Ghil and P.~Malanotte-Rizzoli.
\newblock Data assimilation in meteorology and oceanography.
\newblock In \emph{Advances in Geophysics}, volume~33, pages 141--266.
  Elsevier, 1991.

\bibitem[Greenewald and Hero(2015)]{greenewald2015robust}
K.~Greenewald and A.~O. Hero.
\newblock Robust kronecker product pca for spatio-temporal covariance
  estimation.
\newblock \emph{IEEE Transactions on Signal Processing}, 63\penalty0
  (23):\penalty0 6368--6378, 2015.

\bibitem[Gupta and Nagar(1999)]{gupta1999matrix}
A.~K. Gupta and D.~K. Nagar.
\newblock \emph{Matrix variate distributions}.
\newblock CRC Press, 1999.

\bibitem[Han et~al.(2014)Han, Liang, Baxter, Yin, Tang, Beach, Caselli, Reiman,
  and Shi]{han2014pituitary}
P.~Han, W.~Liang, L.~C. Baxter, J.~Yin, Z.~Tang, T.~G. Beach, R.~J. Caselli,
  E.~M. Reiman, and J.~Shi.
\newblock Pituitary adenylate cyclase--activating polypeptide is reduced in
  alzheimer disease.
\newblock \emph{Neurology}, 82\penalty0 (19):\penalty0 1724--1728, 2014.

\bibitem[Hastie et~al.(2015)Hastie, Mazumder, Lee, and Zadeh]{hastie2015matrix}
T.~Hastie, R.~Mazumder, J.~D. Lee, and R.~Zadeh.
\newblock Matrix completion and low-rank svd via fast alternating least
  squares.
\newblock \emph{The Journal of Machine Learning Research}, 16\penalty0
  (1):\penalty0 3367--3402, 2015.

\bibitem[Horn and Johnson(2012)]{horn2012matrix}
R.~A. Horn and C.~R. Johnson.
\newblock \emph{Matrix analysis}.
\newblock Cambridge university press, 2012.

\bibitem[Hsieh et~al.(2011)Hsieh, Dhillon, Ravikumar, and
  Sustik]{hsieh2011quic}
C.~J. Hsieh, I.~S. Dhillon, P.~K. Ravikumar, and M.~A. Sustik.
\newblock Sparse inverse covariance matrix estimation using quadratic
  approximation.
\newblock In \emph{Advances in Neural Information Processing Systems}, pages
  2330--2338, 2011.

\bibitem[Huang et~al.(2017)Huang, Chaudhary, and Garmire]{huang2017more}
Sijia Huang, Kumardeep Chaudhary, and Lana~X Garmire.
\newblock More is better: recent progress in multi-omics data integration
  methods.
\newblock \emph{Frontiers in genetics}, 8:\penalty0 84, 2017.

\bibitem[Jiang et~al.(2019)Jiang, Fei, Liu, Roeder, Lafferty, Wasserman, Li,
  and Zhao]{hugeR}
Haoming Jiang, Xinyu Fei, Han Liu, Kathryn Roeder, John Lafferty, Larry
  Wasserman, Xingguo Li, and Tuo Zhao.
\newblock \emph{huge: High-Dimensional Undirected Graph Estimation}, 2019.
\newblock URL \url{https://CRAN.R-project.org/package=huge}.
\newblock R package version 1.3.2.

\bibitem[Johnson et~al.(2007)Johnson, Li, and Rabinovic]{johnson2007combat}
W.~E. Johnson, C.~Li, and A.~Rabinovic.
\newblock Adjusting batch effects in microarray expression data using empirical
  bayes methods.
\newblock \emph{Biostatistics}, 8\penalty0 (1):\penalty0 118--127, 2007.

\bibitem[Ledoit and Wolf(2004)]{ledoit2004well}
Olivier Ledoit and Michael Wolf.
\newblock A well-conditioned estimator for large-dimensional covariance
  matrices.
\newblock \emph{Journal of multivariate analysis}, 88\penalty0 (2):\penalty0
  365--411, 2004.

\bibitem[Li et~al.(2017)Li, Chen, Gao, Pan, Zheng, Zhang, Xu, Bu, and
  Zheng]{li2017synaptic}
Y.~Li, Z.~Chen, Y.~Gao, G.~Pan, H.~Zheng, Y.~Zhang, H.~Xu, G.~Bu, and H.~Zheng.
\newblock Synaptic adhesion molecule pcdh-$\gamma$c5 mediates synaptic
  dysfunction in alzheimer's disease.
\newblock \emph{Journal of Neuroscience}, pages 1051--17, 2017.

\bibitem[Lock et~al.(2013)Lock, Hoadley, Marron, and Nobel]{lock2013joint}
E.~F. Lock, K.~A. Hoadley, J.~S. Marron, and A.~B. Nobel.
\newblock Joint and individual variation explained (jive) for integrated
  analysis of multiple data types.
\newblock \emph{The Annals of Applied Statistics}, 7\penalty0 (1):\penalty0
  523, 2013.

\bibitem[Meng and Rubin(1993)]{meng1993ecm}
X.~L. Meng and D.~B. Rubin.
\newblock Maximum likelihood estimation via the ecm algorithm: A general
  framework.
\newblock \emph{Biometrika}, 80\penalty0 (2):\penalty0 267--278, 1993.
\newblock ISSN 00063444.
\newblock URL \url{http://www.jstor.org/stable/2337198}.

\bibitem[Mostafavi et~al.(2018)Mostafavi, Gaiteri, Sullivan, White, Tasaki, Xu,
  Taga, Klein, Patrick, Komashko, et~al.]{mostafavi2018rosmap}
S.~Mostafavi, C.~Gaiteri, S.~E. Sullivan, C.~C. White, S.~Tasaki, J.~Xu,
  M.~Taga, H.~U. Klein, E.~Patrick, V.~Komashko, et~al.
\newblock A molecular network of the aging human brain provides insights into
  the pathology and cognitive decline of alzheimer's disease.
\newblock \emph{Nature Neuroscience}, 21\penalty0 (6):\penalty0 811--819, 2018.
\newblock ISSN 1546-1726.
\newblock URL \url{https://doi.org/10.1038/s41593-018-0154-9}.

\bibitem[Ponnapalli et~al.(2011)Ponnapalli, Saunders, Van~Loan, and
  Alter]{ponnapalli2011higher}
S.~P. Ponnapalli, M.~A. Saunders, C.~F. Van~Loan, and O.~Alter.
\newblock A higher-order generalized singular value decomposition for
  comparison of global mrna expression from multiple organisms.
\newblock \emph{PloS one}, 6\penalty0 (12):\penalty0 e28072, 2011.

\bibitem[Rapcs{\'a}k(1991)]{rapcsak1991gconvex}
T.~Rapcs{\'a}k.
\newblock Geodesic convexity in nonlinear optimization.
\newblock \emph{Journal of Optimization Theory and Applications}, 69\penalty0
  (1):\penalty0 169--183, Apr 1991.
\newblock ISSN 1573-2878.
\newblock \doi{10.1007/BF00940467}.
\newblock URL \url{https://doi.org/10.1007/BF00940467}.

\bibitem[Rothman et~al.(2008)Rothman, Bickel, Levina, and
  Zhu]{rothman2008sparse}
A.~J. Rothman, P.~J. Bickel, E.~Levina, and J.~Zhu.
\newblock Sparse permutation invariant covariance estimation.
\newblock \emph{Electronic Journal of Statistics}, 2:\penalty0 494--515, 2008.
\newblock \doi{10.1214/08-EJS176}.

\bibitem[Shivappa et~al.(2010)Shivappa, Trivedi, and
  Rao]{shivappa2010audiovisual}
S.~T. Shivappa, M.~M. Trivedi, and B.~D. Rao.
\newblock Audiovisual information fusion in human-computer interfaces and
  intelligent environments: A survey.
\newblock \emph{Proceedings of the IEEE}, 98\penalty0 (10):\penalty0
  1692--1715, 2010.

\bibitem[Singh and Gordon(2008)]{singh2008collective}
A.~P. Singh and G.~J. Gordon.
\newblock Relational learning via collective matrix factorization.
\newblock In \emph{Proceedings of the 14th ACM SIGKDD International Conference
  on Knowledge Discovery and Data Mining}, pages 650--658. ACM, 2008.

\bibitem[{The Cancer Genome Atlas Research
  Network}(2011)]{cancer2011integrated}
{The Cancer Genome Atlas Research Network}.
\newblock Integrated genomic analyses of ovarian carcinoma.
\newblock \emph{Nature}, 474\penalty0 (7353):\penalty0 609, 2011.

\bibitem[Tseng(2001)]{tseng2001convergence}
P.~Tseng.
\newblock Convergence of a block coordinate descent method for
  nondifferentiable minimization.
\newblock \emph{Journal of Optimization Theory and Applications}, 109\penalty0
  (3):\penalty0 475--494, 2001.

\bibitem[Tsiligkaridis et~al.(2013)Tsiligkaridis, Hero, and
  Zhou]{tsiligkaridis2013convergence}
T.~Tsiligkaridis, A.~O. Hero, and S.~Zhou.
\newblock On convergence of kronecker graphical lasso algorithms.
\newblock \emph{IEEE Transactions on Signal Processing}, 61:\penalty0
  1743--1755, 2013.

\bibitem[{Vishnoi}(2018)]{vishnoi2018gconvex}
N.~K. {Vishnoi}.
\newblock {Geodesic Convex Optimization: Differentiation on Manifolds,
  Geodesics, and Convexity}.
\newblock \emph{ArXiv e-prints}, June 2018.

\bibitem[Westerhuis et~al.(1998)Westerhuis, Kourti, and
  MacGregor]{westerhuis1998multiblock}
J.~A. Westerhuis, T.~Kourti, and J.~F. MacGregor.
\newblock Analysis of multiblock and hierarchical pca and pls models.
\newblock \emph{Journal of Chemometrics}, 12\penalty0 (5):\penalty0 301--321,
  1998.

\bibitem[Wiesel(2012)]{wiesel2012geodesic}
A.~Wiesel.
\newblock Geodesic convexity and covariance estimation.
\newblock \emph{IEEE Transactions on Signal Processing}, 60\penalty0
  (12):\penalty0 6182, 2012.

\bibitem[Yin and Li(2012)]{yin2012model}
J.~Yin and H.~Li.
\newblock Model selection and estimation in the matrix normal graphical model.
\newblock \emph{Journal of Multivariate Analysis}, 107:\penalty0 119--140,
  2012.

\bibitem[Yu et~al.(2015)Yu, Wang, and Samworth]{yu2015daviskahan}
Y.~Yu, T.~Wang, and R.~J. Samworth.
\newblock A useful variant of the davis-kahan theorem for statisticians.
\newblock \emph{Biometrika}, 102\penalty0 (2):\penalty0 315--323, 2015.
\newblock \doi{10.1093/biomet/asv008}.
\newblock URL \url{http://dx.doi.org/10.1093/biomet/asv008}.

\bibitem[Zhang and Sra(2016)]{zhang2016first}
H.~Zhang and S.~Sra.
\newblock First-order methods for geodesically convex optimization.
\newblock In \emph{Conference on Learning Theory}, pages 1617--1638, 2016.

\bibitem[Zhou(2014{\natexlab{a}})]{zhou2014gemini}
S.~Zhou.
\newblock Gemini: Graph estimation with matrix variate normal instances.
\newblock \emph{The Annals of Statistics}, 42\penalty0 (2):\penalty0 532--562,
  2014{\natexlab{a}}.

\bibitem[Zhou(2014{\natexlab{b}})]{zhou2014supplement}
S.~Zhou.
\newblock Supplement to ``gemini: Graph estimation with matrix variate normal
  instances''.
\newblock 2014{\natexlab{b}}.

\end{thebibliography}

\end{document}